\documentclass[deca,nonblindrev]{style/informs3}

 \OneAndAHalfSpacedXI 

\usepackage{thmtools, thm-restate}
\usepackage{mathtools}
\usepackage{amsfonts,amsmath}
\usepackage{dsfont}
\usepackage{tikz}
\usetikzlibrary{calc,positioning, math}
\usepackage{pgfplots}
\usepackage{soul}
\pgfplotsset{compat=1.16}
\usepackage{pstricks}
\usepackage{stmaryrd}

\usepackage{tabularx}
\newcolumntype{Y}{>{\hsize=.5\hsize}X}
\newcolumntype{Z}{>{\hsize=1.25\hsize}X}
\usepackage{float}
\usepackage{multirow}
\usepackage{multicol}
\usepackage[colorlinks=true,breaklinks=true,bookmarks=true,urlcolor=blue, citecolor=blue,linkcolor=blue,bookmarksopen=false,draft=false]{hyperref}
\usepackage[ruled,noend,linesnumbered]{algorithm2e}

\tikzset{
  pics/hidingloc/.style args = {#1,#2,#3,#4,#5}{
     code = {
        \coordinate (A) at (0,0);
        \coordinate (B) at (1,0.5);
        \draw[fblue,fill=fblue] (A) rectangle ($(A|-B)!#4!(B)$);
        \draw[thick] (A) rectangle (B);
        \coordinate (C) at (0,1);
        \coordinate (D) at (1,#3+1);
        \draw[fred,fill=fred] (C) rectangle ($(C-|D)!#5/#3!(D)$);
        \draw[thick] (C) rectangle (D);
        \foreach \y in {2,...,\the\numexpr#3} {
            \draw[thick] (0,\y) -- (1,\y);
        }
        \node[] (detection_rate) at (0.5, #3+1.5) {#2};
        \node[] (location) at (0.5, -0.5) {#1};
     }
  }
}

\newcommand{\A}{\mathcal{A}}
\newcommand{\I}{\mathcal{I}}
\newcommand{\J}{\mathcal{J}}
\newcommand{\K}{\mathcal{K}}
\newcommand{\R}{\mathbb{R}}
\newcommand{\Z}{\mathbb{Z}}

\newcommand{\Seek}{\textnormal{S}\xspace}
\newcommand{\Hide}{\textnormal{H}\xspace}

\DeclarePairedDelimiter\ceil{\lceil}{\rceil}
\DeclarePairedDelimiter\floor{\lfloor}{\rfloor}

\newcommand{\rs}{r_\textnormal{S}}
\newcommand{\rh}{r_\textnormal{H}}
\newcommand{\hrs}{\hat{r}_\textnormal{S}}
\newcommand{\hrh}{\hat{r}_\textnormal{H}}
\newcommand{\gcap}{b}
\newcommand{\gpure}{z}
\newcommand{\gres}{r}
\newcommand{\spure}{x}
\newcommand{\hpure}{y}
\newcommand{\gmarg}[2]{\rho_{#1}(#2)}
\newcommand{\smarg}[2]{\rho_{#1}(#2)}
\newcommand{\hmarg}[2]{\rho_{#1}(#2)}
\newcommand{\gmargv}{\rho}
\newcommand{\hmargv}{\rho}
\newcommand{\smargv}{\rho}

\newcommand{\lb}{\llbracket}
\newcommand{\rb}{\rrbracket}

\definecolor{fblue}{RGB}{0,35,149}
\definecolor{fred}{RGB}{237,41,57}

\renewcommand{\cap}{c}

\usepackage{natbib}
 \bibpunct[, ]{(}{)}{,}{a}{}{,}%
 \def\bibfont{\small}%
\TheoremsNumberedThrough     

\EquationsNumberedThrough    

\begin{document}


\RUNAUTHOR{Bahamondes and Dahan}

\RUNTITLE{Hide-and-Seek Game with Capacitated Locations and Imperfect Detection}

\TITLE{Hide-and-Seek Game with Capacitated Locations and Imperfect Detection}

\ARTICLEAUTHORS{%
\AUTHOR{Basti\'an Bahamondes, Mathieu Dahan}
\AFF{School of Industrial and Systems Engineering, Georgia Institute of Technology, Atlanta, Georgia 30332, \{\EMAIL{bbahamondes3@gatech.edu}, \EMAIL{mathieu.dahan@isye.gatech.edu}\}}
} 

\ABSTRACT{%
We consider a variant of the hide-and-seek game in which a seeker inspects multiple hiding locations to find multiple items hidden by a hider. Each hiding location has a maximum hiding capacity and a probability of detecting its hidden items when an inspection by the seeker takes place. The objective of the seeker (resp. hider) is to minimize (resp. maximize) the expected number of undetected items. This model is motivated by strategic inspection problems, where a security agency is tasked with coordinating multiple inspection resources to detect and seize illegal commodities hidden by a criminal organization. To solve this large-scale zero-sum game, we leverage its structure and show that its mixed strategies Nash equilibria can be characterized using their unidimensional marginal distributions, which are Nash equilibria of a lower dimensional continuous zero-sum game. This leads to a two-step approach for efficiently solving our hide-and-seek game: First, we analytically solve the continuous game and compute the equilibrium marginal distributions. Second, we derive a combinatorial algorithm to coordinate the players' resources and compute equilibrium mixed strategies that satisfy the marginal distributions. We show that this solution approach computes a Nash equilibrium of the hide-and-seek game in quadratic time with linear support. Our analysis reveals a complex interplay between the game parameters and allows us to evaluate their impact on the players’ behaviors in equilibrium and the criticality of each location.
}%


\HISTORY{}

\maketitle

%


\section{Introduction}
\label{sec: introduction}

In this article, we study a variant of the \emph{hide-and-seek} game in which two players, the hider and the seeker, coordinate multiple resources among heterogeneous locations. Specifically, the hider determines where to allocate multiple items within capacitated locations. Simultaneously, the seeker inspects a limited number of locations to detect the hidden items. However, detection is supposed to be imperfect: When inspecting a location, the seeker finds the items with a location-specific probability that captures the local effects undermining the seeker's detection capabilities. The seeker (resp. hider) aims to select a (possibly randomized) strategy that minimizes (resp. maximizes) the expected number of items that are undetected. The objective of this work is to efficiently solve this large-scale simultaneous zero-sum game, that is, to compute a mixed strategy Nash equilibrium (NE), and gather insights on the players' equilibrium behaviors.

Our model is motivated by security applications involving for instance a security agency interested in dispatching multiple units to inspect warehouses used by a criminal organization to store illegal commodities such as drugs or weapons \citep{hochbaum2011nuclear}. In such settings, the security agency aims to schedule the patrolling operations of their units to detect and seize the illegal commodities \citep{hess2013police}. Another motivating application of our model involves a utility company tasked with coordinating multiple imperfect sensors to inspect its service network against failures caused by a malicious cyber-physical attacker who is able to target multiple components of the network \citep{pirani2021strategic}. Interestingly, another application of interest concerns auditing election results \citep{Blocki_Christin_Datta_Procaccia_Sinha_2015,behnezhad2018battlefields}.  In such problems, an auditor allocates a limited number of election officials into several polling locations in order to detect electoral fraud by means of recounts. The fraudster may be a malicious organization who is interested in manipulating the results by coordinating its members to tamper with the votes.

Previous related works in the hide-and-seek literature have not simultaneously considered multiple resources for both players, heterogeneous hiding capacities, and imperfect detection \citep{gal2014succession,dziubinski2018hide}. This may reduce the applicability of the results, particularly in security settings. However, simultaneously considering these features introduces new challenges: On one hand, the combinatorial nature of both players' sets of actions due to the resource multiplicity prevents us from computationally solving the game using linear programming techniques or approximation algorithms \citep{FREUND199979,lipton2003playing,doi:10.1287/opre.2019.1853}. On the other hand, the complex interplay between the game's features renders the analytical solutions from previous works inapplicable. Hence, we focus on the following research questions: \emph{(i) How to optimally coordinate multiple imperfect inspection resources to detect multiple hidden commodities?} \emph{(ii) How are the optimal inspection and hiding strategies jointly impacted by the detection, location, and players' characteristics?}

%

\subsection{Contributions} 

In this article, we formulate the hide-and-seek game as a simultaneous zero-sum game $\Gamma$ and extend previous models in the literature by considering the coordination of multiple resources for both players in locations with heterogeneous hiding capacities and detection probabilities. We then leverage the game's structure to derive equilibrium properties of the NE of $\Gamma$. In particular, we show that a strategy profile is a NE of $\Gamma$ if and only if the corresponding marginal inspection probabilities and expected numbers of hidden items at each location form a NE of a lower dimensional continuous game $\widetilde{\Gamma}$ (Proposition \ref{prop:equivalence}). From this equivalence, we derive a two-step approach for solving the hide-and-seek game $\Gamma$. 

First, we analytically solve the continuous game $\widetilde{\Gamma}$ (Theorem \ref{thm: NE disjoint monitoring sets}). We find that NE can be generally classified into three main regime patterns determined by complex parameters that account for the interplay between the players' resources and the heterogeneity of the locations. To the best of our knowledge, the features of our model lead to new NE regimes that have not been observed in the literature. In fact, we show that our analytical solutions describe all pure strategies NE of $\widetilde{\Gamma}$ almost surely (Proposition \ref{prop:iff}).

By solving $\widetilde{\Gamma}$, we obtain marginal inspection probabilities and expected numbers of hidden items at each location in equilibrium of $\Gamma$. Thus, the second step consists of computing a mixed strategy profile of $\Gamma$ that is consistent with these unidimensional marginal distributions. To this end, we extend the algorithm of \cite{dziubinski2018hide} to feasibly coordinate the allocation of multiple resources (Algorithm \ref{alg: construction of prob distribution from marginals}). We show that the algorithm runs in quadratic time and returns equilibrium inspection and hiding strategies with linear supports (Theorem \ref{thm:algorithm} and Corollary \ref{final_corollary}). 

Thus, our approach efficiently solves the hide-and-seek game $\Gamma$. By providing mixed inspection strategies with linear support, our solutions can easily be implemented in practice via a randomized scheduling of inspections that can be performed on a day-to-day basis. Furthermore, our analytical solution of the continuous game $\widetilde{\Gamma}$ decodes the complex interplay between the game parameters and provides insights with respect to their impact on the players’ equilibrium behaviors and the criticality of each location. Such insights can be leveraged by security agencies to inform their inspection decisions.

\subsection{Related Work}

%
%
%
%
%
%

The hide-and-seek game is a two-person zero-sum game introduced by \cite{von1953certain}. In its original version, the hider and the seeker interact on a square matrix of nonnegative entries: The hider selects an entry $a_{ij}$ and the seeker simultaneously selects either a row or a column of the matrix. If the row or column selected by the seeker contains the entry chosen by the hider, then the hider pays the seeker $a_{ij}$; otherwise, the seeker pays the hider $a_{ij}$.
This game has been studied as a general model of strategic mismatch \citep{crawford2007fatal} and its equilibrium strategies are well known \citep{von1953certain,flood1972hide,karlin2016game}. 
It also belongs to the more general category of search games, in which a searcher is concerned with the optimal way of looking for a hidden adversary in a search space; see for example \cite{lidbetter2013search,lidbetter2019searching,clarkson2022classical}, and the surveys of \cite{alpern2006theory} and \cite{hohzaki2016search}.  

Nonetheless, in practical applications, the seeker may be able to simultaneously inspect multiple locations, and the hider may be able to hide multiple items across the search space. Furthermore, the seeker's inspection resources can be affected by local conditions undermining their detection capabilities. Thus, in order to achieve a better utilization of their resources, each player may benefit from efficiently coordinating their allocation across the different locations, a problem that the original model does not address.

One of such practical applications arises in problems of strategic sensor placement for network inspection \citep{milovsevic2019network,pirani2021strategic,dahan2022network,9867753}, in which the defender of a network positions sensors in a subset of given locations to detect attacks caused by a strategic attacker, who can target multiple network components. 
Such models typically account for the detection range of the sensors: Positioning a sensor at a location allows the defender to monitor a subset of network components---referred to as a monitoring set---and potentially detect attacks occurring within it. As a result, attacks may be detected from multiple locations; this overlapping feature renders such games challenging to solve.  \cite{dahan2022network} studied a two-person zero-sum game version of this model under the assumption of perfect detection, and derived approximate NE strategies by means of minimum set covers and maximum set packings. 
\cite{milovsevic2019network} and \cite{9867753} studied variants of this model by respectively considering the critical values of network components and imperfect detection. They derive heuristic approaches to compute good quality solutions in the case of a single attack resource. \cite{pirani2021strategic} formulated a game in which sensors are positioned in the nodes of a networked control system to detect attacks on them, and considered imperfect detection through a linear filter that processes the sensors' measurements to detect attacks. The authors derived equilibrium results using tools from structured systems and graph theory. Finally, a different but related model which features location-specific imperfect detection is the network interdiction problem by \cite{washburn1995two}, in which an interdictor sets up a single inspection checkpoint along one of the arcs of a directed graph, with the aim of interdicting an evader who attempts to traverse a path between two nodes. 
The authors show that NE strategies can be computed in polynomial time using network flow techniques.

In fact, our hide-and-seek game can be used to model a class of instances of the strategic sensor placement problems in which the monitoring sets are mutually disjoint. In such instances, our results are directly applicable and generalize the equilibrium characterizations from \citep{washburn1995two,dahan2022network,9867753}. 
Instances with disjoint monitoring arise in situations where it is desirable to reduce sensor interference or the energetic cost of the network \citep{cardei2005improving,rs6010740}. In other contexts such as in security games, disjoint monitoring is naturally satisfied \citep{powell2009sequential,behnezhad2018battlefields,musegaas2022stackelberg}.

Among the variants of the hide-and-seek game previously examined in the literature, our game is most closely related to the ones by \cite{dziubinski2018hide} and \cite{gal2014succession}. \cite{dziubinski2018hide} consider a version of the game with multiple resources for both players, in which they interact on a set of unit capacity locations, each one associated with a nonnegative value. Simultaneously, the hider (resp. seeker) selects a subset of locations to hide his objects (resp. to inspect). Once the choices are made, the seeker pays the hider the value of each uninspected location containing a hidden object. In contrast, our model considers homogeneous values for all locations, but incorporates heterogeneous hiding capacities and probabilities of successful inspections.


\cite{gal2014succession} propose a pursuit-evasion model of the interaction between a prey and a predator. The prey chooses a location to hide from the predator, who is able to inspect multiple locations. However, if the predator visits the location where the prey is hiding, the capture is uncertain and occurs with some probability. The predator (resp. prey) seeks to maximize (resp. minimize) the probability of capture. Our work extends this model by allowing multiple preys to coordinately hide in heterogeneously capacitated locations. Although the subject of animal behavior is beyond the scope of our work, our model extension addresses analogous situations arising in security domains in which a security agency can dispatch multiple inspection units to inspect heterogeneous locations used by a criminal organization to store multiple illegal commodities.





In both of these games, as in ours, a player's mixed strategy consists of a probability distribution over the set of resource allocations that satisfy capacity constraints and the resource budget. Thus, when players have access to multiple resources, their strategy spaces become exponentially large.  
One approach to handle the dimensionality consists in characterizing the players' strategies in a lower dimensional space. In \cite{dziubinski2018hide} and \cite{gal2014succession}, the games' structures permit the characterization of NE in terms of their marginal probabilities of inspecting each location for the seeker, and their marginal probabilities of hiding an item in each location for the hider.
Then, in order to compute the mixed strategies NE, it becomes necessary to construct probability distributions over the feasible resource allocations that are compatible with these marginal probabilities.

This two-step approach of characterizing equilibrium strategies in terms of marginal distributions and then computing compatible mixed strategies has been previously proposed in the literature, e.g., by \cite{korzhyk2010complexity} and \cite{letchford2013solving} to compute Stackelberg equilibria in security games; by \cite{chan2016multilinear} to compute approximate NE in multilinear games; and by \cite{ahmadinejad2019duels} to compute NE for zero-sum bilinear games, with applications to the Colonel Blotto game \citep{borel1921theorie}. In all these cases, the computation of the equilibrium marginal distributions is carried out via linear programming, and the computation of the mixed strategies from the marginal distributions follows by either an efficient implementation of Birkhoff-von Neumann's  theorem and its generalizations \citep{budish2013designing}, or by the more general algorithm by Gr{\"o}tschel, Lov{\'a}sz, and Schrijver (\cite{grotschel2012geometric}, Theorem 6.5.11) that implements Carath\'eodory's theorem using linear programming techniques. Finding an efficient implementation of Carath\'eodory's theorem has also been addressed in other contexts, such as in mechanism design \citep{cai2012optimal, hoeksma2013two}, scheduling \citep{hoeksma2016efficient}, and ranking systems \citep{kletti2022introducing,kletti2022pareto}.



Our implementation of the two-step approach is closely related to that of \cite{dziubinski2018hide}. In contrast to the above-mentioned literature, we derive analytical expressions for the equilibrium marginal distributions, which allows us to fully understand the interplay between the game parameters and to provide detailed insights regarding their impact on the players' equilibrium behaviors. \cite{dziubinski2018hide} provide a combinatorial algorithm for constructing a mixed strategy with linear support in quadratic time with respect to the number of locations. Their algorithm iteratively decomposes a given vector representing the marginal probabilities of allocating one resource in each location into a convex combination of a linear number of integer resource allocations, and it can be interpreted as a tailored and more efficient implementation of Gr{\"o}tschel, Lov{\'a}sz, and Schrijver's algorithm. Our algorithm extends this decomposition to the more general case in which the marginal distributions represent expected numbers of resources allocated in locations constrained by capacities, and where the budget of resources does not need to be exhausted, as opposed to Dziubi{\'n}ski and 
Roy's setting.

The rest of the article is organized as follows. In Section~\ref{sec: model description}, we formulate our hide-and-seek game. We then characterize and parametrically analyze its NE in Section \ref{sec: analytical characterization of equilibrium strategies} using a lower dimensional continuous game, which we solve analytically. In Section \ref{sec: equilibrium computation}, we derive a combinatorial algorithm to coordinate the players' resources and compute a NE of the hide-and-seek game. We then provide some concluding remarks in Section~\ref{sec: conclusion}. Finally, the proofs of our results are listed in the electronic companion.

\section{Problem Description}
\label{sec: model description}
We consider a hide-and-seek game involving a seeker who is looking for multiple homogeneous items hidden by a hider in a search space consisting of a set of $n$ hiding locations $\llbracket 1,n\rrbracket \coloneqq  \{1,\ldots,n \}$. The locations are \emph{capacitated}; namely, the hider can hide up to $c_{i} \in \mathbb{Z}_{>0}$ items in each location $i$. We let $m \coloneqq \sum_{i=1}^{n}c_i$ be the total hiding capacity. In addition, the seeker has \emph{imperfect} location-specific detection capabilities. 
Specifically, by inspecting a location $i$, the seeker effectively finds the hidden items (if any) in that location with probability $p_{i} \in (0,1]$ (and therefore, the inspection fails and leaves the hidden items undetected with probability $1-p_{i}$). We assume that inspection failures at different locations occur independently. We refer to $p_{i}$ as the \emph{detection rate} of location $i$.

We assume that both the hider and seeker are strategic, and hence we adopt a game-theoretic framework to study their behaviors. We define a simultaneous two-player strategic zero-sum game $\Gamma\coloneqq\langle \{\Seek,\Hide\}, (\Delta_{\Seek}, \Delta_{\Hide}), (-U,U) \rangle$ 
where S is the seeker and H is the hider. S can select up to $\rs\in \mathbb{Z}_{>0}$ hiding locations to inspect. Simultaneously, H can select up to $\rh \in \mathbb{Z}_{>0}$ items to hide. To model the players' action sets, we define a generic set of feasible resource allocations given a vector of capacities $\gcap \in \Z^n_{>0}$ and a budget of resources $\gres \in \Z_{>0}$ as follows:
\begin{align}
\A(\gcap,\gres) \coloneqq  \{\gpure\in \Z^n:\; 0 \leq \gpure_i \leq \gcap_i,\ \forall\, i\in \lb1,n\rb,\text{ and } \sum_{i=1}^{n} \gpure_i \leq \gres \}. \label{generic_allocation}
\end{align}
 The set $\A(\gcap,\gres)$ contains all the vectors in $\Z^n$ representing allocations of up to $\gres$ resources within the locations $i \in \lb1,n\rb$, respecting their capacities given by $\gcap_i$. Then, the pure action sets for S and H are given by  $\A_{\Seek} \coloneqq \A(\boldsymbol{1}_{n}, \rs)$ and $\A_{\Hide} \coloneqq \A(\cap,\rh)$ respectively, where $\boldsymbol{1}_n$ is the vector of ones in $\Z^n$ and $c = (c_i)_{i \in \lb1,n\rb}$ is the vector of hiding capacities. Thus, for every $\spure \in \A_{\Seek}$ and $i \in \lb1,n\rb$, $\spure_i=1$ if S inspects location $i$ and $\spure_i=0$ otherwise, and for each $\hpure \in \A_{\Hide}$ and $i \in \lb1,n\rb$, $\hpure_i$ represents the number of items that H hides at location $i$. 

We consider that the quantity of interest for the players is the average number of undetected items. Thus, we define the players' \emph{payoff function} as follows: 
\begin{align}
 \forall \, (\spure,\hpure) \in \A_{\Seek} \times \A_{\Hide}, \ u(\spure,\hpure) \coloneqq  \sum_{i=1}^{n} (1-p_i\spure_i)\hpure_i. \label{payoff}
\end{align}
For every $i \in \lb1,n\rb$, $1-p_ix_i$ represents the probability that items in location $i$ are undetected given the pure inspection action $\spure\in \A_{\Seek}$.

In such combinatorial security settings, players significantly benefit from randomizing their actions \citep{washburn1995two, pita2008deployed, zhu2015game, gupta2016dynamic, hota2016optimal,   bertsimas2016power, miao2018hybrid}. Thus, we allow the players to use \emph{mixed strategies}, defined as probability distributions over their sets of pure actions. The set of probability distributions over the set of generic feasible resource allocations $\A(\gcap,\gres)$ is defined as $\Delta(\gcap,\gres) \coloneqq  \left\{ \sigma  \in  [0,1]^{\A(\gcap,\gres)}:\, \sum_{{\gpure\in \A(\gcap,\gres)}}\sigma_\gpure = 1 \right\}$. Then, the sets of mixed strategies for S and H are given by $\Delta_{\Seek} \coloneqq \Delta(\boldsymbol{1}_n,\rs)$ and $\Delta_{\Hide} \coloneqq \Delta(\cap,\rh)$, respectively. For every mixed strategy $\sigma^{\Seek} \in \Delta_{\Seek}$ (resp. $\sigma^{\Hide} \in \Delta_{\Hide}$) and every $\spure\in \A_{\Seek}$ (resp. $\hpure\in \A_{\Hide}$), $\sigma^{\Seek}_{\spure}$ (resp. $\sigma^{\Hide}_{\hpure}$) is the probability that action $\spure \in \A_{\Seek}$ (resp. $\hpure \in \A_{\Hide}$) is executed by S (resp. H).

Given a strategy profile $(\sigma^{\Seek}, \sigma^{\Hide}) \in \Delta_{\Seek} \times \Delta_{\Hide}$, the expected payoff is then defined as $U(\sigma^{\Seek}, \sigma^{\Hide})\coloneqq \mathbb{E}_{(\spure, \hpure) \sim (\sigma^{\Seek}, \sigma^{\Hide})}[u(\spure,\hpure)] = \sum_{\spure\in \A_{\Seek}} \sum_{\hpure\in \A_{\Hide}} \sigma^{\Seek}_{\spure} \sigma^{\Hide}_{\hpure} u(\spure,\hpure)$. 
We assume that S (resp. H) seeks to minimize (resp. maximize) $U$. For ease of exposition, we use $U(\spure,\sigma^{\Hide})$ (resp. $U(\sigma^{\Seek}, \hpure)$) to denote the case where $\sigma^{\Seek}_{\spure}=1$ (resp. $\sigma^{\Hide}_{\hpure}=1$) for some $\spure\in \A_{\Seek}$ (resp. $\hpure\in \A_{\Hide}$).

Our game $\Gamma$ is relevant to settings where a city police department is interested in inspecting warehouses that are used by a criminal organization to store illegal commodities (e.g., drugs, weapons). In such settings, $\lb1,n\rb$ represents the set of warehouses, and for each $i \in \lb1,n\rb$, $\cap_i$ represents the maximum number of illegal commodities that can be stored in warehouse $i$. The police department can coordinate multiple police units to simultaneously inspect a maximum of $\rs$ warehouses, while the criminal organization has $\rh$ units of illegal commodities to hide within the warehouses.
The detection rates $p_i$ for $i \in \lb1,n\rb$ capture the local effects that might undermine the detection capabilities of the police units, e.g., warehouse characteristics that can impact the efficacy of drug-sniffing dogs \citep{JEZIERSKI2014112}. The objective of the police department (resp. criminal organization) is to minimize (resp. maximize) the number of illegal commodities that are undetected by the police department. Then, a mixed inspection (resp. hiding) strategy represents a randomized schedule of coordinated operations for the police department (resp. criminal organization).

The standard solution concept for simultaneous noncooperative games is given by \emph{Nash equilibria} (NE), that is, strategy profiles for which no player has an incentive to unilaterally deviate in order to improve their payoff. Thus, a strategy profile $(\sigma^{\Seek^*}, \sigma^{\Hide^*}) \in \Delta_\Seek \times \Delta_\Hide$ is a NE of the game $\Gamma$ if it satisfies
\begin{align*}
  \forall\, (\sigma^{\Seek}, \sigma^{\Hide}) \in  \Delta_{\Seek} \times \Delta_{\Hide}, \  U(\sigma^{\Seek^*}, \sigma^{\Hide}) \leq  U(\sigma^{\Seek^*}, \sigma^{\Hide^*}) \leq U(\sigma^{\Seek}, \sigma^{\Hide^*}).
\end{align*}
We refer to $U(\sigma^{\Seek^*} ,  \sigma^{\Hide^*})$ as the \emph{value of the game} $\Gamma$. Since $\Gamma$ is a zero-sum game with finite pure action sets, Von Neumann's minimax theorem \citep{von1928theorie} implies that the value of the game is unique and the game can be solved using the following linear program (LP):
\begin{alignat}{4}
\underset{t \in \R,\, \sigma^{\Seek}\in \Delta_{\Seek}}{\text{minimize}} \quad && t \quad && \text{subject to} \quad &  U(\sigma^{\Seek}, \hpure) \leq t, \quad & \forall\, \hpure\in \A_{\Hide}. \tag{LP} \label{LP}
\end{alignat}
Specifically, the equilibrium inspection strategies, equilibrium hiding strategies, and value of the game $\Gamma$ are given by the optimal primal solutions, optimal dual solutions, and optimal value of \eqref{LP}, respectively. A remarkable consequence of this result is that no player can benefit from observing the mixed strategy of the other player before making a decision. Nonetheless, since the cardinality of $\A_{\Seek}$ (resp. $\A_{\Hide}$) grows combinatorially with $\rs$ (resp. $\rh$), \eqref{LP} becomes computationally challenging to solve, even for small-sized instances. Similarly, algorithms for computing approximate NE are inapplicable for realistic instances of the game $\Gamma$ \citep{FREUND199979,lipton2003playing,doi:10.1287/opre.2019.1853}. 

Thus, we propose a two-step solution approach for solving the game. First in Section \ref{sec: analytical characterization of equilibrium strategies}, we reduce the dimensionality of the problem by characterizing NE using the marginal inspection probability and expected number of hidden items in each location. Then in Section \ref{sec: equilibrium computation}, we derive an algorithm to coordinate the players' resources and recover NE that are consistent with the characterized marginal probabilities and expected numbers of hidden items in equilibrium.



\section{Analytical Characterization of Equilibrium Strategies}
\label{sec: analytical characterization of equilibrium strategies}


In this section, we show that the NE of the game $\Gamma$ can be characterized using the corresponding marginal inspection probability and expected number of hidden items in each location. We prove that these unidimensional quantities are NE of a smaller-sized continuous game, which we solve analytically. This analytical characterization permits us to examine the impact of the problem parameters on the players' behaviors in equilibrium.


\subsection{Continuous Equivalence}\label{subsec:preliminary}

To simplify our analysis of the game $\Gamma$, we first derive properties of generic randomized resource allocations. Given a vector of capacities $\gcap \in \Z_{>0}^n$ and a budget of resources $\gres \in \Z_{>0}$, we denote by $\widetilde{\A}(\gcap,\gres) \coloneqq  \{\gmargv\in \R^n:\; 0 \leq \gmargv_i \leq \gcap_i,\ \forall\, i\in \lb1,n\rb,\text{ and } \sum_{i=1}^{n} \gmargv_i \leq \gres \}$ the linear programming relaxation of the set of generic feasible resource allocations $\A(\gcap,\gres)$.
Then, for every probability distribution $\sigma \in \Delta(\gcap,\gres)$ over $\A(\gcap,\gres)$, we denote as $\gmargv(\sigma) = (\gmargv_i(\sigma))_{i \in\lb1,n\rb}$ the vector of expected numbers of resources allocated at each location, given by:
\begin{align}
    \label{def: rho_sigma_1 and mu_sigma_2}
    \forall\, i\in \lb1,n\rb, \ \gmarg{i}{\sigma} \coloneqq \mathbb{E}_{\gpure \sim\sigma}[\gpure_i]=\sum_{\gpure \in \A(\gcap,\gres)}\gpure_i \sigma_\gpure.
\end{align}

We present the following relation between $\Delta(\gcap,\gres)$ and $\widetilde{\A}(\gcap,\gres)$:

\begin{restatable}{lemma}{LemmaConvexHull}
\label{lemma:convex_hull}
 Consider a vector of capacities $\gcap \in \Z^n_{>0}$, a budget of resources $\gres \in \Z_{>0}$, and a vector $\gmargv^\prime \in \R^n$. Then, $\gmargv^\prime \in \widetilde{\A}(\gcap,\gres)$ if and only if there exists a probability distribution $\sigma \in \Delta(\gcap,\gres)$ that satisfies $\gmarg{i}{\sigma} = \gmargv_i^\prime$ for all $i \in \lb1,n\rb$.
 
 \end{restatable}
 
Lemma \ref{lemma:convex_hull} is a consequence of the integrality of the polyhedron $\widetilde{\A}(\gcap,\gres)$. Thus, $\widetilde{\A}(\gcap,\gres)$ represents the set of vectors of expected numbers of allocated resources at each location resulting from a probability distribution in $\Delta(\gcap,\gres)$. In particular, for every inspection strategy $\sigma^\Seek \in \Delta_\Seek$ and hiding strategy $\sigma^\Hide \in \Delta_\Hide$, $\smargv(\sigma^\Seek) = (\smarg{i}{\sigma^\Seek})_{i \in \lb1,n\rb}$ and $\hmargv(\sigma^\Hide) = (\hmarg{i}{\sigma^\Hide})_{i \in \lb1,n\rb}$ respectively represent the vectors of marginal inspection probabilities and expected numbers of hidden items across the locations.

Lemma \ref{lemma:convex_hull} permits us to relate the game $\Gamma$ to the continuous zero-sum game $\widetilde{\Gamma} \coloneqq \langle \{\Seek,\Hide\},(\widetilde{\A}_\Seek,\widetilde{\A}_\Hide),(-\tilde{u},\tilde{u})\rangle$, where \Seek and \Hide respectively select a vector of continuous inspection effort $\smargv^\Seek \in \widetilde{\A}_\Seek \coloneqq \widetilde{\A}(\boldsymbol{1}_n,\rs)$ and a vector of continuous amount of hidden items $\hmargv^\Hide \in \widetilde{\A}_\Hide \coloneqq \widetilde{\A}(c,\rh)$,  and the players' payoff in $\widetilde{\Gamma}$ is given by
\begin{align}
\forall \, (\smargv^\Seek,\hmargv^\Hide) \in \widetilde{\A}_\Seek \times \widetilde{\A}_\Hide, \ \tilde{u}(\smargv^\Seek,\hmargv^\Hide) = \sum_{i=1}^n (1 - p_i \smargv^\Seek_i)\hmargv^\Hide_i, \label{continuous_payoff}
\end{align}
as shown by the following proposition:

\begin{restatable}{proposition}{ContinuousEquivalence}
\label{prop:equivalence}
The games $\Gamma$ and $\widetilde{\Gamma}$ are related as follows:
\begin{itemize}
\item[--]
For every strategy profile $(\sigma^{\Seek},\sigma^{\Hide}) \in \Delta_{\Seek} \times \Delta_{\Hide}$ in $\Gamma$, $U(\sigma^{\Seek},\sigma^{\Hide}) = \tilde{u}(\smargv(\sigma^{\Seek}),\hmargv(\sigma^{\Hide}))$. \item[--]
$(\smargv^{\Seek^*},\hmargv^{\Hide^*}) \in \widetilde{\A}_\Seek \times \widetilde{\A}_\Hide$ is a NE of $\widetilde{\Gamma}$ if and only if there exists a NE  $(\sigma^{\Seek^*},\sigma^{\Hide^*}) \in \Delta_{\Seek} \times \Delta_{\Hide}$ of $\Gamma$ that satisfies $\smargv(\sigma^{\Seek^*}) = \smargv^{\Seek^*}$ and $\hmargv(\sigma^{\Hide^*}) = \hmargv^{\Hide^*}$. 

\item[--] The values of the games $\Gamma$ and $\widetilde{\Gamma}$ are identical.
\end{itemize}
\end{restatable}

From this proposition, we deduce that NE of the game $\Gamma$ can be characterized from NE of the game $\widetilde{\Gamma}$. In particular, in a NE $(\smargv^{\Seek^*},\hmargv^{\Hide^*})$ of $\widetilde{\Gamma}$, $\smargv^{\Seek^*}$ (resp. $\smargv^{\Hide^*}$) represents the marginal inspection probabilities (resp. expected numbers of hidden items) at every location in a NE of $\Gamma$. This provides a significant computational advantage, as these quantities can be represented with vectors of size $n$, while the players' strategies in $\Gamma$ require vectors of exponentially large sizes.
In addition, marginal inspection probabilities and expected numbers of hidden items more conveniently quantify the criticality of locations for each player.

\subsection{Preliminary Analysis}

We next derive the intuition behind the strategies for both players in equilibrium. In particular, we describe the players' incentives and constraints that result from the features of our model. This discussion will permit us to introduce and motivate the key quantities that are needed to analytically solve the game $\widetilde{\Gamma}$ and characterize the NE of the game $\Gamma$.


%
%
%
%
%
%
%
%
%
%
%
%

From Proposition \ref{prop:equivalence}, we deduce that the players' expected payoff $U(\sigma^{\Seek},\sigma^{\Hide})$ in $\Gamma$ can be expressed as the sum over the hiding locations $i \in \lb1,n\rb$ of the expected numbers of hidden items that remain undetected, namely, $\left(1-p_i\smarg{i}{\sigma^{\Seek}}\right) \hmarg{i}{\sigma^{\Hide}}$. 
We refer to $1-p_i \rho_{i}(\sigma^{\Seek})$ as the \emph{undetection probability} of location $i$, that is, the probability that the hidden items at location $i$ remain undetected when S plays $\sigma^{\Seek}$. We also refer to $p_i\hmarg{i}{\sigma^{\Hide}}$ as the \emph{detection performance} at location $i \in \lb1,n\rb$, which represents the expected number of hidden items that S is able to detect by inspecting location $i$ when H plays $\sigma^{\Hide}$. These quantities will guide our equilibrium analysis, as H's incentive is to hide items in locations with highest undetection probabilities, and S's incentive is to inspect locations with highest detection performances. 

Due to the players' incentives, we can easily show that when each player has one resource unit (i.e., $\rs = \rh = 1$), then a NE $(\sigma^{\Seek^*},\sigma^{\Hide^*})$ satisfies $\smarg{i}{\sigma^{\Seek^*}} = \hmarg{i}{\sigma^{\Hide^*}} = (1/p_i)/\left(\sum_{j=1}^{n}1/p_{j}\right)$ for every $i \in \lb1,n\rb$, as in  \cite{gal2014succession}. In other words, S inspects each location $i$ with marginal probability proportional to $1/p_i$ so as to equalize the undetection probability of every location. Similarly, H's equilibrium strategy equalizes the detection performance of every location.


Now, if we consider a general number of player resources $\rs \geq 1$ and $\rh \geq 1$, an analogous intuition would suggest that $\smarg{i}{\sigma^{\Seek^*}} = (\rs/p_i)/\left(\sum_{j=1}^{n}1/p_{j}\right)$ and $\hmarg{i}{\sigma^{\Hide^*}} = (\rh/p_i)/\left(\sum_{j=1}^{n}1/p_{j}\right)$ for every $i\in \lb1,n\rb$. From Proposition \ref{prop:equivalence}, $\smargv(\sigma^{\Seek^*})$ and $\hmargv(\sigma^{\Hide^*})$ must belong to $\widetilde{\A}_\Seek$ and $\widetilde{\A}_\Hide$, respectively. However, if $\rs$ is large enough and the detection rates are heterogeneous enough, then $(\rs/p_i)/\left(\sum_{j=1}^{n}1/p_{j}\right)\leq 1$ may be violated for some locations. In such cases, S cannot ensure the desired level of inspection to these locations, thus rendering them more attractive for H. Analogously, if $\rh$ is large enough, the detection rates are heterogeneous enough, and the hiding capacities are small enough, then $(\rh/p_i)/\left(\sum_{j=1}^{n}1/p_{j}\right) \leq \cap_i$ may also be violated for some locations. H cannot ensure the desired level of detection performance across such locations, thus rendering them less attractive for S.

Hence, the features of our model (i.e., multiple player resources and heterogeneous locations) create new challenges that we need to address. For the remainder of this section, we order the locations such that $p_{1}c_{1} \leq \cdots \leq p_{n}c_{n}$. We refer to $p_i c_i$ as the \emph{detection potential} of location $i \in\lb1,n\rb$, that is, the maximum expected number of detected items when S inspects location $i$.
%
Moreover, for any given $i\in \lb 0,n-1\rb$, we define a bijective mapping $\pi^{i}: \lb 1,n-i\rb\to \lb i+1,n\rb$ that satisfies $p_{\pi^{i}(1)} \leq \cdots \leq p_{\pi^i(n-i)}$, i.e., that orders the set of locations $\lb i+1,n\rb$ by their detection rates. For convenience, we define $p_0c_0 \coloneqq 0$ and $p_{\pi^{i}(0)} \coloneqq 0$ for every $i \in \lb 0,n-1\rb$. We also denote $S^{i}_k \coloneqq \sum_{j=k}^{n-i} 1/p_{\pi^i(j)}$ for every $i \in \lb 0,n-1\rb$ and $k \in \lb 1,n-i+1\rb$.
%

We remark that when $\rs\geq n$, an equilibrium inspection strategy for S is to inspect each location, and an equilibrium hiding strategy for H is to hide $\min\{\rh,m\}$ items in the locations $i$ with smallest detection rates $p_i$. Similarly, when $\rh \geq m$, an equilibrium hiding strategy is to exhaust all the hiding capacities, and an equilibrium inspection strategy is to inspect the $\min\{\rs,n\}$ locations $i$ with the highest detection potentials $p_ic_i$. Henceforth, we study the games $\Gamma$ and $\widetilde{\Gamma}$ when $0 < \rs < n$ and $0 < \rh < m$.

From the discussion above, we find that H may not be able to ensure the desired level of detection performance for the locations with lowest detection potentials. If we denote by $\lb 1,i\rb$ such locations, then S will not inspect them, as her incentive is to allocate her resources among the remaining locations $\lb i+1,n\rb$ with higher detection performances. For the remaining locations, S's incentive is to equalize the undetection probabilities. However, this may not be possible if the detection rates $p_j$ for $j \in \lb i+1,n\rb$ are heterogeneous. Instead, S can inspect the $k_i$ most unreliable locations $\{\pi^i(1),\dots,\pi^i(k_i)\}$, that is, the locations with lowest detection rates, and equalize the undetection probabilities for the locations $\{\pi^i(k_i+1),\dots,\pi^i(n-i)\}$. For every $i \in \lb 0,n-1\rb$, the value of $k_i$ is given by the following expression:
\begin{align*}
    k_i \coloneqq \max\left\{ k \in \lb 0, n-i\rb :\, k + p_{\pi^{i}(k)}S^{i}_{k+1} < \rs \right\}.
\end{align*}

As mentioned above, locations for which S cannot achieve the desired level of inspection become more attractive for H. Thus, H's incentive is to exhaust the capacities of the $\ell_i$ most unreliable locations of $\lb i+1,n\rb$, that is, $\{\pi^i(1),\dots,\pi^i(\ell_i)\}$, and equalize the detection performance in locations $\{\pi^i(\ell_i+1),\dots,\pi^i(n-i)\}$. In addition, H must ensure that the detection performance in the latter set of locations is no less than that of the locations in $\lb 1,i\rb$, so that S does not have an incentive to reallocate her resources to $\lb 1,i\rb$ (that were initially uninspected by S). Thus, for every $i \in \lb 0,n-1\rb$, the value of $\ell_i$ is given by the following expression:
\begin{align*}
    \ell_i \coloneqq \max \left\{ \ell \in \lb 0, n-i\rb:\,    \sum_{j=1}^{i} c_{j} + \sum_{j=1}^{\ell} c_{\pi^{i}(j)} + p_{i}c_{i}S^{i}_{\ell+1} < \rh \right\}.
\end{align*}
We note that $\ell_i$ exists when the number of items to hide satisfies $\rh > \sum_{j=1}^i \cap_j + p_i c_i S_1^i$.

Interestingly, the interplay between $k_i$ and $\ell_i$ will play an important role in solving $\widetilde{\Gamma}$ and characterizing the NE of $\Gamma$. Furthermore, it is crucial to determine the number $i$ of locations with smallest detection potentials that will be exhausted by H and uninspected by S, given the players' equilibrium interactions in the remaining set of locations.

\subsection{Analytical Characterization} \label{sec:Main_Result}

From the discussion above, we observe that the players' behaviors in equilibrium depend on capacities, detection rates, detection potentials, numbers of resources, and the parameters $k_i$ and $\ell_i$. To capture this complex interplay and characterize the NE, we define the following key thresholds:
\begin{alignat*}{6}
\tau_{-1} &\coloneqq 0, & \qquad & & \tau_i &\coloneqq \sum_{j=1}^{i}c_{j} + \sum_{j=1}^{k_{i}}  c_{\pi^{i}(j)} + p_{i+1}c_{i+1}S^{i}_{k_{i}+1}, & \qquad & & \forall \, i&\in \lb 0,n-1\rb, \\
 & & \qquad & & \nu_i &\coloneqq \sum_{j=1}^{i}c_{j} + \sum_{j=1}^{k_{i}}c_{\pi^{i}(j)} + p_{i}c_{i}S^{i}_{k_{i}+1}, & \quad & & \forall \, i&\in \lb 0,n-1\rb.
\end{alignat*}

First, we show that these thresholds partition the interval $[0,m]$ in the following manner.

\begin{restatable}{lemma}{Thresholds}
\label{lem: threshold inequalities}
The thresholds $\tau_{-1},\ldots, \tau_{n-1}$, and $\nu_{0},\ldots, \nu_{n-1}$, satisfy $\tau_{-1}=0$, $\tau_{n-1} = m$, and $\tau_{i-1} \leq \nu_i \leq \tau_{i}$ for all $i\in \lb 0, n-1\rb$.
\end{restatable}

Thus, the interval $[0,m]$ is subdivided by the thresholds as $0=\tau_{-1} \leq \nu_0 \leq \tau_0 \leq \nu_1 \leq \cdots \leq \tau_{n-2} \leq \nu_{n-1} \leq \tau_{n-1}=m$. In fact, given a number of items to hide $\rh \in \lb1,m-1\rb$, the subinterval in which $\rh$ resides corresponds to a precise configuration of the parameters and determines a specific equilibrium regime, as shown in the following theorem:

\newlength{\widest}
\settowidth{\widest}{$\displaystyle \frac{h - \sum_{j=1}^{i^{*}}c_{j} -  \sum_{j=1}^{k_{i^{*}}} c_{\pi^{i^{*}}(j)} }{p_{i}S^{i^{*}}_{k_{i^{*}} + 1}}$}

\newlength{\widestx}


\settowidth{\widestx}{$\displaystyle h  -  \sum_{j=1}^{i^*} \hspace{-0.1em} c_{j}  -  \sum_{j=1}^{\ell_{i^{*}}-1}  c_{\pi^{i^*}(j)}  -  p_{i^*}c_{i^*}S^{i^*}_{\ell_{i^{*}}+1} $}

\begin{restatable}{theorem}{NEDisjointMonitoringSets}
\label{thm: NE disjoint monitoring sets}
Given the players' resources $\rs\in \lb1,n-1\rb$ and  $\rh \in \lb1,m-1\rb$, let $i^{*}\in \lb 0, n-1\rb$ satisfying $\tau_{i^{*}-1} < \rh \leq \tau_{i^{*}}$. Then, any strategy profile $(\smargv^{\Seek^*},\hmargv^{\Hide^*})\in \widetilde{\A}_{\Seek} \times \widetilde{\A}_{\Hide}$ that satisfies the conditions below is a NE of $\widetilde{\Gamma}$. Furthermore, any strategy profile $(\sigma^{\Seek^*},\sigma^{\Hide^*})\in \Delta_{\Seek} \times \Delta_{\Hide}$ that satisfies $\rho(\sigma^{\Seek^*}) = \rho^{\Seek^*}$ and $\rho(\sigma^{\Hide^*}) = \rho^{\Hide^*}$ is a NE of $\Gamma$.

\begin{enumerate}
 \item[Regime Pattern 1:] If $\nu_{i^*} < \rh \leq \tau_{i^*}$, then $k_{i^*} \leq \ell_{i^*}$  and sufficient equilibrium conditions are given by:
%
%
%
    \begin{align}
    \OneAndAHalfSpacedXI
        \smargv_{i}^{\Seek^*} &= 
            \begin{dcases}
                \makebox[\widestx][l]{$0$} &\text{if } i\in \I,\\
                1 & \text{if } i \in \J,\\
                \frac{\rs-k_{i^{*}}}{p_i S^{i^{*}}_{k_{i^{*}}+1}} & \text{if } i \in \K,
            \end{dcases} \label{Regime_1_S}\\[2ex]
            \hmargv_{i}^{\Hide^*} &= 
            \begin{dcases}
                \makebox[\widestx][l]{$c_i$} &\text{if } i\in \I \cup \J,\\
                \frac{\rh - \sum_{j=1}^{i^{*}}c_{j} - \sum_{j=1}^{k_{i^{*}}}c_{\pi^{i^{*}}(j)} }{p_i S^{i^{*}}_{k_{i^{*}}+1}} & \text{if } i\in \K,
            \end{dcases}\label{Regime_1_H}
    \end{align}
    where $\I \coloneqq  \{1,\ldots, i^*\}$, $\J \coloneqq \{ \pi^{i^*}(1),\ldots, \pi^{i^*}(k_{i^*})\}$, and $\K \coloneqq \{ \pi^{i^*}(k_{i^{*}}+1),\ldots, \pi^{i^*}(n - i^{*})\}$.
The value of the games $\Gamma$ and $\widetilde{\Gamma}$ is given by 
\begin{align*}
    \rh  - \sum_{j=1}^{k_{i^{*}}} p_{\pi^{i^{*}}(j)}c_{\pi^{i^{*}}(j)} - \frac{ \Big(\rs - k_{i^*}  \Big) \left( \rh - \sum_{j=1}^{i^{*}}c_{j} - \sum_{j=1}^{k_{i^{*}}}c_{\pi^{i^{*}}(j)} \right)}{S^{i^*}_{k_{i^*}+1}}.
\end{align*}

\item[Regime Pattern 2:]  If $i^* = 0$ and $\tau_{-1} < \rh \leq \nu_{0}$, then $\ell_{0^*} < k_{0^*}$ and sufficient equilibrium conditions are given by:
 \begin{align}
    \OneAndAHalfSpacedXI
        \phantom{\smargv_i^{\Seek^*}}&\phantom{\;=}\left\{\begin{aligned}
        \rho^{\Seek^*}_i &= \makebox[0.43\widestx][l]{$1$} \quad &&\text{if } i \in \J\cup\{ \pi^{0}(\ell_{0}+1) \},\\
        \frac{p_{\pi^{0}(\ell_{0}+1)}}{p_{i}} \leq \rho^{\Seek^*}_i & \leq 1  &&\text{if } i \in \K\setminus \{ \pi^{0}(\ell_{0}+1) \}, \\
        \sum_{i=1}^{n} \rho^{\Seek^*}_i &\leq \rs, &&\\
        \end{aligned}\right.\label{Regime_2_S}\\[2ex]
        \rho^{\Hide^*}_i  &= 
            \begin{dcases}
                \makebox[\widestx][l]{$c_i$} &\text{if } i  \in   \J,\\
                 \rh  - \sum_{j=1}^{\ell_{0}}  c_{\pi^{0}(j)}   &\text{if } i = \pi^{0}(\ell_{0}+1), \\
              0 &\text{if } i \in \K \setminus \{ \pi^{0}(\ell_{0}+1) \},
            \end{dcases} \label{Regime_2_H}
\end{align}
where $\J \coloneqq \{\pi^{0}(1),\ldots, \pi^{0}(\ell_{0})\}$ and $\K \coloneqq \{ \pi^{0}(\ell_{0}+1),\ldots, \pi^{0}(n)\}$. The value of the games $\Gamma$ and $\widetilde{\Gamma}$ is given by
\begin{align*}
    \rh - \sum_{j=1}^{\ell_{0}} p_{\pi^{0}(j)}c_{\pi^{0}(j)} - p_{\pi^{0}(\ell_{0}+1)} \left( \rh - \sum_{j=1}^{\ell_{0}}c_{\pi^{0}(j)} \right).
\end{align*}

   \item[Regime Pattern 3:]  If $i^* \geq 1$ and $\tau_{i^*-1} < \rh \leq \nu_{i^*}$, then $\ell_{i^*} < k_{i^*}$ and sufficient equilibrium conditions are given by:
    \begin{align}
    \OneAndAHalfSpacedXI
        \smargv_i^{\Seek^*} &= 
            \begin{dcases}
                \makebox[\widestx][l]{$0$} &\text{if } i \in \I \setminus \left\{ i^* \right\}, \\
                \rs -  \ell_{i^{*}} - p_{\pi^{i^*}(\ell_{i^{*}}+1)} S^{i^*}_{\ell_{i^*}+1} &\text{if } i=i^*,\\
                1 &\text{if } i \in \J\cup \{ \pi^{i^*}(\ell_{i^*}+1) \},\\
                \frac{p_{\pi^{i^*}(\ell_{i^{*}}+1)}}{p_{i}} &\text{if } i \in \K\setminus \{ \pi^{i^*}(\ell_{i^*}+1) \}, 
            \end{dcases} \label{Regime_3_S}\\[2ex] 
        \hmargv_i^{\Hide^*}   &= 
            \begin{dcases}
                \makebox[\widestx][l]{$c_i$} &\text{if } i  \in   \I  \cup  \J,\\
                 \rh -  \sum_{j=1}^{i^*}  c_{j}  - \sum_{j=1}^{\ell_{i^{*}}}  c_{\pi^{i^*}(j)} -  p_{i^*}c_{i^*}S^{i^*}_{\ell_{i^{*}}+2}  &\text{if } i = \pi^{i^*}(\ell_{i^{*}}+1), \\
                \frac{p_{i^*}c_{i^*}}{p_{i}} &\text{if } i \in \K \setminus \{ \pi^{i^*}(\ell_{i^*}+1) \},
            \end{dcases} \label{Regime_3_H}
\end{align}
where $\I \coloneqq  \{1,\ldots, i^*\}$, $\J \coloneqq \{ \pi^{i^*}(1),\ldots, \pi^{i^*}(\ell_{i^*})\}$, and $\K \coloneqq \{ \pi^{i^*}(\ell_{i^{*}}+1),\ldots, \pi^{i^*}(n - i^{*})\}$. The value of the games $\Gamma$ and $\widetilde{\Gamma}$ is given by
\begin{align*}
    \rh - p_{i^{*}} \left(\rs - \ell_{i^*} -p_{\pi^{i^*}(\ell_{i^{*}}+1)}S^{i^{*}}_{\ell_{i^{*}}+1} \right)c_{i^{*}} - \hspace{-0.05cm}\sum_{j=1}^{\ell_{i^{*}}} p_{\pi^{i^{*}}(j)}c_{\pi^{i^{*}}(j)} - p_{\pi^{i^*}(\ell_{i^{*}}+1)} \left( \rh - \hspace{-0.05cm}\sum_{j=1}^{i^{*}}c_{j} - \hspace{-0.05cm}\sum_{j=1}^{\ell_{i^{*}}}c_{\pi^{i^{*}}(j)} \right).
\end{align*}

   \end{enumerate}
\end{restatable}

From Theorem \ref{thm: NE disjoint monitoring sets} and thanks to a carefully selected set of thresholds and parameters, we can analytically solve the game $\widetilde{\Gamma}$ and characterize NE of $\Gamma$. First, given $i^{*}\in \lb 0, n-1\rb$ satisfying $\tau_{i^{*}-1} < \rh \leq \tau_{i^{*}}$, we generally find that $\lb 1,i^*\rb$ represents the collection of locations for which H cannot equalize the detection performance due to their small detection potentials, as intuited in Section \ref{subsec:preliminary}. When this occurs, S's incentive is to utilize her resources to inspect the locations $\lb i^*+1, n\rb$ with higher detection performance, thus leaving the locations $\lb1,i^*\rb$ uninspected while being exhausted by H. However, the players' behaviors in the remaining locations differ depending on the subinterval in which $\rh$ belongs. Indeed, we find that the threshold $\nu_{i^*}$ determines the relation between $k_{i^*}$ and $\ell_{i^*}$, which in turn impacts the set of locations $\J$ that S deterministically inspects and that H exhausts. It also dictates how the players should randomize their remaining resources throughout the locations in equilibrium. 



Theorem \ref{thm: NE disjoint monitoring sets} shows that three major equilibrium regime patterns emerge as a result of the complex and nonlinear interplay between the game parameters, captured by the selected thresholds $\tau$ and $\nu$. In fact, these regime patterns generalize the equilibrium results from the game studied in \cite{gal2014succession}, in which a prey hides from a predator that can inspect multiple locations with heterogeneous detection capabilities. That game can be derived from ours by setting $\cap_i=1$ for all $i\in \lb1,n\rb$ and $\rh=1$. In such a setting, we can show that $i^*=0$ and $\ell_0 = 0$, and the NE regimes observed by the authors correspond to Regime Pattern 1 when $k_0=0$ and to Regime Pattern 2 when $k_0 \geq 1$. 
%
%
We also note that the equilibrium behaviors described in Theorem~\ref{thm: NE disjoint monitoring sets} in its full generality have not been observed in previously studied models that considered homogeneous detection rates \citep{dziubinski2018hide,dahan2022network} or one unit of resources for one or both players \citep{washburn1995two,karlin2016game}.




In fact, if we allow the vector of capacities and the players' resources to be continuous in the game $\widetilde{\Gamma}$, and if we consider the set $\Psi$ of parameters for which $\widetilde{\Gamma}$ is nontrivial, i.e.,
\begin{align}
\Psi \coloneqq \left\{(n,p,c,\rs,\rh) \, : \, n \in \Z_{>0}, \ p \in (0,1]^n, \ c \in \R_{>0}^n, \ \rs \in (0,n), \ \rh \in (0,\sum_{i=1}^n c_i)\right\}, \label{set_Psi}
\end{align}
then we obtain the following stronger result:
\begin{restatable}{proposition}{PropIFF}
\label{prop:iff}
The set of parameters for which conditions \eqref{Regime_1_S}-\eqref{Regime_3_H}  in Theorem \ref{thm: NE disjoint monitoring sets} are necessary and sufficient for a strategy profile $(\smargv^{\Seek^*},\hmargv^{\Hide^*}) \in \widetilde{\A}_\Seek \times \widetilde{\A}_\Hide$ to be a NE of $\widetilde{\Gamma}$ is a dense subset of $\Psi$.
\end{restatable}
%
%

Thus, Proposition \ref{prop:iff} shows that the analytical expressions in Theorem \ref{thm: NE disjoint monitoring sets} describe \emph{all} pure NE of the continuous game $\widetilde{\Gamma}$ almost surely. Similarly, for the original discrete game $\Gamma$, conditions \eqref{Regime_1_S}-\eqref{Regime_3_H} characterize all NE of $\Gamma$, apart from some edge cases that are described in the electronic companion (Proposition \ref{prop:iff_conditions}). 
Next, we provide further insights on the equilibrium behavior in each regime pattern and illustrate them with examples.

\textbf{Regime Pattern 1: $\nu_{i^*} < \rh \leq \tau_{i^*}$.} In this regime, S does not inspect the set of locations $\I=\{1,\dots,i^*\}$, and instead focuses her resources on the remaining locations. Due to the heterogeneity of the detection rates in $\{i^*+1,\dots,n\}$, S deterministically inspects the set of $k_{i^*}$ most unreliable locations, $\J =\{\pi^{i^*}(1),\ldots,\allowbreak \pi^{i^*}(k_{i^*})\}$, and randomizes her $\rs - k_{i^*}$ resources in $\K =\{\pi^{i^*}(k_{i^*}+1),\ldots, \pi^{i^*}(n-i^*)\}$ so as to equalize the undetection probabilities in $\K$. The feasibility of S's strategy is guaranteed by the definition of $k_{i^*}$.

As a result of S's inspection strategy, H's hiding strategy exhausts all locations in $\I$ that are not inspected by S and all locations in $\J$ for which S cannot equalize the undetection probabilities. Then H randomizes his remaining $\rh - \sum_{i \in \I\cup\J}c_i$ resources so as to equalize the detection performance across locations in $\K$. We note that the feasibility and equilibrium guarantee of H's strategy is a consequence of the subinterval in which $\rh$ belongs. Indeed, since $\nu_{i^*} < \rh$, then $k_{i^*} \leq \ell_{i^*}$, which implies that H can exhaust the locations in $\I\cup\J$ and still provide a detection performance that is sufficient so as to not incentivize S to reallocate some of her resources towards the locations in $\I$. Furthermore, since $\rh$ is upper bounded by $\tau_{i^*}$, then H can feasibly equalize the detection performances in $\K$.

We illustrate this regime pattern with the following example:

\begin{example}
\label{ex: nash eqilibrium regime pattern 1}
Consider the hide-and-seek model represented in Figure \ref{fig: regime pattern 1} and assume that $\rs=3$ and $\rh=7$.
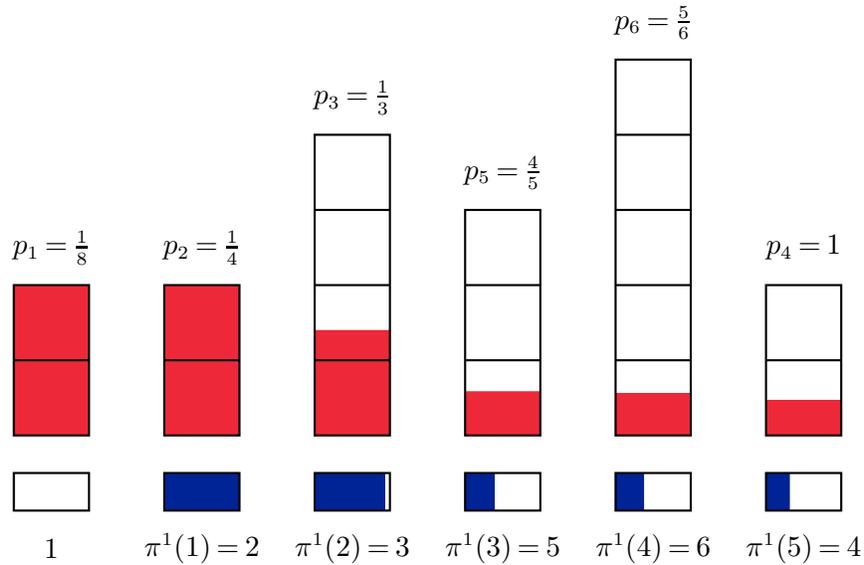
\begin{figure}[htbp]
    \centering
    \begin{tikzpicture}
        \pic at (0,0) {hidingloc={$1$,$p_1=\frac{1}{8}$,2,0,2}};
        \pic at (2,0) {hidingloc={$\pi^{1}(1)=2$,$p_2=\frac{1}{4}$,2,1,2}};
        \pic at (4,0) {hidingloc={$\pi^{1}(2)=3$,$p_3=\frac{1}{3}$,4,40/43,60/43}};
        \pic at (6,0) {hidingloc={$\pi^{1}(3)=5$,$p_5 = \frac{4}{5}$,3,50/129,25/43}};
        \pic at (8,0) {hidingloc={$\pi^{1}(4)=6$,$p_6=\frac{5}{6}$,5,16/43,24/43}};
        \pic at (10,0) {hidingloc={$\pi^{1}(5)=4$,$p_4=1$,2,40/129,20/43}};
    \end{tikzpicture}
    \caption{Illustration of a NE for Regime Pattern 1 when $\rs = 3$ and $\rh=6$. The hiding capacity of each location is represented by the corresponding number of squares. Marginal inspection probabilities (resp. expected numbers of hidden items) in equilibrium are represented by the blue (resp. red) colors.}
    \label{fig: regime pattern 1}
\end{figure}

In this case, $449/80 = \nu_{0} < \rh \leq \tau_1=289/40$ and $i^* = 1$. Furthermore, $k_1 = \ell_1 = 1$ and $\pi^{1}(1) = 2$, $\pi^{1}(2) = 3$, $\pi^{1}(3) = 5$, $\pi^{1}(4) = 6$, $\pi^{1}(5) = 4$. Therefore, in equilibrium, S selects an inspection strategy $\sigma^{\Seek^*} \in \Delta_{\Seek}$ such that $\smarg{1}{\sigma^{\Seek^*}}=0$, $\smarg{2}{\sigma^{\Seek^*}}=1$,  $\smarg{3}{\sigma^{\Seek^*}}=40/43$, $\smarg{4}{\sigma^{\Seek^*}}=40/129$, $\smarg{5}{\sigma^{\Seek^*}}=50/129$, and $\smarg{6}{\sigma^{\Seek^*}}=16/43$. On the other hand, H selects a hiding strategy $\sigma^{\Hide^*}\in \Delta_{\Hide}$ such that $\hmarg{1}{\sigma^{\Hide^*}} =\hmarg{2}{\sigma^{\Hide^*}}=2$, $\hmarg{3}{\sigma^{\Hide^*}} =60/43$, $\hmarg{4}{\sigma^{\Hide^*}} =20/43$, $\hmarg{5}{\sigma^{\Hide^*}} =25/43$, and $\hmarg{6}{\sigma^{\Hide^*}} =24/43$. The value of the game is $U(\sigma^{\Seek^*}, \sigma^{\Hide^*})=5+49/86$. \hfill $\triangle$
\end{example}

\textbf{Regime Pattern 2: $i^* = 0$ and $\tau_{-1} < \rh \leq \nu_{0}$.} In this regime, $\I = \emptyset$ and every location is inspected by S with positive probability. In fact, H's number of resources is too small relative to S's capability of inspecting hiding locations. As a consequence, H's equilibrium strategy consists in greedily hiding items into the locations with smallest detections rates. This results in exhausting the locations in $\J = \{\pi^0(1),\dots,\pi^0(\ell_0)\}$, and assigning his remaining $\rh - \sum_{i \in \J} c_i$ resources in location $\pi^0(\ell_0+1)$. The feasibility H's strategy is guaranteed by the definition of $\ell_0$.

Since S has enough resources, as guaranteed by $k_0 \geq \ell_0+1 > \ell_0$, then her inspection strategy consists in deterministically inspecting locations in $\J\cup\{\pi^0(\ell_0+1)\}$, and randomizing sufficient resources to ensure that the undetection probabilities of the remaining locations are no more than that of location $\pi^0(\ell_0+1)$ so as to prevent H from reallocating some items from $\J$. Interestingly, this can be achieved by S without necessarily utilizing all her resources. 

We illustrate this regime pattern with the following example:

\begin{example}
\label{ex: nash eqilibrium regime pattern special case}
Consider the hide-and-seek model of Figure \ref{fig: regime pattern 2} and assume that $\rs = 6$ and $\rh=6$.

\begin{figure}[htbp]
    \centering
    \begin{tikzpicture}
        \pic at (0,0) {hidingloc={$\pi^{0}(1)=1$,$p_1=\frac{1}{8}$,2,1,2}};
        \pic at (2,0) {hidingloc={$\pi^{0}(2)=2$,$p_2=\frac{1}{4}$,2,1,2}};
        \pic at (4,0) {hidingloc={$\pi^{0}(3)=3$,$p_3=\frac{1}{3}$,4,1,2}};
        \pic at (6,0) {hidingloc={$\pi^{0}(4)=5$,$p_5 = \frac{4}{5}$,3,5/12,0}};
        \pic at (8,0) {hidingloc={$\pi^{0}(5)=6$,$p_6=\frac{5}{6}$,5,2/5,0}};
        \pic at (10,0) {hidingloc={$\pi^{0}(6)=4$,$p_4=1$,2,1/3,0}};
    \end{tikzpicture}
    \caption{Illustration of a NE for Regime Pattern 2 when $\rs = 6$ and $\rh=6$. The hiding capacity of each location is represented by the corresponding number of squares. Marginal inspection probabilities (resp. expected numbers of hidden items) in equilibrium are represented by the blue (resp. red) colors.}
    \label{fig: regime pattern 2}
\end{figure}
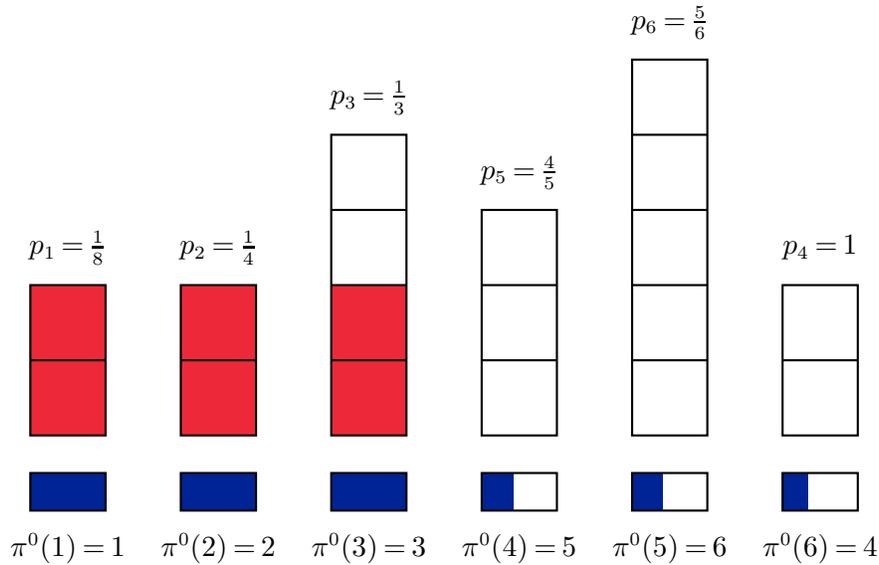

In this case, $0 = \tau_{-1} < \rh \leq \nu_0 = 8$ and $i^* = 0$. Furthermore, $2=\ell_{0} < k_{0} = 4$ and $\pi^0(1) = 1$, $\pi^0(2) = 2$, $\pi^{0}(3) = 3$, $\pi^{0}(4) = 5$, $\pi^{0}(5) = 6$, $\pi^{0}(6) = 4$. Therefore, one equilibrium inspection strategy for S is given by $\sigma^{\Seek^*} \in \Delta_{\Seek}$ satisfying $\rho_{1}(\sigma^{\Seek^*})=\rho_{2}(\sigma^{\Seek^*})=\rho_{3}(\sigma^{\Seek^*})=1$,  $\rho_{4}(\sigma^{\Seek^*})= 1/3$, $\rho_{5}(\sigma^{\Seek^*}) = 5/12$, and $\rho_{6}(\sigma^{\Seek^*}) = 2/5$. We note that S can implement this strategy with only $5 < \rs$ units of resources. On the other hand, H chooses a hiding strategy $\sigma^{\Hide^*}\in \Delta_{\Hide}$ such that $\rho_{1}(\sigma^{\Hide^*})=\rho_{2}(\sigma^{\Hide^*})=\rho_{3}(\sigma^{\Hide^*})=2$, and $\rho_{4}(\sigma^{\Hide^*})=\rho_{5}(\sigma^{\Hide^*})=\rho_{6}(\sigma^{\Hide^*})=0$. The value of the game is $U(\sigma^{\Seek^*}, \sigma^{\Hide^*})=4+7/12$. \hfill $\triangle$
\end{example}

\textbf{Regime Pattern 3: $i^* \geq 1$ and $\tau_{i^*-1} < \rh \leq \nu_{i^*}$.}  In this final regime, we observe an interesting and more complex behavior from the players' equilibrium strategies. H cannot equalize the detection performance across locations in $\I = \{1,\dots,i^*\}$ due to their capacity constraints. Thus, locations in $\I$ are exhausted by H and initially left uninspected by S. Then, since $\ell_{i^*}+1 \leq k_{i^*}$, S cannot achieve the desired level of undetection probability for locations in $\J\cup \{\pi^{i^*}(\ell_{i^*}+1)\} = \{\pi^{i^*}(1),\ldots, \pi^{i^*}(\ell_{i^*}+1)\}$. Therefore, H's incentive is to exhaust locations in $\J$ and randomize some of his resources to equalize the detection performance across $\K = \{\pi^{i^*}(\ell_{i^*}+1),\ldots, \pi^{i^*}(n-i^*)\}$ to that of $i^*$ (i.e., $p_{i^*}c_{i^*}$). Interestingly, H is left with $\rh - \sum_{i \in \I\cup\J}c_i - p_{i^*}c_{i^*} S_{\ell_{i^*}+1}^{i^*} > 0$ resources that he additionally allocates to location $\pi^{i^*}(\ell_{i^*}+1)$, which S deterministically inspects. The feasibility and equilibrium guarantees of this strategy follow from the definition of $\ell_{i^*}$.

As a result of H's strategy, S deterministically inspects locations in $\J \cup\{\pi^{i^*}(\ell_{i^*}+1)\}$ and randomizes some of her resources to equalize the undetection probabilities in $\K\setminus\{\pi^{i^*}(\ell_{i^*}+1)\}$ to that of $\pi^{i^*}(\ell_{i^*}+1)$. This is possible since $\ell_{i^*}+1 \leq k_{i^*}$. Interestingly, S still has $\rs - \ell_{i^*} - 1 - p_{\pi^{i^*}(\ell_{i^*}+1)} S_{\ell_{i^*}+2}^{i^*} > 0$ resources that she can now allocate among the $i^* \geq 1$ locations in $\I$ that were previously left uninspected. S's incentive is to allocate her remaining resources on the location in $\I$ with highest detection performance, namely, $i^*$. Note that feasibility and equilibrium guarantees for this strategy is a consequence of $\rh > \tau_{i^*-1}$. In particular, the resulting undetection probability in $i^*$ is no less than that of locations in $\K$, thus ensuring that H will not reallocate some of his items from $i^*$ to locations in $\K$.

We illustrate this final regime pattern with the following example:

\begin{example}
\label{ex: nash eqilibrium regime pattern 3}
Consider the hide-and-seek model of Figure \ref{fig: regime pattern 3} and assume that $\rs=4$ and $\rh=10$.

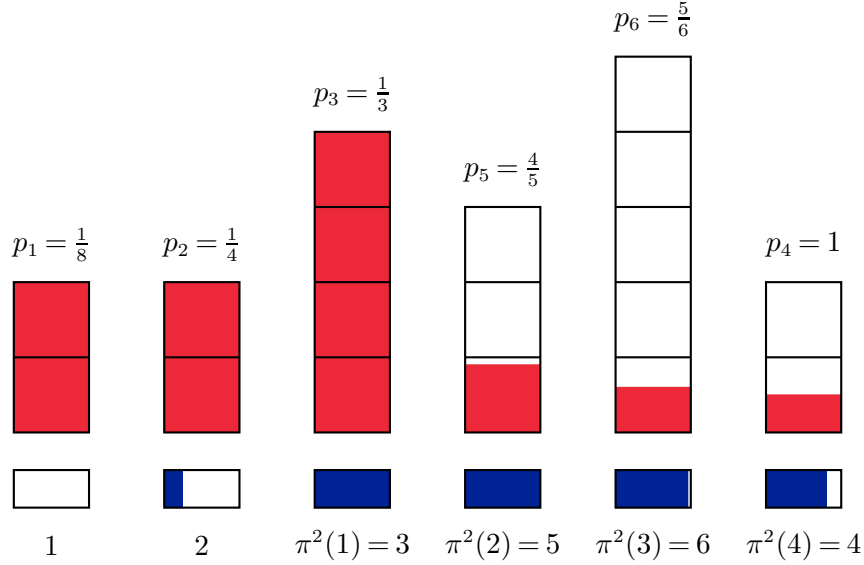
\begin{figure}[htbp]
    \centering
    \begin{tikzpicture}
        \pic at (0,0) {hidingloc={$1$,$p_1=\frac{1}{8}$,2,0,2}};
        \pic at (2,0) {hidingloc={$2$,$p_2=\frac{1}{4}$,2,6/25,2}};
        \pic at (4,0) {hidingloc={$\pi^{2}(1)=3$,$p_3=\frac{1}{3}$,4,1,4}};
        \pic at (6,0) {hidingloc={$\pi^{2}(2)=5$,$p_5 = \frac{4}{5}$,3,1,9/10}};
        \pic at (8,0) {hidingloc={$\pi^{2}(3)=6$,$p_6=\frac{5}{6}$,5,24/25,3/5}};
        \pic at (10,0) {hidingloc={$\pi^{2}(4)=4$,$p_4=1$,2,4/5,1/2}};
    \end{tikzpicture}
    \caption{Illustration of a NE for Regime Pattern 3 when $\rs = 4$ and $\rh=10$. The hiding capacity of each location is represented by the corresponding number of squares. Marginal inspection probabilities (resp. expected numbers of hidden items) in equilibrium are represented by the blue (resp. red) colors.}
    \label{fig: regime pattern 3}
\end{figure}


In this case, $ 389/40 = \tau_1 < \rh \leq \nu_2 = 33/2$ and $i^* = 2$. Furthermore, $k_{2} = 3 > 1 = \ell_2
 $ and $\pi^{2}(1) = 3$, $\pi^{2}(2) = 5$, $\pi^{2}(3) = 6$, $\pi^{2}(4) = 4$. Therefore, in equilibrium, S selects an inspection strategy $\sigma^{\Seek^*} \in \Delta_{\Seek}$ such that $\smarg{1}{\sigma^{\Seek^*}}=0$, $\smarg{2}{\sigma^{\Seek^*}}=6/25$,  $\smarg{3}{\sigma^{\Seek^*}}=1$, $\smarg{4}{\sigma^{\Seek^*}}=4/5$, $\smarg{5}{\sigma^{\Seek^*}}=1$ and $\smarg{6}{\sigma^{\Seek^*}}=24/25$. On the other hand, H chooses a hiding strategy $\sigma^{\Hide^*}\in \Delta_{\Hide}$ such that $\hmarg{1}{\sigma^{\Hide^*}}=\hmarg{2}{\sigma^{\Hide^*}}=2$, $\hmarg{3}{\sigma^{\Hide^*}}=4$, $\hmarg{4}{\sigma^{\Hide^*}}=1/2$, $\hmarg{5}{\sigma^{\Hide^*}}=9/10$ and $\hmarg{6}{\sigma^{\Hide^*}}=3/5$. The value of the game is $U(\sigma^{\Seek^*}, \sigma^{\Hide^*})=6+71/75$. \hfill $\triangle$
\end{example}


%
%
%
%
%
%
\subsection{Parametric Analysis}
\label{sec: parametric analysis}

We continue our analysis by illustrating the impact of the players' resources on the equilibrium regimes of the game $\Gamma$ (and $\widetilde{\Gamma}$). To this end, we consider the hide-and-seek instance defined by the 6 locations, 18 hiding capacities, and detection rates from Figures \ref{fig: regime pattern 1}-\ref{fig: regime pattern 3}. Then, we plot the regions determined by the subintervals $[\tau_{i-1},\nu_i]$ and $[\nu_i,\tau_i]$ for each $i\in \lb 0,n-1\rb$ as a function of $\rs$ and $\rh$.  The resulting plot is shown in Figure~\ref{fig: parametric analysis}. 
\begin{figure}[htbp!]
    \centering
    \begin{tikzpicture}[scale=1]
        \definecolor{color1}{RGB}{7,31,65}  
        \definecolor{color2}{RGB}{0,75,90}
        \definecolor{color3}{RGB}{246,199,111} 
        \definecolor{color4}{RGB}{243,171,113}   
        \definecolor{color5}{RGB}{237,69,52}
        \definecolor{color6}{RGB}{188,52,44}
        \begin{axis}[align =center,
            xmin=1, xmax=18,
            ymin=1, ymax=6,
            every axis plot post/.append style={
            draw=none,
            fill=.!85,
            },
            axis lines=none,
        ]
        \addplot[color=color1, mark=no] table [y=b1, x=tau_m1, col sep=comma] {data/tau_thresholds.csv} \closedcycle;
        \addplot[color=color1!50!color2, mark=no] table [y=b1, x=nu_0, col sep=comma] {data/tau_thresholds.csv} \closedcycle;
        \addplot[color=color2, mark=no] table [y=b1, x=tau_0, col sep=comma] {data/tau_thresholds.csv} \closedcycle;
        \addplot[color=color2!70!color3, mark=no] table [y=b1, x=nu_1, col sep=comma] {data/tau_thresholds.csv} \closedcycle;
        \addplot[color=color3, mark=no] table [y=b1, x=tau_1, col sep=comma] {data/tau_thresholds.csv} \closedcycle;
        \addplot[color=color3!50!color4, mark=no] table [y=b1, x=nu_2, col sep=comma] {data/tau_thresholds.csv} \closedcycle;
        \addplot[color=color4, mark=no] table [y=b1, x=tau_2, col sep=comma] {data/tau_thresholds.csv} \closedcycle;
        \addplot[color=color4!50!color5, mark=no] table [y=b1, x=nu_3, col sep=comma] {data/tau_thresholds.csv} \closedcycle;
        \addplot[color=color5, mark=no] table [y=b1, x=tau_3, col sep=comma] {data/tau_thresholds.csv} \closedcycle;
        \addplot[color=color5!50!color6, mark=no] table [y=b1, x=nu_4, col sep=comma] {data/tau_thresholds.csv} \closedcycle;
        \addplot[color=color6, mark=no] table [y=b1, x=tau_4, col sep=comma] {data/tau_thresholds.csv} \closedcycle;
        \addplot[color=black, mark=no] table [y=b1, x=nu_5, col sep=comma] {data/tau_thresholds.csv} \closedcycle;
        \end{axis}
        \begin{axis}[
            xlabel={$\rh$},
            ylabel={$\rs$},
            xmin=1, xmax=18,
            ymin=1, ymax=6,
            xtick={0,2,..., 18},
            ytick={0,...,6},
            legend pos=outer north east,
            legend image post style={line width=7pt},
            ]
            \addplot[color=color1, mark=no] table [y=b1, x=tau_m1, col sep=comma] {data/tau_thresholds.csv};
            \addplot[color=color1!50!color2, mark=no] table [y=b1, x=nu_0, col sep=comma] {data/tau_thresholds.csv};
            \addplot[color=color2, mark=no] table [y=b1, x=tau_0, col sep=comma] {data/tau_thresholds.csv};
            \addplot[color=color2!70!color3, mark=no] table [y=b1, x=nu_1, col sep=comma] {data/tau_thresholds.csv};
            \addplot[color=color3, mark=no] table [y=b1, x=tau_1, col sep=comma] {data/tau_thresholds.csv};
            \addplot[color=color3!50!color4, mark=no] table [y=b1, x=nu_2, col sep=comma] {data/tau_thresholds.csv};
            \addplot[color=color4, mark=no] table [y=b1, x=tau_2, col sep=comma] {data/tau_thresholds.csv};
            \addplot[color=color4!50!color5, mark=no] table [y=b1, x=nu_3, col sep=comma] {data/tau_thresholds.csv};
            \addplot[color=color5, mark=no] table [y=b1, x=tau_3, col sep=comma] {data/tau_thresholds.csv};
            \addplot[color=color5!50!color6, mark=no] table [y=b1, x=nu_4, col sep=comma] {data/tau_thresholds.csv};
            \addplot[color=color6, mark=no] table [y=b1, x=tau_4, col sep=comma] {data/tau_thresholds.csv};
            \addplot[color=black, mark=no] table [y=b1, x=nu_5, col sep=comma] {data/tau_thresholds.csv};
        \addlegendentry{$\tau_{-1} < \rh \leq \;\nu_0$ $\phantom{=\tau_5}$}
        \addlegendentry{$\nu_{0\phantom{-}} < \rh \leq \;\tau_0$ $\phantom{=\tau_5}$} 
        \addlegendentry{$\tau_{0\phantom{-}} < \rh \leq \;\nu_1$ $\phantom{=\tau_5}$}
        \addlegendentry{$\nu_{1\phantom{-}} < \rh \leq \;\tau_1$ $\phantom{=\tau_5}$}
        \addlegendentry{$\tau_{1\phantom{-}} < \rh \leq \;\nu_2$ $\phantom{=\tau_5}$}
        \addlegendentry{$\nu_{2\phantom{-}} < \rh \leq \;\tau_2$ $\phantom{=\tau_5}$}
        \addlegendentry{$\tau_{2\phantom{-}} < \rh \leq \;\nu_3$ $\phantom{=\tau_5}$}
        \addlegendentry{$\nu_{3\phantom{-}} < \rh \leq \;\tau_3$ $\phantom{=\tau_5}$}
        \addlegendentry{$\tau_{3\phantom{-}} < \rh \leq \;\nu_4$ $\phantom{=\tau_5}$}
        \addlegendentry{$\nu_{4\phantom{-}} < \rh \leq \;\tau_4$ $\phantom{=\tau_5}$}
        \addlegendentry{$\tau_{4\phantom{-}} < \rh \leq \;\nu_5=\tau_5$}
        \end{axis}
    \end{tikzpicture}
    \caption{Illustration of equilibrium regions as a function of the number of resources $\rs$ and $\rh$ for the hide-and-seek instance from Figures \ref{fig: regime pattern 1}-\ref{fig: regime pattern 3}.}
    \label{fig: parametric analysis}
\end{figure}
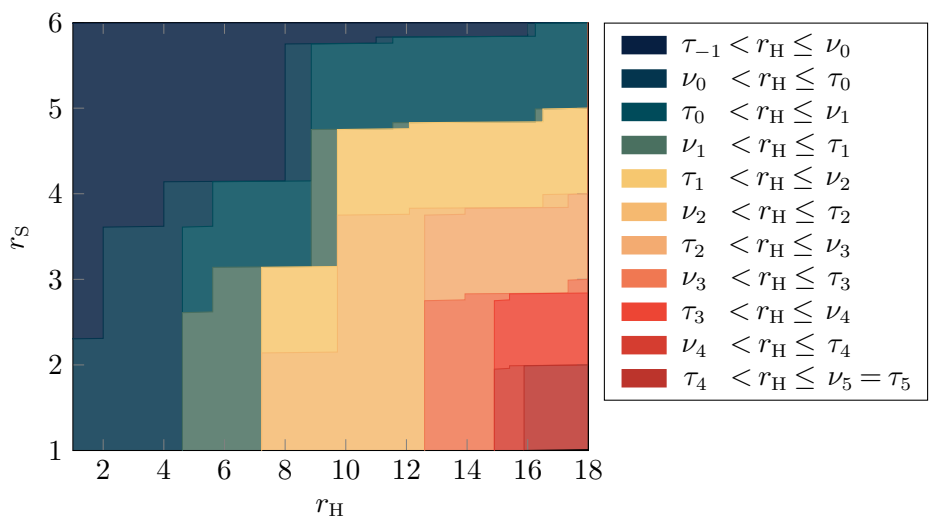

We first observe that given a specific regime, the values of $\rs$ and $\rh$ for which that regime holds form a complex region that may even be disconnected. Indeed, the borders that represent the values of the thresholds $\tau_i$ and $\nu_i$ are defined by step functions of $\rs$ through the parameter $k_i$. Furthermore, for certain values of $\rs$, some thresholds coincide, thus making certain regimes unattainable for any value of $\rh$.

Interestingly, when the number of inspection resources $\rs$ is high, very few equilibrium regimes are possible. In such cases, S has enough resources to not leave any location uninspected, resulting in $i^* = 0$. Conversely, when $\rs$ is low, S's strategy is highly sensitive to the number of items to hide $\rh$. Indeed, if $\rh$ is small, then H can equalize the detection performance across all locations and S should inspect all locations with positive probability. As $\rh$ increases, S must carefully determine which locations to inspect and prefers leaving $i^*$ locations uninspected to find more items in the remaining locations. Furthermore, when $\rs$ is low, nearly all the regimes that are achievable follow Regime Pattern 1 (for different values of $i^*$) since $k_{i^*} \leq \ell_{i^*}$ is most likely to hold for small amounts of inspection resources. Similarly, when $\rh$ is high, a single unit of inspection resources incentives S to focus her inspection on the last location 6 (i.e., $i^*=5$). As $\rs$ increases, S can allocate more resources to $i^*$ according to \eqref{Regime_3_S} (Regime Pattern 3) until $i^*$ is deterministically inspected. At this point, a new regime following Pattern 3 emerges, with a smaller number of uninspected locations $i^*$. The resulting impact of the players' resources on the value of the game is illustrated in Figure \ref{fig: fraction of undetected items}.

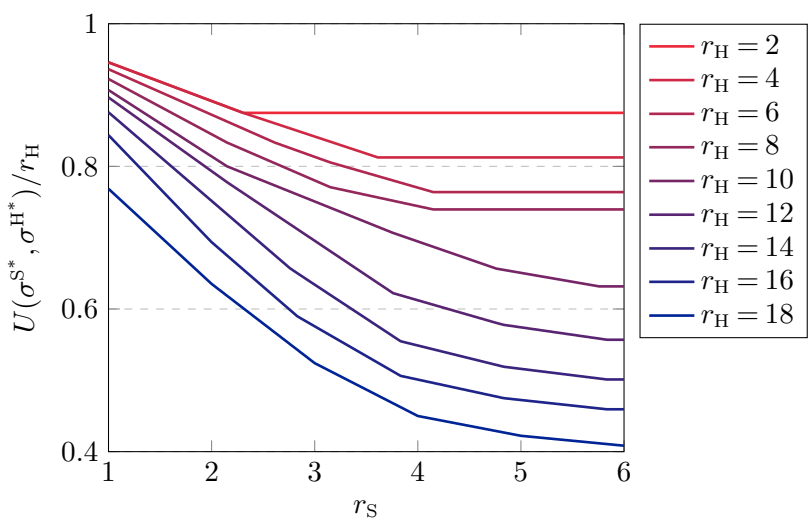
\begin{figure}[htbp!]
    \centering
    \begin{tikzpicture}[scale=1]
        \begin{axis}[align =center,
            xlabel={$\rs$},
            ylabel={$U(\sigma^{\Seek^*},\sigma^{\Hide^*}) / \rh$},
            xmin=1, xmax=6,
            ymin=0.4, ymax=1,
            legend pos=outer north east,
            ymajorgrids=true,
            grid style=dashed,
        ]
        \foreach[evaluate=\i as \colorfrac using (\i-2)*100/(18-2)] \i in {2,4,6,8,10,12,14,16,18}{
            \edef\temp{\noexpand
            \addplot[color=fblue!\colorfrac!fred, line width=1pt] table [x=b1, y=b2_\i, col sep=comma] {data/value_of_the_game_over_rh.csv};
            \noexpand\ifthenelse{\i<10}{\noexpand\addlegendentry{$\noexpand\hspace{-0.5em}r_{\noexpand\textnormal{H}}={\i}$}}{\noexpand\addlegendentry{$r_{\noexpand\textnormal{H}}={\i}$}};
            }
            \temp
        }
        \end{axis}
    \end{tikzpicture}
    \caption{Fraction of undetected items in equilibrium as a function of $\rs$, for different values of $\rh$.}
    \label{fig: fraction of undetected items}
\end{figure}

Specifically, Figure \ref{fig: fraction of undetected items} compares the fractions of undetected items in equilibrium for different amounts of players' resources. This figure first shows the value of focusing inspection resources on locations $\lb i^*+1,n\rb$ when the number of items to hide $\rh$ is high, due to the heterogeneity of the capacities and detection rates. It also shows that the gain in performance by having additional inspection resources varies as a function of the other parameters. This analysis can be utilized in situations when S must determine the appropriate number of inspection resources that would balance cost and performance.

In general, our analytical study provides us with valuable insights on the impact of the problem characteristics (e.g., detection rates, hiding capacities, amounts of resources) on the players' behaviors as well as the locations' importance and criticality in equilibrium, which can be leveraged by security decision makers.

\section{Equilibrium Computation}
\label{sec: equilibrium computation}


In the previous section, Theorem \ref{thm: NE disjoint monitoring sets} solves the continuous game $\widetilde{\Gamma}$, and provides marginal inspection probabilities and expected numbers of hidden items at each location in equilibrium of the game $\Gamma$. However, to solve $\Gamma$, we must determine the coordination of $\rs$ inspection resources and $\rh$ items to hide  in order to satisfy these quantities. Specifically, given the vectors of marginal inspection probabilities $\smargv^{\Seek^*} \in \widetilde{\A}_\Seek$ and expected numbers of hidden items $\hmargv^{\Hide^*} \in \widetilde{\A}_\Hide$ in Theorem~\ref{thm: NE disjoint monitoring sets}, we next seek to efficiently compute mixed strategies $(\sigma^{\Seek^*},\sigma^{\Hide^*})\in \Delta_{\Seek} \times \Delta_{\Hide}$ such that $\smargv(\sigma^{\Seek^*}) = \smargv^{\Seek^*}$ and $\hmargv(\sigma^{\Hide^*}) = \hmargv^{\Hide^*}$. From Proposition \ref{prop:equivalence}, this will ensure that $(\sigma^{\Seek^*},\sigma^{\Hide^*})$ is a NE of $\Gamma$.

Thus, we aim to solve the following generic problem: Given a vector of capacities $\gcap \in \Z_{>0}^n$, a budget of resources $\gres \in \Z_{>0}$, and a vector $\gmargv \in \widetilde{\A}(\gcap,\gres)$, find a solution to the feasibility problem $\{\sigma \in \R^{\A(\gcap,\gres)}_{\geq 0} : \, \sum_{\gpure \in \A(\gcap,\gres)} \sigma_\gpure \gpure = \gmargv, \ \sum_{\gpure \in \A(\gcap,\gres)} \sigma_\gpure = 1\}$, which is guaranteed to exist by Lemma \ref{lemma:convex_hull}.  Although this problem involves an exponential number of variables, Carathéodory's theorem guarantees that a solution exists with a support of size at most $n+1$. In fact, this problem can be solved in polynomial time using the ellipsoid method. However, this method is known to be practically inefficient \citep{behnezhad2017faster}. Thus, we derive another algorithm to efficiently solve the feasibility problem.

Specifically, we extend the algorithm proposed by \cite{dziubinski2018hide} that computes in time $O(n^2)$ a probability distribution with linear support over the set $\{\gpure\in \{0,1\}^n:\; \sum_{i=1}^{n} \gpure_i = \gres \}$ consistent with prescribed unidimensional marginal probabilities. However, we cannot apply their algorithm to construct the equilibrium inspection and hiding strategies in our game $\Gamma$, as our model involves locations with possibly heterogeneous capacities, and probabilities may be assigned to resource allocations $\gpure$ that do not utilize the whole budget of resources $\gres$ (e.g., in Theorem \ref{thm: NE disjoint monitoring sets} -- Regime Pattern 2). 


Thus, given a vector of capacities $\gcap \in \Z_{>0}^n$, a budget of resources $\gres \in \Z_{>0}$, and a vector $\gmargv \in \widetilde{\A}(\gcap,\gres)$, our algorithm returns a probability distribution $\sigma \in \Delta(\gcap,\gres)$ that satisfies $\sum_{\gpure \in \A(\gcap,\gres)} \sigma_\gpure \gpure_i =\gmargv_i$ for every $i \in \lb 1,n\rb$. The general idea of each iteration of the algorithm is to express a given vector $\gmargv$ in $\widetilde{\A}(\gcap,\gres)$ as a convex combination of a vector in $\A(\gcap,\gres)$ (i.e., with only integer components) and a vector in $\widetilde{\A}(\gcap,\gres)$ with one more integer component than $\gmargv$ has. 
To this end, the algorithm allocates $\floor*{\gmargv_i}$ resources at each location $i \in \lb 1,n\rb$, and then determines where to allocate some of the remaining $\bar{\gres}\coloneqq \gres - \sum_{i=1}^n \floor*{\gmargv_i}$ resources given the fractional part of each component of $\gmargv$, defined as $\bar{\gmargv}_i \coloneqq \gmargv_i - \floor*{\gmargv_i}$.
%
%
%
The algorithm determines the maximum number of locations $q$ to allocate the remaining resources  $\bar{\gres}$. Naturally, $q$ is upper bounded by $\bar{\gres}$ and the number of positive components of $\bar{\gmargv}$. Then, the algorithm carefully assigns positive probability to a resource allocation that first assigns $\floor*{\gmargv_i}$ resources at each location, and then assigns one additional resource at each of the $q$ locations with highest fractional parts $\bar{\gmargv}_i$.
%
%
%
The algorithm then updates the vector $\bar{\gmargv}$ so that the new vector contains at least one more integer component than $\bar{\gmargv}$ has. The algorithm iterates until $\bar{\gmargv} \in \{0,1\}^n$, at which point the algorithm assigns the remaining probability to $\floor*{\gmargv}+ \bar{\gmargv} \in \A(\gcap,\gres)$. We refer the reader to Algorithm~\ref{alg: construction of prob distribution from marginals} for the detailed pseudocode.

\begin{algorithm}[htbp!]
\label{alg: construction of prob distribution from marginals}
\DontPrintSemicolon
\OneAndAHalfSpacedXI
    \SetKwInOut{Input}{Input}
    \SetKwInOut{Output}{Output}

\Input{A vector of capacities $\gcap \in \Z^n_{>0}$, a budget of resources $\gres \in  \mathbb{Z}_{> 0}$, and a vector $\gmargv \in \widetilde{\A}(\gcap,\gres)$.}
%
%

    \Output{A probability distribution $\sigma \in \Delta(\gcap,\gres)$ satisfying $\mathbb{E}_{\gpure \sim\sigma}[\gpure] = \gmargv$.}

    $\sigma \gets \boldsymbol{0}_{\A(\gcap,\gres)}$,  \quad $\bar{\gmargv}^1 \gets \gmargv - \floor*{\gmargv}$, \quad $\bar{\gres} \gets \gres - \sum_{i=1}^{n}\floor*{\gmargv_i}$\;
     $\gamma^1 \gets 1$, \quad $k \gets 1$ \;
     \While{$\bar{\gmargv}^k \notin \{0,1\}^n$\label{alg:while_start}} {
      $\theta^k \gets$ Permutation of $\lb 1,n\rb$ such that $\bar{\gmargv}^k_{\theta^k(1)} \geq \cdots \geq \bar{\gmargv}^k_{\theta^k(n)}$\;
        $q^k \gets \min\left\{ \bar{\gres}, \left| \left\{ i \in \lb 1,n\rb:\, \bar{\gmargv}^k_i > 0   \right\} \right|  \right\}$ \;
        \If{$q^k < n$}{
        $\delta^k \gets \min\{\bar{\gmargv}^k_{\theta^k(q^k)}, 1-\bar{\gmargv}^k_{\theta^k(q^k+1)} \}$\;}
        \Else{$\delta^k \gets \bar{\gmargv}^k_{\theta^k(q^k)}$}
        $e^k \gets \boldsymbol{0}_{n}$ \;
        \ForEach{$j \in \{1,\ldots, q^k\}$} {
            $e_{\theta^k(j)}^k \gets 1$\;
        }
        $\sigma_{\floor*{\gmargv} + e^k} \gets \sigma_{\floor*{\gmargv}  + e^k} + \gamma^k\delta^k$ \;
        
         $\bar{\gmargv}^{k+1} \gets \frac{1}{1-\delta^k}(\bar{\gmargv}^k - \delta^ke^k)$\;
         
%

	$\gamma^{k+1} \gets \gamma^k(1-\delta^k)$\;
              $k \gets k+1$  \;\label{alg:while_end}
     }
     
     
    $\sigma_{\floor*{\gmargv}+\bar{\gmargv}^k} \gets \sigma_{\floor*{\gmargv}+\bar{\gmargv}^k} + \gamma^k$ \;
    \Return{$\sigma$}\;

    \caption{Resource Coordination.}
\end{algorithm}


We note that at each iteration $k$ of the while loop (\ref{alg:while_start}-\ref{alg:while_end}), $\bar{\gmargv}^k = \delta^k e^k + (1-\delta^k) \bar{\gmargv}^{k+1}$. In fact, the algorithm selects $e^k$ and $\delta^k$ so that $\bar{\gmargv}^k$ can be expressed as a convex combination of $e^k \in \{0,1\}^n$ and a new vector $\bar{\gmargv}^{k+1}$ that has at least one more integer component than $\bar{\gmargv}^k$ does. This guarantees that the algorithm terminates. At termination, the algorithm expresses $\gmargv$ as a convex combination of resource allocations in $\A(\gcap,\gres)$. This can be translated into the following theorem:
%
%
%
%
%
%
\begin{restatable}{theorem}{TheoremAlgorithm}
\label{thm:algorithm}
Given a vector of capacities  $\gcap \in \Z^n_{>0}$, a budget of resources $\gres \in \Z_{>0}$, and a vector $\gmargv \in \widetilde{\A}(\gcap,\gres)$, Algorithm \ref{alg: construction of prob distribution from marginals} returns a probability distribution $\sigma \in \Delta(\gcap,\gres)$ satisfying $\mathbb{E}_{\gpure\sim\sigma}[\gpure] = \gmargv$. Furthermore, $\sigma$ has a support of size at most $n+1$ and is computed in time $O(n^2)$.
\end{restatable}

Thus, Algorithm~\ref{alg: construction of prob distribution from marginals} matches the support size guaranteed by Carathéodory's theorem on the polytope $\widetilde{\A}(\gcap,\gres)$. In fact, the algorithm performs at most $n$ iterations. Furthermore, by reutilizing the sorting of $\bar{\gmargv}^k$ of the previous iteration, we can implement each iteration (except for the first one) in time $O(n)$, which guarantees an overall running time of $O(n^2)$.

We can now summarize the overall solution approach for computing a NE of the game $\Gamma$. First, marginal inspection probabilities and expected numbers of hidden items in each location are computed according to Theorem~\ref{thm: NE disjoint monitoring sets}: Sorting the detection potentials $(p_ic_i)_{i\in\lb 1,n\rb}$ and the detection rates $(p_i)_{i\in\lb 1,n\rb}$ in non-decreasing order requires $O(n \log n)$ steps. Then, for every $i \in \lb0,n-1\rb$, the mapping $\pi^i$, parameters $k_i$ and $\ell_i$, and thresholds $\tau_i$  and $\nu_i$ can be computed in time $O(n)$. Identifying the subinterval $[\tau_{i^*-1}, \nu_{i^*}]$ or $[\nu_{i^*},\tau_{i^*}]$ in which $\rh$ belongs can be performed in $O(\log n)$ steps, and evaluating the expressions from Theorem \ref{thm: NE disjoint monitoring sets} requires $O(n)$ more steps. Thus, computing equilibrium marginal inspection probabilities and expected numbers of hidden items can be implemented in time $O(n^2)$. Finally, Algorithm~\ref{alg: construction of prob distribution from marginals} computes the mixed strategies consistent with the unidimensional quantities in time $O(n^2)$. Hence, we obtain the final result:

\begin{corollary}
\label{final_corollary}
The game $\Gamma$ can be solved in time $O(n^2)$ with equilibrium strategies of support size at most $n+1$ each.
\end{corollary}

Thus, we obtain an efficient solution approach for solving the large-scale hide-and-seek game $\Gamma$ with multiple resources and heterogeneous locations. Finally, our approach provides solutions with small supports that can be easily implemented by security decision makers.
%
%
%
%
%

\section{Conclusion}
\label{sec: conclusion}

In this work, we investigated a hide-and-seek game in which a seeker inspects locations to find items hidden by a hider. We extended previous models in the literature by considering the coordination of multiple resources for both players in locations with heterogeneous hiding capacities and probabilities of detecting hidden items where the search takes place. The objective of the seeker (resp. hider) is to minimize (resp. maximize) the expected number of undetected items. To compute the mixed strategies Nash equilibria of this large-scale zero-sum game, we proposed a solution approach that first derives analytical equilibrium properties and then efficiently coordinates the players' resources. In particular, we showed that the marginal inspection probabilities and expected numbers of hidden items in each location in equilibrium form a Nash equilibrium of a continuous zero-sum game. By carefully selecting a set of parameters and thresholds, we analytically solved this continuous game. Our analysis highlighted a complex interplay between the game parameters and permitted us to evaluate their impact on the players' behaviors in equilibrium and the criticality of each location. Then, we derived a quadratic time algorithm that coordinates the players' resources to satisfy the characterized equilibrium marginal distributions and computes a Nash equilibrium of the hide-and-seek game with linear support. Our insights and solution approach can be used to inform security agencies that are interested in scheduling multiple units to detect and seize illegal commodities hidden by a criminal organization.

This work can be extended in multiple directions by considering additional features that would widen the applicability of our results. One extension is to consider different commodities hidden by the hider with heterogeneous values, as well as different inspection resources with different detection characteristics. Another extension is to incorporate heterogeneous valuations of the locations to extend the original setting by \cite{von1953certain}. Finally, an interesting research direction is to consider a repeated version of the hide-and-seek game with players learning the initially unknown characteristics of their opponent while they interact.

\ACKNOWLEDGMENT{%
This work was supported by the Georgia Tech Stewart fellowship and new faculty start up grant. We are grateful to Mohit Singh for his valuable feedback.
}


%
%
%


\bibliographystyle{style/informs2014} 
\bibliography{references} 



\ECSwitch

\ECHead{Proofs of Statements}




\section{Proofs of Section \ref{sec: analytical characterization of equilibrium strategies}}


 
 \begin{proof}{\emph{Proof of Lemma \ref{lemma:convex_hull}.}} Let $\gcap \in \Z^n_{>0}$ be a vector of capacities and $\gres \in \Z_{>0}$ be a budget of resources. First, we observe that $\widetilde{\A}(\gcap,\gres)$ is the convex hull of $\A(\gcap,\gres)$. Indeed, since the row vector $\boldsymbol{1}^\top_n$ is totally unimodular, then \eqref{generic_allocation} is an ideal formulation of $\A(\gcap,\gres)$ and the polytope $\widetilde{\A}(\gcap,\gres)$ has integral extreme points, which we denote $\gpure^1,\dots,\gpure^I$. Consider a vector $\gmargv^\prime \in \R^n$. If $\gmargv^\prime \in \widetilde{\A}(\gcap,\gres)$, then there exist $\lambda_1,\dots,\lambda_I \in [0,1]$ such that $\gmargv^\prime = \sum_{i=1}^I \lambda_i \gpure^i$ and $\sum_{i=1}^I \lambda_i = 1$. Then, $\sigma \in [0,1]^{\A(\gcap,\gres)}$ defined by $\sigma_{z^i} = \lambda_{i}$ for $i \in \lb 1,I\rb$ and $\sigma_z = 0$ if $z \notin \{z^1,\dots,z^I\}$ is a probability distribution in $\Delta(\gcap,\gres)$ and satisfies $\gmarg{i}{\sigma} = \gmargv^\prime_i$ for all $i \in \lb 1,n\rb$. Conversely, if there exists $\sigma \in \Delta(\gcap,\gres)$ that satisfies $\gmargv^\prime_i = \gmarg{i}{\sigma} = \sum_{\gpure \in \A(\gcap,\gres)} \gpure_i \sigma_\gpure$ for all $i \in \lb 1,n\rb$, then $\gmargv^\prime$ is a convex combination of elements on $\A(\gcap,\gres)$, and belongs to its convex hull, $\widetilde{\A}(\gcap,\gres)$.
 \hfill\Halmos
 \end{proof}


 \begin{proof}{\emph{Proof of Proposition \ref{prop:equivalence}.}}
 \begin{itemize}
 \item[--]
 Consider a strategy profile $(\sigma^{\Seek},\sigma^{\Hide}) \in \Delta_\Seek \times \Delta_\Hide$. Then,
\begin{align}
    U(\sigma^{\Seek},\sigma^{\Hide}) \overset{\eqref{payoff}}{=} \sum_{i=1}^n (1- p_i \mathbb{E}_{\spure \sim \sigma^\Seek}[\spure_i])\mathbb{E}_{\hpure \sim\sigma^{\Hide}}[\hpure_i] \overset{\eqref{def: rho_sigma_1 and mu_sigma_2},\eqref{continuous_payoff}}{=} \tilde{u}(\smargv(\sigma^\Seek),\hmargv(\sigma^\Hide)). \label{equivalence_payoffs}
 \end{align}
 \item[--] Let $(\smargv^{\Seek^*},\hmargv^{\Hide^*}) \in \widetilde{\A}_\Seek \times   \widetilde{\A}_\Hide$ be a NE of $\widetilde{\Gamma}$. From Lemma \ref{lemma:convex_hull}, let $(\sigma^{\Seek^*},\sigma^{\Hide^*}) \in \Delta_\Seek \times \Delta_\Hide$ that satisfies $\smargv(\sigma^{\Seek^*}) = \smargv^{\Seek^*}$ and $\hmargv(\sigma^{\Hide^*}) = \hmargv^{\Hide^*}$. Using equilibrium conditions in $\widetilde{\Gamma}$, we obtain:
 \begin{align*}
 &\forall \sigma^\Seek \in \Delta_\Seek, \ U(\sigma^{\Seek^*},\sigma^{\Hide^*}) \overset{\eqref{equivalence_payoffs}}{=} \tilde{u}(\smargv^{\Seek^*},\hmargv^{\Hide^*}) \leq \tilde{u}(\smargv(\sigma^\Seek),\hmargv^{\Hide^*}) \overset{\eqref{equivalence_payoffs}}{=} U(\sigma^{\Seek},\sigma^{\Hide^*}),\\
 &\forall \sigma^\Hide \in \Delta_\Hide, \ U(\sigma^{\Seek^*},\sigma^{\Hide^*}) \overset{\eqref{equivalence_payoffs}}{=} \tilde{u}(\smargv^{\Seek^*},\hmargv^{\Hide^*}) \geq \tilde{u}(\smargv^{\Seek^*},\hmargv(\sigma^\Hide)) \overset{\eqref{equivalence_payoffs}}{=} U(\sigma^{\Seek^*},\sigma^{\Hide}).
 \end{align*}
 Thus, $(\sigma^{\Seek^*},\sigma^{\Hide^*})$ is a NE of $\Gamma$.
 
 Conversely, let $(\sigma^{\Seek^*},\sigma^{\Hide^*}) \in \Delta_\Seek \times \Delta_\Hide$ be a NE of $\Gamma$, then we must show that $(\smargv(\sigma^{\Seek^*}), \hmargv(\sigma^{\Hide^*}))$ is a NE of $\widetilde{\Gamma}$. First, from Lemma \ref{lemma:convex_hull}, we know that $(\smargv(\sigma^{\Seek^*}), \hmargv(\sigma^{\Hide^*})) \in \widetilde{\A}_\Seek \times \widetilde{\A}_{\Hide}$. Then, given $(\smargv^\Seek,\hmargv^\Hide) \in \widetilde{\A}_\Seek \times \widetilde{\A}_{\Hide}$, let $(\sigma^\Seek,\sigma^\Hide) \in \Delta_\Seek \times \Delta_\Hide$ satisfying $\smargv(\sigma^\Seek) = \smargv^\Seek$ and $\hmargv(\sigma^\Hide) = \hmargv^\Hide$ (Lemma \ref{lemma:convex_hull}). Using equilibrium conditions in $\Gamma$, we obtain:
 \begin{align*}
 &\tilde{u}(\smargv(\sigma^{\Seek^*}), \hmargv(\sigma^{\Hide^*})) \overset{\eqref{equivalence_payoffs}}{=} U(\sigma^{\Seek^*},\sigma^{\Hide^*}) \leq U(\sigma^\Seek,\sigma^{\Hide^*}) \overset{\eqref{equivalence_payoffs}}{=}  \tilde{u}(\smargv^\Seek, \hmargv(\sigma^{\Hide^*})),\\ 
 \text{and } &\tilde{u}(\smargv(\sigma^{\Seek^*}), \hmargv(\sigma^{\Hide^*})) \overset{\eqref{equivalence_payoffs}}{=} U(\sigma^{\Seek^*},\sigma^{\Hide^*}) \geq U(\sigma^{\Seek^*},\sigma^{\Hide}) \overset{\eqref{equivalence_payoffs}}{=}  \tilde{u}(\smargv(\sigma^{\Seek^*}),\hmargv^\Hide).
 \end{align*}
 Thus, $(\smargv(\sigma^{\Seek^*}), \hmargv(\sigma^{\Hide^*}))$ is a NE of $\widetilde{\Gamma}$.

 \item[--] Given a NE $(\sigma^{\Seek^*},\sigma^{\Hide^*}) \in \Delta_\Seek \times \Delta_\Hide$ of $\Gamma$, we know that $(\smargv(\sigma^{\Seek^*}), \hmargv(\sigma^{\Hide^*}))$ is a NE of $\widetilde{\Gamma}$ and $U(\sigma^{\Seek^*},\sigma^{\Hide^*}) \overset{\eqref{equivalence_payoffs}}{=}  \tilde{u}(\smargv(\sigma^{\Seek^*}),\hmargv(\sigma^{\Hide^*}))$. Thus, the values of the games $\Gamma$ and $\widetilde{\Gamma}$ are identical.  \hfill\Halmos
 \end{itemize}

 \end{proof}
 
 

We note that when some detection rates are identical, permutations $\pi^i$ ordering $\lb i+1,n\rb$ by their detection rates may not be unique. To simplify our proofs, we assume without loss of generality that $\pi^i$ maintains the order between identical detection rates, i.e., $\pi^i(j) < \pi^i(k)$ when $1 \leq j < k \leq n-i$ and $p_{\pi^i(j)} = p_{\pi^i(k)}$, thus rendering $\pi^i$ unique for every $i \in \lb 0,n-1\rb$.

%

Before proving Lemma \ref{lem: threshold inequalities} we need the following auxiliary lemmas.

\begin{lemma}
\label{lem: relation pi_(i-1) and pi_(i)}
Let $i\in \lb 1,n-1\rb$ and $j^*\in \lb 1,n-i+1\rb$ be such that $\pi^{i-1}(j^{*})=i$. Then:
\begin{align}
\OneAndAHalfSpacedXI
\label{eq: relation pi^i and pi^(i-1)}
    \pi^{i-1}(j) = 
        \begin{cases}
            \pi^{i}(j) & \text{if } j \in \lb1, j^{*} - 1\rb,\\
            i & \text{if } j=j^{*}, \\
            \pi^{i}(j-1) & \text{if } j \in \lb j^{*} + 1, n-i+1 \rb.
        \end{cases}
\end{align}

Moreover, for every $k\in \lb 0, n-i+1\rb$,
\begin{align}
\OneAndAHalfSpacedXI
\label{eq: relation sum cardinalities pi^i and pi^(i-1)}
    \sum_{j=1}^{k}c_{\pi^{i-1}(j)} = \begin{dcases}
        \sum_{j=1}^{k}c_{\pi^{i}(j)}  & \text{if } k \in \lb 0, j^{*}-1 \rb,\\
        \sum_{j=1}^{k-1}c_{\pi^{i}(j)} + c_i  & \text{if } k \in \lb j^{*}, n-i+1 \rb,
    \end{dcases}
\end{align}

and
\begin{align}
\OneAndAHalfSpacedXI
\label{eq: relation sum probs pi^i and pi^(i-1)}
    S^{i-1}_{k+1} = \begin{cases}
        S^{i}_{k+1} + \frac{1}{p_i} & \text{if } k \in \lb 0, j^{*} -1 \rb,\\
        S^{i}_{k}  & \text{if } k \in \lb j^{*}, n-i+1 \rb.
    \end{cases}
\end{align}

\end{lemma}

\begin{proof}{\emph{Proof of Lemma \ref{lem: relation pi_(i-1) and pi_(i)}.}} Let $i\in \lb 1,n-1\rb$ and $j^*\in \lb 1,n-i+1\rb$ be such that $\pi^{i-1}(j^{*})=i$. The permutation $\pi^{i-1}$ sorts locations $\lb i,n\rb$ in order of non-decreasing detection rates. After removing $p_{i} = p_{\pi^{i-1}(j^*)}$ from the chain of inequalities, we obtain $p_{\pi^{i-1}(1)} \leq \cdots \leq p_{\pi^{i-1}(j^*-1)} \leq p_{\pi^{i-1}(j^*+1)} \leq \cdots \leq p_{\pi^{i-1}(n-i+1)}$, which sorts $\lb i+1,n\rb$ by detection rates, thus providing \eqref{eq: relation pi^i and pi^(i-1)}. 

As a consequence, for every $k \in \lb 0,j^*-1\rb$,
\begin{align*}
&\sum_{j=1}^k \cap_{\pi^{i-1}(j)} = \sum_{j=1}^k \cap_{\pi^{i}(j)},\\
&S_{k+1}^{i-1} = \sum_{j=k+1}^{j^*-1} \frac{1}{p_{\pi^{i-1}(j)}} + \frac{1}{p_i} + \sum_{j=j^*+1}^{n-i+1} \frac{1}{p_{\pi^{i-1}(j)}} = \frac{1}{p_i} + \sum_{j=k+1}^{j^*-1} \frac{1}{p_{\pi^{i}(j)}} + \sum_{j=j^*+1}^{n-i+1} \frac{1}{p_{\pi^{i}(j-1)}} =  \frac{1}{p_i} + S_{k+1}^{i}.
\end{align*}
Similarly, for every $k \in \lb j^*,n-i+1\rb$,
\begin{align*}
&\sum_{j=1}^k \cap_{\pi^{i-1}(j)} = \sum_{j=1}^{j^*-1} \cap_{\pi^{i}(j)} +  \cap_{i} + \sum_{j=j^*+1}^k \cap_{\pi^{i}(j-1)} = \sum_{j=1}^{k-1} \cap_{\pi^{i}(j)} + \cap_{i},\\
&S_{k+1}^{i-1} = \sum_{j=k+1}^{n-i+1} \frac{1}{p_{\pi^{i-1}(j)}} = \sum_{j=k+1}^{n-i+1} \frac{1}{p_{\pi^{i}(j-1)}}  =  S_{k}^{i}.
\end{align*}
%
\hfill \Halmos
\end{proof}

\begin{lemma}\label{lemma:ki}
For every $i \in \lb 0,n-1\rb$, $k_i$ exists and
$$K_i \coloneqq \left\{ k \in \lb 0,n-i\rb:\, k + p_{\pi^{i}(k)}S^{i}_{k+1} < \rs \right\} = \lb 0,k_i\rb.$$

Furthermore, for every $i\in \lb 1,n-1\rb$ and $j^*\in \lb1,n-i+1\rb$ such that $\pi^{i-1}(j^{*})=i$,
\begin{align}
\OneAndAHalfSpacedXI
\label{ineq: relation k(i-1) and k(i)}
    k_{i-1} \leq \begin{cases}
        k_{i} & \text{if } k_{i-1} \in \lb 0, j^*-1\rb,\\
        k_{i} + 1 & \text{if } k_{i-1} \in \lb j^*, n-i +1 \rb.
    \end{cases} 
\end{align}

\end{lemma}

\begin{proof}{\emph{Proof of Lemma \ref{lemma:ki}.}} Let $i \in \lb 0, n-1\rb$. Since $p_{\pi^i(0)} S_1^i = 0 < \rs$, then  $0 \in K_i$ and $k_i$ exists. Next, we consider $k \in \lb 1, n-i\rb$. If $k \in K_i$, then $k-1 \in K_i$, as shown below:
\begin{align*}
\rs > k + p_{\pi^i(k)} S_{k+1}^i = k-1 + p_{\pi^i(k)} \left(S_{k+1}^i + \frac{1}{p_{\pi^i(k)}}\right) \geq k-1 + p_{\pi^i(k-1)} S_k^i,
\end{align*}
where we used the fact that $k \in K_i$ and $\pi^i$ orders locations in $\lb i+1,n\rb$ by their detection rates. Therefore, $K_i = \lb 0,k_i\rb$.

We next analyze $k_i$ as a function of $i$. Let $i\in \lb 1, n-1\rb$ and $j^*\in \lb 1, n-i+1\rb$ be such that $\pi^{i-1}(j^{*})=i$. If $k_{i-1} \in \lb 0, j^*-1\rb$, we obtain:
\begin{align*}
\rs > k_{i-1} + p_{\pi^{i-1}(k_{i-1})} S_{k_{i-1}+1}^{i-1} \overset{\eqref{eq: relation pi^i and pi^(i-1)},\eqref{eq: relation sum probs pi^i and pi^(i-1)}}{\geq} k_{i-1} + p_{\pi^{i}(k_{i-1})} S_{k_{i-1}+1}^{i}.
\end{align*}
Since $k_{i-1} \leq j^* - 1 \leq n-i$, we deduce that $k_{i-1} \in K_i$, and $k_{i-1} \leq k_i$.  

If $k_{i-1} \in \lb j^*+1, n-i+1\rb$, we obtain:
\begin{align*}
\rs  \overset{\eqref{eq: relation pi^i and pi^(i-1)},\eqref{eq: relation sum probs pi^i and pi^(i-1)}}{>} k_{i-1} + p_{\pi^{i}(k_{i-1}-1)} S_{k_{i-1}}^{i} > k_{i-1} -1 + p_{\pi^{i}(k_{i-1}-1)} S_{k_{i-1}-1 +1}^{i}.
\end{align*}
Since $k_{i-1} -1 \leq n-i$, then $k_{i-1} -1 \in K_i$ and $k_{i-1} -1 \leq k_i$. 

Finally, if $k_{i-1} = j^*$, we obtain:
\begin{align*}
\rs \overset{\eqref{eq: relation sum probs pi^i and pi^(i-1)}}{>} k_{i-1} + p_{\pi^{i-1}(k_{i-1}-1)} S_{k_{i-1}}^{i}  \overset{\eqref{eq: relation pi^i and pi^(i-1)}}{>}  k_{i-1} -1 + p_{\pi^{i}(k_{i-1}-1)} S_{k_{i-1}-1 +1}^{i}.
\end{align*}
Thus, $k_{i-1} -1 \in K_i$ and $k_{i-1} -1 \leq k_i$. Note that we used throughout that $p_{\pi^i(0)}=0$. 
\hfill \Halmos
\end{proof}

\begin{lemma}
\label{lem: relation ell_(i-1) and ell_(i)}
For every $i \in \lb 0,n-1\rb$, if $\rh > \sum_{j=1}^i \cap_j + p_i c_i S_1^i$, then $\ell_i$ exists and
$$L_i \coloneqq \left\{ \ell \in \lb 0, n-i\rb:\,   \sum_{j=1}^{i} c_{j} + \sum_{j=1}^{\ell} c_{\pi^{i}(j)} +  p_{i}c_{i}S^{i}_{\ell+1} < \rh \right\} = \lb 0,\ell_i\rb.$$

\end{lemma}

\begin{proof}{\emph{Proof of Lemma \ref{lem: relation ell_(i-1) and ell_(i)}.}}
Let $i \in \lb 0,n-1\rb$ and suppose that $\rh > \sum_{j=1}^i \cap_j + p_i c_i S_1^i$. Then, $0 \in L_i$ and $\ell_i$ exists.
Next, consider $\ell \in \lb 1, n-i\rb$. If $\ell \in L_i$, then $\ell-1 \in L_i$, as shown below:
\begin{align*}
\rh - \sum_{j=1}^i c_j >   \sum_{j=1}^{\ell-1} c_{\pi^i(j)} + \frac{p_{\pi^i(\ell)}}{p_{\pi^i(\ell)}} c_{\pi^i(\ell)} + p_ic_i S_{\ell+1}^i \geq \sum_{j=1}^{\ell-1} c_{\pi^i(j)} + p_ic_i S_{\ell}^i,
\end{align*}
where we used the fact that $\ell \in L_i$ and locations in $\lb 1,n\rb$ are ordered by their detection potentials. Therefore, $L_i = \lb 0,\ell_i\rb$.
%
%
\hfill \Halmos
\end{proof}


We are now ready to prove Lemma \ref{lem: threshold inequalities}.


\begin{proof}{\emph{Proof of Lemma \ref{lem: threshold inequalities}.}}
First, $\tau_{-1} = 0$ by definition. Moreover, we observe that $k_{n-1} \in \{0,1\}$, which implies that $\tau_{n-1} = m$: If $k_{n-1} = 0$, then $\tau_{n-1} = \sum_{j=1}^{n-1} c_j + p_nc_n/p_n = m$. If $k_{n-1} = 1$, then $\tau_{n-1} = \sum_{j=1}^{n-1} c_j + c_n = m$.
We next show that $\tau_{i-1} \leq \nu_{i} \leq \tau_{i}$ for all $i\in \lb 0,n-1\rb$.

We note that the inequality $\nu_{i} \leq \tau_{i}$ follows directly from the fact that $p_ic_i \leq p_{i+1}c_{i+1}$. Thus, it only remains to show that $\tau_{i-1} \leq \nu_{i}$. This is trivial for $i=0$, so we may assume that $i\in \lb 1,n-1\rb$. We note that
\begin{align}
\label{eq: nu_i - tau_i}
    \nu_i - \tau_{i-1} &= c_i + \sum_{j=1}^{k_{i}}c_{\pi^{i}(j)} - \sum_{j=1}^{k_{i-1}}c_{\pi^{i-1}(j)} - p_ic_i\left( S^{i-1}_{k_{i-1}+1} - S^{i}_{k_{i}+1} \right).
\end{align}
Let $j^*\in \lb 1,n-i+1\rb$  be such that $\pi^{i-1}(j^*)=i$. If $k_{i-1} \in \lb 0,j^*-1\rb$, we obtain:
\begin{align*}
  \nu_i - \tau_{i-1} &\overset{\eqref{eq: relation sum cardinalities pi^i and pi^(i-1)},\eqref{eq: relation sum probs pi^i and pi^(i-1)}}{=} c_i + \sum_{j=1}^{k_{i}}c_{\pi^{i}(j)} - \sum_{j=1}^{k_{i-1}}c_{\pi^{i}(j)}  - p_ic_i\left(S^{i}_{k_{i-1}+1} - S^{i}_{k_{i}+1} + \frac{1}{p_i}\right)\\
  & \overset{\eqref{ineq: relation k(i-1) and k(i)}}{=} \sum_{j=k_{i-1}+1}^{k_i} \frac{p_{\pi^{i}(j)}c_{\pi^{i}(j)} - p_ic_i}{p_{\pi^{i}(j)}}  \geq 0.
\end{align*}

If $k_{i-1} \in \lb j^*, n-i+1\rb$, we obtain:
\begin{align*}
  \nu_i - \tau_{i-1} &\overset{\eqref{eq: relation sum cardinalities pi^i and pi^(i-1)},\eqref{eq: relation sum probs pi^i and pi^(i-1)}}{=} \sum_{j=1}^{k_{i}}c_{\pi^{i}(j)} - \sum_{j=1}^{k_{i-1}-1}c_{\pi^{i}(j)}  - p_ic_i\left(S^{i}_{k_{i-1}} - S^{i}_{k_{i}+1}\right)  \overset{\eqref{ineq: relation k(i-1) and k(i)}}{=} \sum_{j=k_{i-1}}^{k_i} \frac{p_{\pi^{i}(j)}c_{\pi^{i}(j)} - p_ic_i}{p_{\pi^{i}(j)}}  \geq 0.
\end{align*}

\hfill \Halmos
\end{proof}

The following lemma derives properties satisfied by our auxiliary parameters:
\begin{lemma}
\label{lem:ki_vs_li}
Let $i\in \lb 0,n-1\rb$ be such that $\rh > \tau_{i-1}$. Then, the following statements hold:
\begin{itemize}
\item[--] $\ell_i$ exists. Furthermore, when $i \geq 1$, let $j^* \in \lb 1,n-i+1\rb$ satisfying $\pi^{i-1}(j^*) = i$. Then,
\begin{align}
\OneAndAHalfSpacedXI
     k_{i-1} \leq \begin{cases}
        \ell_i  & \text{if } k_{i-1} \in \lb 0,j^*-1 \rb \\
        \ell_i + 1 & \text{if } k_{i-1} \in \lb j^*,n-i +1\rb.
    \end{cases}\label{eq:ki-1_vs_li}
\end{align}

\item[--] If $\nu_i < \rh$, then $k_i \leq \ell_i$. If $\rh \leq \nu_i$, then $k_i > \ell_i$.

\end{itemize}

\end{lemma}

\begin{proof}{\emph{Proof of Lemma~\ref{lem:ki_vs_li}.}}

\begin{itemize}
\item[--]
Let $i\in \lb 0,n-1\rb$. If $i = 0$ and $\rh > \tau_{-1} = 0$, then $0 \in L_0$ and $\ell_0$ exists. We now assume that $i \in \lb 1,n-1\rb$ and $\rh > \tau_{i-1}$. Let $j^* \in \lb 1,n-i+1\rb$ be such that $\pi^{i-1}(j^*) = i$. If $k_{i-1} \in \lb 0,j^*-1\rb$, then:
\begin{align*}
\rh > \tau_{i-1} \overset{\eqref{eq: relation sum cardinalities pi^i and pi^(i-1)},\eqref{eq: relation sum probs pi^i and pi^(i-1)}}{=} \sum_{j=1}^{i} c_j + \sum_{j=1}^{k_{i-1}} c_{\pi^{i}(j)} + p_i c_i S_{k_{i-1}+1}^{i}.
\end{align*}
Since $k_{i-1} \leq j^* - 1 \leq n-i$, then $k_{i-1} \in L_i$, $\ell_i$ exists, and $k_{i-1} \leq \ell_i$. 

If $k_{i-1} \in \lb j^*,n-i+1\rb$, then:
\begin{align*}
\rh \overset{\eqref{eq: relation sum cardinalities pi^i and pi^(i-1)},\eqref{eq: relation sum probs pi^i and pi^(i-1)}}{>} \sum_{j=1}^{i} c_j + \sum_{j=1}^{k_{i-1}-1} c_{\pi^{i}(j)} + p_i c_i S_{k_{i-1}-1+1}^{i}.
\end{align*}
Since $k_{i-1} -1 \leq n-i$, then $k_{i-1} -1 \in L_i$, $\ell_i$ exists, and $k_{i-1} -1 \leq \ell_i$.

\item[--] Let $i\in \lb 0,n-1\rb$ be such that $\rh > \tau_{i-1}$. If $\rh > \nu_i$, then by definition of $\nu_i$, $k_i \in L_i$ and $k_i \leq \ell_i$. On the other hand, if $\rh \leq \nu_i$, then $k_i \notin L_i$. Since $\ell_i$ exists and $L_i = \lb 0,\ell_i\rb$ (Lemma \ref{lem: relation ell_(i-1) and ell_(i)}), then $k_i > \ell_i$. \hfill \Halmos

\end{itemize}

\end{proof}



We can now prove the first theorem of this article.

\begin{proof}{\emph{Proof of Theorem \ref{thm: NE disjoint monitoring sets}.}}
Let $\rs\in \lb 1,n-1\rb$ and  $\rh \in \lb 1,m-1\rb$. Let $i^{*}\in \lb 0, n-1\rb$ satisfying $\tau_{i^{*}-1} < \rh \leq \tau_{i^{*}}$.

\begin{enumerate}
\item[\emph{Regime Pattern 1:}] Suppose that $\nu_{i^*} < \rh \leq \tau_{i^*}$. From Lemma \ref{lem:ki_vs_li}, we know that $k_{i^*} \leq \ell_{i^*}$. Let $\smargv^{\Seek^*} \in \R^n$ and $\hmargv^{\Hide^*} \in \R^n$ satisfying \eqref{Regime_1_S} and \eqref{Regime_1_H}, respectively. We will show that $(\smargv^{\Seek^*},\hmargv^{\Hide^*}) \in \widetilde{\A}_\Seek \times \widetilde{\A}_\Hide$ and is a NE of $\widetilde{\Gamma}$.

First, we note that $k_{i^*} < n-i^*$. Indeed, if $k_{i^*} = n-i^*$, then $\rh > \nu_{i^*} = \sum_{j=1}^n c_j = m$, which is a contradiction. Therefore, $\K \neq \emptyset$ and $S_{k_{i^*}+1}^{i^*} > 0$. By definition of $k_{i^*}$, we obtain 
\begin{align}
\rs > k_{i^*} + p_{\pi^{i^*}(k_{i^*})}S_{k_{i^*}+1}^{i^*} \geq k_{i^*},\label{used_for_knapsack_4}
 \end{align}
 which implies that $\smargv_i^{\Seek^*} \geq 0$ for every $i \in \K$. Furthermore, since $k_{i^*}+1 \leq n-i^*$ and $k_{i^*}+1\notin K_{i^*}$, then:
\begin{align}
\rs \leq k_{i^*}+1 + p_{\pi^{i^*}(k_{i^*}+1)}\left(S_{k_{i^*}+1}^{i^*} - \frac{1}{p_{\pi^{i^*}(k_{i^*}+1)}}\right) = k_{i^*} + p_{\pi^{i^*}(k_{i^*}+1)}S_{k_{i^*}+1}^{i^*}.\label{used_for_knapsack_3}
\end{align}
Thus, for every $i \in \K$, $\smargv_i^{\Seek^*} \leq 1$. Finally,
$$\sum_{i=1}^n\smargv_i^{\Seek^*} = |\J| + \frac{\rs - k_{i^*}}{S_{k_{i^*}+1}^{i^*}}S_{k_{i^*}+1}^{i^*} = \rs.$$ 
Therefore, $\smargv^{\Seek^*} \in \widetilde{\A}_\Seek$.
Next, we show that $\hmargv^{\Hide^*} \in \widetilde{\A}_\Hide$. Since $k_{i^*} \leq \ell_{i^*}$, then $k_{i^*} \in L_{i^*}$ and
\begin{align}
\rh >  \sum_{j=1}^{i^{*}}c_{j} + \sum_{j=1}^{k_{i^{*}}}c_{\pi^{i^{*}}(j)} + p_{i^*} c_{i^*} S^{i^{*}}_{k_{i^{*}}+1} \geq  \sum_{j=1}^{i^{*}}c_{j} + \sum_{j=1}^{k_{i^{*}}}c_{\pi^{i^{*}}(j)}.\label{used_for_knapsack_1}
\end{align}
Thus, $\hmargv_i^{\Hide^*} \geq 0$ for every $i \in \K$. Furthermore, since $\rh \leq \tau_{i^*}$ and $p_{i^*+1}c_{i^*+1} \leq p_i c_i$ for every $i \in \J\cup\K$, we obtain:
\begin{align}
\forall \, i \in \K, \ \hmargv_i^{\Hide^*} = \frac{\rh - \sum_{j=1}^{i^*}c_j - \sum_{j=1}^{k_{i^*}} c_{\pi^{i^*}(j)} }{p_iS_{k_{i^*}+1}^{i^*}} \leq \frac{p_{i^*+1}c_{i^*+1}}{p_i} \leq c_i.\label{used_for_knapsack_2}
\end{align} 

Finally,
$$\sum_{i=1}^n \hmargv_i^{\Hide^*} = \sum_{i=1}^{i^*}c_i + \sum_{j=1}^{k_{i^*}} c_{\pi^{i^*}(j)} + \left(\rh - \sum_{j=1}^{i^*}c_j - \sum_{j=1}^{k_{i^*}} c_{\pi^{i^*}(j)} \right) \frac{S_{k_{i^*}+1}^{i^*}}{S_{k_{i^*}+1}^{i^*}}  = \rh.$$
Therefore, $\hmargv^{\Hide^*} \in \widetilde{\A}_\Hide$.

Next, we show that $(\smargv^{\Seek^*},\hmargv^{\Hide^*})$ is a NE of $\widetilde{\Gamma}$. To this end, we first note the following: 
\begin{alignat}{3}
\forall \, \hmargv^{\Hide} \in \widetilde{\A}_\Hide, \ \min_{\smargv^{\Seek} \in \widetilde{\A}_\Seek} \tilde{u}(\smargv^{\Seek},\hmargv^{\Hide}) \ \overset{\eqref{continuous_payoff}}{=} \ \sum_{i=1}^n \hmargv_i^{\Hide} \ - \ \max_{\smargv^{\Seek}} \ & \sum_{i=1}^n p_i\hmargv_i^{\Hide} \smargv_i^{\Seek}\nonumber\\
\text{s.t.} \ & \sum_{i=1}^n \smargv_i^{\Seek} \leq \rs \label{knapsack_1}\\
& 0 \leq \smargv_i^{\Seek} \leq 1, && \quad \forall \, i \in \lb 1,n\rb.\nonumber
\end{alignat}

Thus, a best response to $\hmargv^{\Hide} \in \widetilde{\A}_\Hide$ is an optimal solution to a continuous knapsack problem with $n$ different (fractional) objects of unitary weights and a knapsack capacity equal to $\rs$. Each object $i \in \lb 1,n\rb$ has a profit equal to $p_i \smargv_i^{\Hide}$. An optimal solution consists in filling the capacity of the knapsack with the objects with highest profits.

Then, given $\hmargv^{\Hide^*}$ satisfying \eqref{Regime_1_H}, the ``profit'' of each object is given by:
\begin{align}
\OneAndAHalfSpacedXI
 p_i\hmargv_i^{\Hide^*} &= 
            \begin{dcases}
                p_ic_i &\text{if } i\in \I \cup \J,\\
                \frac{\rh - \sum_{j=1}^{i^{*}}c_{j} - \sum_{j=1}^{k_{i^{*}}}c_{\pi^{i^{*}}(j)} }{S^{i^{*}}_{k_{i^{*}}+1}} & \text{if } i\in \K.
            \end{dcases}\label{Profits_Regime_1_H}
\end{align} 

We know that for every $i \in \I$ and every $l \in \J$:
\begin{align}
p_i c_i\leq p_{i^*}c_{i^*} \overset{\eqref{used_for_knapsack_1}}{<} \frac{\rh - \sum_{j=1}^{i^{*}}c_{j} - \sum_{j=1}^{k_{i^{*}}}c_{\pi^{i^{*}}(j)} }{S^{i^{*}}_{k_{i^{*}}+1}} \overset{\eqref{used_for_knapsack_2}}{\leq } p_{i^*+1}c_{i^*+1} \leq p_l c_l.\label{Profits_Ineq_Regime_1_H}
\end{align}
Therefore, the objects in $\J$ are the most profitable, followed by the objects in $\K$ that have equal profit, followed by the objects in $\I$. We now must determine bounds on $\rs$. We showed in \eqref{used_for_knapsack_4} that $\rs> k_{i^*} = |\J|$. An upper bound is given as follows:
\begin{align*}
\rs \overset{\eqref{used_for_knapsack_3}}{\leq} k_{i^*}+1 + \sum_{j=k_{i^*}+2}^{n-i^*} \frac{p_{\pi^{i^*}(k_{i^*}+1)}}{p_{\pi^{i^*}(j)}} \leq n-i^* = |\J| + |\K|.
\end{align*}

Thus, one best response to $\hmargv^{\Hide^*}$ selects all the objects in $\J$ and any fraction of the objects in $\K$ until the knapsack is full. Hence, $\smargv^{\Seek^*}$ defined in \eqref{Regime_1_S} is a best response to $\hmargv^{\Hide^*}$.

To show that $\hmargv^{\Hide^*}$ is a best response to $\smargv^{\Seek^*}$, we similarly observe the following:
\begin{alignat}{3}
\forall \, \smargv^{\Seek} \in \widetilde{\A}_\Seek, \ \max_{\hmargv^{\Hide} \in \widetilde{\A}_\Hide} \tilde{u}(\smargv^{\Seek},\hmargv^{\Hide}) \overset{\eqref{continuous_payoff}}{=} \max_{\hmargv^{\Hide}} \ & \sum_{i=1}^n (1- p_i \smargv_i^{\Seek})\hmargv_i^{\Hide}\nonumber\\
\text{s.t.} \ & \sum_{i=1}^n \hmargv_i^{\Hide} \leq \rh\label{knapsack_2}\\
& 0 \leq \hmargv_i^{\Hide} \leq \cap_i, && \forall \, i \in \lb 1,n\rb.\nonumber
\end{alignat}

Thus, a best response to $\smargv^{\Seek} \in \widetilde{\A}_\Seek$ is an optimal solution to another continuous knapsack problem with $n$ different (fractional) objects of unitary weights and a knapsack capacity equal to $\rh$. Each object $i \in \lb 1,n\rb$ is available $\cap_i$ times and has a profit equal to $(1 - p_i \smargv_i^{\Seek})$. An optimal solution consists in selecting as many copies as possible of the objects with highest profits until filling the capacity of the knapsack.

Then, given $\smargv^{\Seek^*}$ satisfying \eqref{Regime_1_S}, the ``profit'' of each object is given by:
\begin{align}
\OneAndAHalfSpacedXI
 1-p_i\smargv_i^{\Seek^*} &=  
            \begin{dcases}
            1 & \text{if }  i\in \I,\\
                1 - p_i &\text{if } i\in  \J,\\
               1 -  \frac{\rs - k_{i^*} }{S^{i^{*}}_{k_{i^{*}}+1}} & \text{if } i\in \K.
            \end{dcases}\label{Profits_Regime_1_S}
\end{align} 

We have the following inequalities:
\begin{align}
\forall \, i \in \J, \ 1 - \frac{\rs - k_{i^*} }{S^{i^{*}}_{k_{i^{*}}+1}} \overset{\eqref{used_for_knapsack_4}}{<}  1 - p_{\pi^{i^*}(k_{i^*})} \leq1 - p_i < 1.\label{Profits_Ineq_Regime_1_S}
\end{align}

Therefore, the objects in $\I$ are the most profitable, followed by the objects in $\J$, followed by the objects in $\K$ that have equal profit. To determine which objects will be selected given $\rh$, we recall that $\rh < m = \sum_{i=1}^n \cap_i$. Furthermore,  \eqref{used_for_knapsack_1} implies that $\rh > \sum_{i \in \I\cup\J} \cap_i$. 

Thus, one best response to $\smargv^{\Seek^*}$ consists in selecting all copies of the objects in $\I$ and $\J$, and in selecting any fraction of the objects in $\K$ until the knapsack is full. Hence, $\hmargv^{\Hide^*}$ defined in \eqref{Regime_1_H} is a best response to $\smargv^{\Seek^*}$.


As a consequence, $(\smargv^{\Seek^*},\hmargv^{\Hide^*})$ is a NE of $\widetilde{\Gamma}$. From Proposition \ref{prop:equivalence}, we deduce that any strategy profile $(\sigma^{\Seek^*},\sigma^{\Hide^*})\in \Delta_{\Seek} \times \Delta_{\Hide}$ that satisfies $\smargv(\sigma^{\Seek^*}) = \smargv^{\Seek^*}$ and $\hmargv(\sigma^{\Hide^*}) = \hmargv^{\Hide^*}$ (which exists as a consequence of Lemma \ref{lemma:convex_hull}) is a NE of $\Gamma$. 


Furthermore, the value of the games $\Gamma$ and $\widetilde{\Gamma}$ is given by:
\begin{align*}
U(\sigma^{\Seek^*},\sigma^{\Hide^*}) = \tilde{u}(\smargv^{\Seek^*},\hmargv^{\Hide^*})  = \rh  - \sum_{j=1}^{k_{i^{*}}} p_{\pi^{i^{*}}(j)}c_{\pi^{i^{*}}(j)} - \frac{ \Big(\rs - k_{i^*}  \Big) \left( \rh - \sum_{j=1}^{i^{*}}c_{j} - \sum_{j=1}^{k_{i^{*}}}c_{\pi^{i^{*}}(j)} \right)}{S^{i^*}_{k_{i^*}+1}}.
\end{align*}

\item[\emph{Regime Pattern 2:}] We now consider the case when $i^* = 0$ and $\tau_{-1} < \rh \leq \nu_{0}$. From Lemma \ref{lem:ki_vs_li}, we know that $k_{0} > \ell_{0}$. Let $\smargv^{\Seek^*} \in \R^n$ and $\hmargv^{\Hide^*} \in \R^n$ satisfying \eqref{Regime_2_S} and \eqref{Regime_2_H}, respectively. We will analogously show that $(\smargv^{\Seek^*},\hmargv^{\Hide^*}) \in \widetilde{\A}_\Seek \times \widetilde{\A}_\Hide$ and is a NE of $\widetilde{\Gamma}$.

First, we note that $\ell_{0} < k_{0} \leq n$, which implies that $\ell_{0} +1 \leq n$. Thus, $S_{\ell_{0}+2}^{0}$ is well defined and $\K \neq \emptyset$. For every $i \in \K\setminus\{\pi^0(\ell_0+1)\}$, $p_{\pi^{0}(\ell_{0}+1)} \leq p_{i}$. Furthermore, since $\ell_{0} +1 \in K_{0}$, we obtain:
\begin{align}
\sum_{i=1}^n \smargv_i^{\Seek^*} = \ell_0+1 +  p_{\pi^0(\ell_0 +1)}S_{\ell_0+2}^0 < \rs. \label{for_knapsack_regime_2_1}
\end{align}

Thus, $\smargv^{\Seek^*} \in \widetilde{\A}_\Seek.$ Next, we show that $\hmargv^{\Hide^*} \in \widetilde{\A}_\Hide$. Since $\ell_0 \in L_0$, $\ell_0+1 \notin L_0$, and $\ell_0+1 \leq n$, then:
\begin{align}
0 < \rh - \sum_{j=1}^{\ell_0} c_{\pi^0(j)} = \hmargv_{\pi^0(\ell_0+1)}^{\Hide^*} = \rh - \sum_{j=1}^{\ell_0+1} c_{\pi^0(j)} + c_{\pi^0(\ell_0+1)} \leq c_{\pi^0(\ell_0+1)}.\label{for_knapsack_regime_2_2}
\end{align}

Since $\sum_{i=1}^n \hmargv_{i}^{\Hide^*} = \rh$, we can then conclude that $\hmargv^{\Hide^*} \in \widetilde{\A}_\Hide$.

Next, we show that $(\smargv^{\Seek^*},\hmargv^{\Hide^*})$ is a NE of $\widetilde{\Gamma}$. Given $\hmargv^{\Hide^*}$ satisfying \eqref{Regime_2_H}, the profit of each object in the knapsack problem \eqref{knapsack_1} is given by:
\begin{align}
\OneAndAHalfSpacedXI
 p_i\hmargv_i^{\Hide^*} &= 
            \begin{dcases}
                p_ic_i &\text{if } i\in \J,\\
                 p_{\pi^0(\ell_0+1)}\left(\rh - \sum_{j=1}^{\ell_0}c_{\pi^0(j)}\right) & \text{if } i = \pi^0(\ell_0+1),\\
                0 & \text{if } i\in \K \setminus\{\pi^0(\ell_0+1)\}.\\
            \end{dcases}\label{Profits_Regime_2_H}
\end{align} 

Since $\rs \overset{\eqref{for_knapsack_regime_2_1}}{>} \ell_0 +1 = |\J|+1$, then a best response to $\hmargv^{\Hide^*}$ will select all the objects in $\J \cup \{\pi^0(\ell_0+1)\}$ (and might not entirely fill the knapsack). Hence, $\smargv^{\Seek^*}$ defined in \eqref{Regime_2_S} is a best response to $\hmargv^{\Hide^*}$.
Then, given $\smargv^{\Seek^*}$ satisfying \eqref{Regime_2_S}, the profit of each object in the knapsack problem \eqref{knapsack_2} is given by:
\begin{align*}
\OneAndAHalfSpacedXI
 1-p_i\smargv_i^{\Seek^*} &=  
            \begin{dcases}
                1 - p_i &\text{if } i\in  \J\cup\{\pi^0(\ell_0+1)\},\\
               1 -  p_{\pi^0(\ell_0+1)} & \text{if } i\in \K\setminus\{\pi^0(\ell_0+1)\}.
            \end{dcases}
\end{align*} 

By definition of $\pi^0$, we have the following inequalities: $1 - p_{\pi^0(1)} \geq \cdots \geq 1 - p_{\pi^0(\ell_0+1)}$. Since \eqref{for_knapsack_regime_2_2} implies that $\rh > \sum_{j=1}^{\ell_0} c_{\pi^0(j)}$, then one best response to $\smargv^{\Seek^*}$ can select all copies of the objects in $\J$ and can select any fraction of the objects in $\K$ until the knapsack is full. Hence, $\hmargv^{\Hide^*}$ defined in \eqref{Regime_2_H} is a best response to $\smargv^{\Seek^*}$. 

Thus, $(\smargv^{\Seek^*},\hmargv^{\Hide^*})$ is a NE of $\widetilde{\Gamma}$. From Proposition \ref{prop:equivalence}, we deduce that any strategy profile $(\sigma^{\Seek^*},\sigma^{\Hide^*})\in \Delta_{\Seek} \times \Delta_{\Hide}$ that satisfies $\smargv(\sigma^{\Seek^*}) = \smargv^{\Seek^*}$ and $\hmargv(\sigma^{\Hide^*}) = \hmargv^{\Hide^*}$ is a NE of $\Gamma$. Furthermore, the value of the games $\Gamma$ and $\widetilde{\Gamma}$ is given by:
\begin{align*}
U(\sigma^{\Seek^*},\sigma^{\Hide^*}) = \tilde{u}(\smargv^{\Seek^*},\hmargv^{\Hide^*})  = \rh - \sum_{j=1}^{\ell_{0}} p_{\pi^{0}(j)}c_{\pi^{0}(j)} - p_{\pi^{0}(\ell_{0}+1)} \left( \rh - \sum_{j=1}^{\ell_{0}}c_{\pi^{0}(j)}\right).
\end{align*}

\item[\emph{Regime Pattern 3:}] Finally, we consider the case when $i^* \geq 1$ and $\tau_{i^*-1} < \rh \leq \nu_{i^*}$. From Lemma \ref{lem:ki_vs_li}, we know that $k_{i^*} > \ell_{i^*}$. Let $\smargv^{\Seek^*} \in \R^n$ and $\hmargv^{\Hide^*} \in \R^n$ satisfying \eqref{Regime_3_S} and \eqref{Regime_3_H}, respectively. We will analogously show that $(\smargv^{\Seek^*},\hmargv^{\Hide^*}) \in \widetilde{\A}_\Seek \times \widetilde{\A}_\Hide$ and is a NE of $\widetilde{\Gamma}$.

Similarly, we note that $\ell_{i^*} < k_{i^*} \leq n-i^*$, which implies that $\ell_{i^*} +1 \leq n-i^*$. Thus, $S_{\ell_{i^*}+2}^{i^*}$ is well defined and $\K \neq \emptyset$. 
Since $\ell_{i^*} +1 \in K_{i^*}$, then
\begin{align}
\rs > p_{\pi^{i^*}(\ell_{i^*}+1)} S_{\ell_{i^*}+2}^{i^*} + \ell_{i^*} +1 = p_{\pi^{i^*}(\ell_{i^*}+1)} S_{\ell_{i^*}+1}^{i^*} + \ell_{i^*}.\label{for_knapsack_regime_3_0}
\end{align}
Thus, $\smargv_{i^*}^{\Seek^*}  > 0$. Next, we will show by contradiction the following upper bound:
\begin{align}
    \label{ineq: technical inequality regime k_i > ell_i}
    \rs - \ell_{i^*} - p_{\pi^{i^*}(\ell_{i^*}+1)} S^{i^*}_{\ell_{i^*} + 1} \leq \min\left\{ \frac{ p_{\pi^{i^*}(\ell_{i^*}+1)} }{  p_{i^*} },\,1 \right\}.
\end{align}

Let us assume that \eqref{ineq: technical inequality regime k_i > ell_i} does not hold, and let $j^* \in \lb 1,n-i^*+1\rb$ satisfying $\pi^{i^*-1}(j^*) = i^*$. If $\ell_{i^*}+1 \leq j^*-1$, then:
\begin{align*}
 \frac{ p_{\pi^{i^*}(\ell_{i^*}+1)} }{  p_{i^*} } &= \min\left\{ \frac{ p_{\pi^{i^*}(\ell_{i^*}+1)} }{  p_{i^*} },\,1 \right\} < \rs - \ell_{i^*} - 1 - p_{\pi^{i^*}(\ell_{i^*}+1)} S^{i^*}_{\ell_{i^*} + 2} \\
 &\overset{\eqref{eq: relation pi^i and pi^(i-1)},\eqref{eq: relation sum probs pi^i and pi^(i-1)}}{=} \rs - \ell_{i^*} - 1 - p_{\pi^{i^*-1}(\ell_{i^*}+1)} S^{i^*-1}_{\ell_{i^*} + 2} + \frac{p_{\pi^{i^*}(\ell_{i^*}+1)}}{p_{i^*}},
\end{align*}
which implies that $\ell_{i^*} + 1 \leq k_{i^*-1}$. However, by Lemma \ref{lem:ki_vs_li}, this can only occur when $k_{i^*-1} \geq j^*$, for which we obtain the following contradiction $j^* \leq k_{i^*-1} \overset{\eqref{eq:ki-1_vs_li}}{\leq} \ell_{i^*} + 1 \leq j^*-1$.

If on the other hand $\ell_{i^*}+1 \geq j^*$, then $j^* < \ell_{i^*}+2 \leq n - i^* +1$ and
\begin{align*}
1 &=\min\left\{ \frac{ p_{\pi^{i^*-1}(\ell_{i^*}+2)} }{  p_{i^*} },\,1 \right\}  = \min\left\{ \frac{ p_{\pi^{i^*}(\ell_{i^*}+1)} }{  p_{i^*} },\,1 \right\} \overset{\eqref{eq: relation pi^i and pi^(i-1)},\eqref{eq: relation sum probs pi^i and pi^(i-1)}}{<} \rs - \ell_{i^*} - 1 - p_{\pi^{i^*-1}(\ell_{i^*}+2)} S^{i^*-1}_{\ell_{i^*} + 3},
\end{align*}
which implies that $\ell_{i^*} + 2 \leq k_{i^*-1}$, thus contradicting \eqref{eq:ki-1_vs_li}. Therefore, \eqref{ineq: technical inequality regime k_i > ell_i} holds.
Finally,
$$\sum_{i=1}^n\smargv_i^{\Seek^*} = \rs - \ell_{i^*} - p_{\pi^*(\ell_{i^*}+1)}S_{\ell_{i^*}+1}^{i^*} + |\J| +1 + p_{\pi^*(\ell_{i^*}+1)}S_{\ell_{i^*}+2}^{i^*} = \rs,$$ 
which implies that $\smargv^{\Seek^*} \in \widetilde{\A}_\Seek$. 

Next, we show that $\hmargv^{\Hide^*} \in \widetilde{\A}_\Hide$. Since locations in $\lb 1,n\rb$ are ordered by their detection potentials, then for every $i \in \K, \ p_i c_i \geq p_{i^*}c_{i^*}$. Since $\ell_{i^*} \in L_{i^*}$, then:
\begin{align}
\hmargv_{\pi^{i^*}(\ell_{i^*}+1)}^{\Hide^*} = \rh - \sum_{j=1}^{i^*} c_j - \sum_{j=1}^{\ell_{i^*}} c_{\pi^{i^*}(j)} - p_{i^*}c_{i^*}S^{i^*}_{\ell_{i^*}+1} + \frac{p_{i^*}c_{i^*}}{p_{\pi^{i^*}(\ell_{i^*}+1)}} > \frac{p_{i^*}c_{i^*}}{p_{\pi^{i^*}(\ell_{i^*}+1)}} \geq 0.
\label{for_knapsack_regime_3_2}
\end{align}

Furthermore, since $\ell_{i^*}+1 \notin L_{i^*}$ and $\ell_{i^*}+1 \leq n - i^*$, then:
\begin{align}
\hmargv_{\pi^{i^*}(\ell_{i^*}+1)}^{\Hide^*} = \rh - \sum_{j=1}^{i^*} c_j - \sum_{j=1}^{\ell_{i^*}+1} c_{\pi^{i^*}(j)} - p_{i^*}c_{i^*}S^{i^*}_{\ell_{i^*}+2}  + c_{\pi^{i^*}(\ell_{i^*}+1)} \leq c_{\pi^{i^*}(\ell_{i^*}+1)}.\label{for_knapsack_regime_3_3}
\end{align}

Finally,
$$\sum_{i=1}^n \hmargv_i^{\Hide^*} = \sum_{i \in \I\cup \J} c_i  + \rh - \sum_{i=1}^{i^*} c_j - \sum_{j=1}^{\ell_{i^*}} c_{\pi^*(j)} - p_{i^*}c_{i^*}S_{\ell_{i^*}+2}^{i^*} + p_{i^*}c_{i^*} S_{\ell_{i^*}+2}^{i^*} = \rh.$$
Therefore, $\hmargv^{\Hide^*} \in \widetilde{\A}_\Hide$.

Next, we show that $(\smargv^{\Seek^*},\hmargv^{\Hide^*})$ is a NE of $\widetilde{\Gamma}$. Given $\hmargv^{\Hide^*}$ satisfying \eqref{Regime_3_H}, the profit of each object in the knapsack problem \eqref{knapsack_1} is given by:
\begin{align}
\OneAndAHalfSpacedXI
 p_i\hmargv_i^{\Hide^*} &= 
            \begin{dcases}
                p_ic_i &\text{if } i\in \I\cup \J,\\
                 p_{\pi^{i^*}(\ell_{i^*}+1)}\left( \rh - \hspace{-0.05cm}  \sum_{j=1}^{i^*}  c_{j}  - \hspace{-0.05cm}\sum_{j=1}^{\ell_{i^{*}}}  c_{\pi^{i^*}(j)} -  p_{i^*}c_{i^*}S^{i^*}_{\ell_{i^{*}}+2}\right) & \text{if } i = \pi^{i^*}(\ell_{i^*}+1),\\
                p_{i^*}c_{i^*} & \text{if } i\in \K \setminus\{\pi^{i^*}(\ell_{i^*}+1)\}.\\
            \end{dcases}\label{Profits_Regime_3_H}
\end{align} 

Since the locations in $\lb 1,n\rb$ are ordered by their detection potentials, then we know that $p_1c_1 \leq \cdots \leq p_{i^*} c_{i^*} \leq p_j c_j$ for every $j \in \J$. Furthermore, we have the following inequality:
\begin{align}
p_{i^*}c_{i^*} \overset{\eqref{for_knapsack_regime_3_2}}{<}  p_{\pi^{i^*}(\ell_{i^*}+1)}\left( \rh -  \sum_{j=1}^{i^*}  c_{j}  - \sum_{j=1}^{\ell_{i^{*}}}  c_{\pi^{i^*}(j)} -  p_{i^*}c_{i^*}S^{i^*}_{\ell_{i^{*}}+2}\right).\label{Profits_Ineq_Regime_3_H}
\end{align}
Thus, the objects in $\J \cup\{\pi^{i^*}(\ell_{i^*}+1)\}$ are the most profitable, followed by the objects in $\{i^*\}\cup\K\setminus\{\pi^{i^*}(\ell_{i^*}+1)\}$ that have equal profit, followed by the objects in $\I\setminus\{i^*\}$.

From \eqref{for_knapsack_regime_3_0}, we know that $\rs > \ell_{i^*} + 1 = |\J \cup\{\pi^{i^*}(\ell_{i^*}+1)\}|$. Furthermore, an upper bound is given as follows:
\begin{align*}
\rs \overset{\eqref{ineq: technical inequality regime k_i > ell_i}}{\leq}\ell_{i^*} + 1 +  \sum_{j=\ell_{i^*} + 1}^{n-i^*} \frac{p_{\pi^{i^*}(\ell_{i^*}+1)}}{p_{\pi^{i^*}(j)}} \leq n-i^*+ 1 =|\{i^*\}\cup\J\cup \K|.
\end{align*}

Therefore, one best response to $\hmargv^{\Hide^*}$ will select all the objects in $\J \cup \{\pi^{i^*}(\ell_{i^*}+1)\}$ and will select any fraction of the objects in $\{i^*\}\cup\K\setminus\{\pi^{i^*}(\ell_{i^*}+1)\}$ until the knapsack is full. Hence, $\smargv^{\Seek^*}$ defined in \eqref{Regime_3_S} is a best response to $\hmargv^{\Hide^*}$.

Then, given $\smargv^{\Seek^*}$ satisfying \eqref{Regime_3_S}, the profit of each object in the knapsack problem \eqref{knapsack_2} is given by:
\begin{align}
\OneAndAHalfSpacedXI
 1-p_i\smargv_i^{\Seek^*} &=  
            \begin{dcases}
            1 &\text{if } i\in\I\setminus\{i^*\}\\
            1 - p_{i^*}(\rs - \ell_{i^*} - p_{\pi^*(\ell_{i^*}+1)}S_{\ell_{i^*}+1}^{i^*}) &\text{if } i=i^*\\
                1 - p_i &\text{if } i\in  \J\cup\{\pi^{i^*}(\ell_{i^*}+1)\},\\
               1 -  p_{\pi^{i^*}(\ell_{i^*}+1)} & \text{if } i\in \K\setminus\{\pi^{i^*}(\ell_{i^*}+1)\}.
            \end{dcases}\label{Profits_Regime_3_S}
\end{align} 

By definition of $\pi^{i^*}$, we have the following inequalities: $1 - p_{\pi^{i^*}(1)} \geq \cdots \geq 1 - p_{\pi^{i^*}(\ell_{i^*}+1)}$. Furthermore, \eqref{ineq: technical inequality regime k_i > ell_i} implies that $ 1 - p_{i^*}(\rs - \ell_{i^*} - p_{\pi^*(\ell_{i^*}+1)}S_{\ell_{i^*}+1}^{i^*}) \geq 1 - p_{\pi^{i^*}(\ell_{i^*}+1)}$.

Since \eqref{for_knapsack_regime_3_2} implies that $\rh > \sum_{j=1}^{i^*}c_{j} + \sum_{j=1}^{\ell_{i^*}} c_{\pi^{i^*}(j)} $, then one best response to $\smargv^{\Seek^*}$ selects all copies of the objects in $\I\cup\J$ and selects any fraction of the objects in $\K$ until the knapsack is full. Hence, $\hmargv^{\Hide^*}$ defined in \eqref{Regime_3_H} is a best response to $\smargv^{\Seek^*}$.

Thus, $(\smargv^{\Seek^*},\hmargv^{\Hide^*})$ is a NE of $\widetilde{\Gamma}$. From Proposition \ref{prop:equivalence}, we deduce that any strategy profile $(\sigma^{\Seek^*},\sigma^{\Hide^*})\in \Delta_{\Seek} \times \Delta_{\Hide}$ that satisfies $\smargv(\sigma^{\Seek^*}) = \smargv^{\Seek^*}$ and $\hmargv(\sigma^{\Hide^*}) = \hmargv^{\Hide^*}$ is a NE of $\Gamma$. Furthermore, the value of the games $\Gamma$ and $\widetilde{\Gamma}$ is given by:
\begin{align*}
U(\sigma^{\Seek^*},\sigma^{\Hide^*}) = \tilde{u}(\smargv^{\Seek^*},\hmargv^{\Hide^*})  = &\rh - p_{i^{*}} \left(\rs - \ell_{i^*} -p_{\pi^{i^*}(\ell_{i^{*}}+1)}S^{i^{*}}_{\ell_{i^{*}}+1} \right)c_{i^{*}} - \sum_{j=1}^{\ell_{i^{*}}} p_{\pi^{i^{*}}(j)}c_{\pi^{i^{*}}(j)}\\
    &\quad - p_{\pi^{i^*}(\ell_{i^{*}}+1)} \left( \rh - \sum_{j=1}^{i^{*}}c_{j} - \sum_{j=1}^{\ell_{i^{*}}}c_{\pi^{i^{*}}(j)} \right).
\end{align*}

\hfill \Halmos

\end{enumerate}

\end{proof}

\begin{proposition}
\label{prop:iff_conditions}

Let $\rs\in \lb 1,n-1\rb$ and  $\rh \in \lb 1,m-1\rb$. Let $i^{*}\in \lb 0, n-1\rb$ satisfying $\tau_{i^{*}-1} < \rh \leq \tau_{i^{*}}$. 
\begin{enumerate}
 \item[Regime Pattern 1: $\nu_{i^*} < \rh \leq \tau_{i^*}$.] Suppose also that $\rh < \tau_{i^*}$ and $\rs < k_{i^*}+1 + p_{\pi^{i^*}(k_{i^*}+1)}S_{k_{i^*}+2}^{i^*}$. Then, \eqref{Regime_1_S}-\eqref{Regime_1_H} are necessary and sufficient conditions for a strategy profile $(\smargv^{\Seek^*},\hmargv^{\Hide^*}) \in \widetilde{\A}_\Seek \times \widetilde{\A}_\Hide$ to be a NE of $\widetilde{\Gamma}$.

  \item[Regime Pattern 2: $i^* = 0$ and $\tau_{-1} < \rh \leq \nu_{0}$.] Suppose also that $p_{\pi^0(\ell_0)}<p_{\pi^0(\ell_0+1)}< 1$ and $p_{\pi^0(\ell_0+1)}< p_{\pi^0(\ell_0+2)}$ if $\ell_0 \leq n-2$. Then, \eqref{Regime_2_S}-\eqref{Regime_2_H} are necessary and sufficient conditions for a strategy profile $(\smargv^{\Seek^*},\hmargv^{\Hide^*}) \in \widetilde{\A}_\Seek \times \widetilde{\A}_\Hide$ to be a NE of $\widetilde{\Gamma}$.

   \item[Regime Pattern 3: $i^* \geq 1$ and $\tau_{i^*-1} < \rh \leq \nu_{i^*}$.] Suppose also that $p_{i^*-1} c_{i^*-1} < p_{i^*} c_{i^*} < p_{i^*+1} c_{i^*+1}$, $\rh < \sum_{j=1}^{i^*} c_{j} + \sum_{j=1}^{\ell_{i^*}+1} c_{\pi^{i^*}(j)} +  p_{i^*}c_{i^*}S^{i^*}_{\ell_{i^*}+2}$, $\rs < k_{i^*-1}+1 + p_{\pi^{i^*-1}(k_{i^*-1}+1)}S^{i^*-1}_{k_{i^*-1}+2}$, $p_{\pi^{i^*}(\ell_{i^*})}<p_{\pi^{i^*}(\ell_{i^*}+1)}< 1$, and $p_{\pi^{i^*}(\ell_{i^*}+1)}< p_{\pi^{i^*}(\ell_{i^*}+2)}$ if  $\ell_{i^*} \leq n -i^*-2$. Then, \eqref{Regime_3_S}-\eqref{Regime_3_H} are necessary and sufficient conditions for a strategy profile $(\smargv^{\Seek^*},\hmargv^{\Hide^*}) \in \widetilde{\A}_\Seek \times \widetilde{\A}_\Hide$ to be a NE of $\widetilde{\Gamma}$.

 \end{enumerate}

\end{proposition}

\begin{proof}{\emph{Proof of Proposition \ref{prop:iff_conditions}.}}

\begin{enumerate}

\item[\emph{Regime Pattern 1:} $\nu_{i^*} < \rh \leq \tau_{i^*}$.] We additionally consider the following non-edge case assumptions: $\rh < \tau_{i^*}$ and $\rs < k_{i^*}+1 + p_{\pi^{i^*}(k_{i^*}+1)}S_{k_{i^*}+2}^{i^*}$.


Let $(\smargv^{\Seek^*},\hmargv^{\Hide^*}) \in \widetilde{\A}_\Seek \times \widetilde{\A}_\Hide$ satisfying \eqref{Regime_1_S}-\eqref{Regime_1_H} and let  $(\smargv^{\Seek^\prime},\hmargv^{\Hide^\prime}) \in \widetilde{\A}_\Seek \times \widetilde{\A}_\Hide$ be any NE of the game $\widetilde{\Gamma}$. Since $\widetilde{\Gamma}$ is a zero-sum game, then $(\smargv^{\Seek^\prime},\hmargv^{\Hide^*})$ is also a NE of $\widetilde{\Gamma}$.

Since $\smargv^{\Seek^\prime}$ is a best response to $\hmargv^{\Hide^*}$, it is an optimal solution to the continuous knapsack problem \eqref{knapsack_1}. The profits of each object are given by \eqref{Profits_Regime_1_H} and satisfy inequalities \eqref{Profits_Ineq_Regime_1_H}. Furthermore, $\rh < \tau_{i^*}$ implies that:
\begin{align*}
\frac{\rh - \sum_{j=1}^{i^{*}}c_{j} - \sum_{j=1}^{k_{i^{*}}}c_{\pi^{i^{*}}(j)} }{S^{i^{*}}_{k_{i^{*}}+1}} < p_{i^*+1}c_{i^*+1}.
\end{align*}
Since $|\J| < \rs \leq |\J| + |\K|$, then any best response to $\hmargv^{\Hide^*}$ must select all the objects in $\J$, must not select any object in $\I$, and must entirely fill the knapsack. Thus, $\smargv_i^{\Seek^\prime} = 0$ for every $i \in \I$, $\smargv_i^{\Seek^\prime} = 1$ for every $i \in \J$, and $\sum_{i=1}^n \smargv_i^{\Seek^\prime} = \rs$.

Next, since $\hmargv^{\Hide^*}$ is a best response to $\smargv^{\Seek^\prime}$, then it is an optimal solution to the continuous knapsack problem \eqref{knapsack_2}. Next, we write the dual of \eqref{knapsack_2} associated with $\smargv^{\Seek^\prime}$:
\begin{alignat}{3}
\OneAndAHalfSpacedXI
\min_{\alpha,\beta} \ & \rh \alpha + \sum_{i=1}^n c_i \beta_i\nonumber\\
\text{s.t.} \ & \alpha + \beta_i \geq 1 - p_i \smargv_i^{\Seek^\prime}, \quad && \forall \, i \in \lb 1,n\rb\label{knapsack_2_dual}\\
& \alpha \geq 0\nonumber\\
& \beta_i \geq 0, && \forall \, i \in \lb 1,n\rb.\nonumber
\end{alignat}
Let $(\alpha^*,\beta^*)$ be an optimal solution of the dual problem \eqref{knapsack_2_dual}. Since $0 < \hmargv^{\Hide^*}_i < c_i$ for every $i \in \K$, then by complementary slackness, $\beta_i^*= 0$ and $\smargv_i^{\Seek^\prime} =  (1-\alpha^*)/p_i$ for every $i \in \K$.
Since  $\smargv^{\Seek^\prime}$ must fill the knapsack \eqref{knapsack_1} entirely, then:
\begin{align*}
\rs = \sum_{i\in \I\cup\J\cup\K} \smargv_i^{\Seek^\prime} = k_{i^*} +(1- \alpha^*) S_{k_{i^*}+1}^{i^*}.
\end{align*}
Thus, for every $i \in \K$, $\smargv_i^{\Seek^\prime} =(\rs - k_{i^*})/(p_i S_{k_{i^*}+1}^{i^*})$. In conclusion, $\smargv^{\Seek^\prime}$ satisfies \eqref{Regime_1_S}.

Similarly, $(\smargv^{\Seek^*},\hmargv^{\Hide^\prime})$ is a NE of $\widetilde{\Gamma}$. Then, $\hmargv^{\Hide^\prime}$ is a best response to $\smargv^{\Seek^*}$ and is an optimal solution to the continuous knapsack problem \eqref{knapsack_2}. The profits of each object are given by \eqref{Profits_Regime_1_S} and satisfy inequalities \eqref{Profits_Ineq_Regime_1_S}. Since $\sum_{i \in \I\cup \J} c_i < \rh < m$, then any best response to $\smargv^{\Seek^*}$ must select all copies of the objects in $\I$ and $\J$. Therefore, $\hmargv_i^{\Hide^\prime} = c_i$ for every $i \in \I \cup \J$. Furthermore, $k_{i^*}+1 \notin K_{i^*}$ and the non-edge case assumptions imply that:
\begin{align}
1 - \frac{\rs - k_{i^*}}{S_{k_{i^*}+1}^{i^*}}\overset{\eqref{used_for_knapsack_3}}{>} 1 - p_{\pi^{i^*}(k_{i^*}+1)} \geq 0.\label{used_for_iff_1}
\end{align}
Therefore, $\hmargv^{\Hide^\prime}$ must entirely fill the knapsack and $\sum_{i=1}^n \hmargv_i^{\Hide^\prime} = \rh$. 

Next, since $\smargv^{\Seek^*}$ is a best response to $\hmargv^{\Hide^\prime}$, then it is an optimal solution to the continuous knapsack problem in \eqref{knapsack_1}. The dual of \eqref{knapsack_1} associated with $\hmargv^{\Hide^\prime}$ is given by:
\begin{alignat}{3}
\min_{\eta,\xi} \ & \rs \eta + \sum_{i=1}^n \xi_i\nonumber\\
\text{s.t.} \ & \eta + \xi_i \geq p_i \hmargv_i^{\Hide^\prime}, \quad && \forall \, i \in \lb 1,n\rb\label{knapsack_1_dual}\\
& \eta \geq 0,\nonumber\\
& \xi_i \geq 0, && \forall \, i \in \lb 1,n\rb.\nonumber
\end{alignat}
Let $(\eta^*,\xi^*)$ be an optimal solution of the dual problem \eqref{knapsack_1_dual}. Since \eqref{used_for_knapsack_4} and \eqref{used_for_iff_1} imply that $0 < \smargv^{\Seek^*}_i < 1$ for every $i \in \K$, then by complementary slackness, $\xi_i^*= 0$ and $\hmargv_i^{\Hide^\prime} =  \eta^*/p_i$ for every $i \in \K$.
%
Since  $\hmargv^{\Hide^\prime}$ must fill the knapsack \eqref{knapsack_2} entirely, then:
\begin{align*}
\rh = \sum_{i\in \I\cup\J\cup\K} \hmargv_i^{\Hide^\prime} = \sum_{j=1}^{i^*} c_j + \sum_{j=1}^{k_{i^*}} c_{\pi^{i^*}(j)}+\eta S_{k_{i^*}+1}^{i^*}.
\end{align*}
Therefore, $\hmargv^{\Hide^\prime}$ satisfies \eqref{Regime_1_H}.

\item[\emph{Regime Pattern 2:} $i^* = 0$ and $\tau_{-1} < \rh \leq \nu_{0}$.] We additionally consider the following non-edge case assumptions: $p_{\pi^0(\ell_0)}<p_{\pi^0(\ell_0+1)}< 1$ and $p_{\pi^0(\ell_0+1)}< p_{\pi^0(\ell_0+2)}$ if $\ell_0 \leq n-2$.


%
%

Let $\hmargv^{\Hide^*} \in \widetilde{\A}_\Hide$ satisfying \eqref{Regime_2_H} and let $\smargv^{\Seek^\prime}$ be an equilibrium strategy for S in $\widetilde{\Gamma}$. Since $\smargv^{\Seek^\prime}$ is a best response to $\hmargv^{\Hide^*}$, it is an optimal solution to the continuous knapsack problem \eqref{knapsack_1}. The profits of each object are given by \eqref{Profits_Regime_2_H}. Furthermore, \eqref{for_knapsack_regime_2_2} and the inequality $\rs > |\J| + 1$ imply that any best response to $\hmargv^{\Hide^*}$ must select all the objects in $\J \cup \{\pi^0(\ell_0+1)\}$. Therefore, $\smargv_i^{\Seek^\prime} = 1$ for every $i \in \J \cup \{\pi^0(\ell_0+1)\}$. 

Since $\hmargv^{\Hide^*}$ is a best response to $\smargv^{\Seek^\prime}$, then it is an optimal solution to the continuous knapsack problem \eqref{knapsack_2}. Since $0 \overset{\eqref{for_knapsack_regime_2_2}}{<} \hmargv^{\Hide^*}_{\pi^0(\ell_0+1)}$, then at optimality of the dual \eqref{knapsack_2_dual}, $\alpha^* = 1 - p_{\pi^0(\ell_0+1)} \smargv^{\Seek^\prime}_{\pi^0(\ell_0+1)} - \beta^*_{\pi^0(\ell_0+1)} = 1 - p_{\pi^0(\ell_0+1)} - \beta^*_{\pi^0(\ell_0+1)} \leq 1 - p_{\pi^0(\ell_0+1)}$.
%
%
%
%
Finally, for every $i \in \K\setminus \{\pi^0(\ell_0+1)\}$, $\hmargv^{\Hide^*}_{\pi^0(\ell_0+1)} = 0 < c_{\pi^0(\ell_0+1)}$. Thus, by complementary slackness, 
$\beta_i^* = 0$ and $\smargv_i^{\Seek^\prime} \geq  (1-\alpha^*)/p_i \geq p_{\pi^0(\ell_0+1)}/p_i$ for every $i \in \K\setminus \{\pi^0(\ell_0+1)\}$.
%
%
%
In conclusion, $\smargv^{\Seek^\prime}$ satisfies \eqref{Regime_2_S}.

Let $\hmargv^{\Hide^\prime}$ be an equilibrium strategy for H in $\widetilde{\Gamma}$, and consider $\smargv^{\Seek^*} \in \A_\Seek$ satisfying
 \begin{align*}
    \OneAndAHalfSpacedXI
        &\begin{aligned}
        \rho^{\Seek^*}_i &= 1 \quad &\text{if } i \in \J\cup\{ \pi^{0}(\ell_{0}+1) \},\\
        \frac{p_{\pi^{0}(\ell_{0}+1)}}{p_{i}} < \rho^{\Seek^*}_i & < 1  &\text{if } i \in \K\setminus \{ \pi^{0}(\ell_{0}+1) \}, \\
        \sum_{i=1}^{n} \rho^{\Seek^*}_i & < \rs. &\\
        \end{aligned}
        \end{align*}
Such a vector exists as a consequence of \eqref{for_knapsack_regime_2_1} and since $p_{\pi^{0}(\ell_{0}+1)} < p_{i}$ for every $i \in \K\setminus \{ \pi^{0}(\ell_{0}+1)\}$ under the non-edge case assumptions. Then, $\hmargv^{\Hide^\prime}$ is a best response to $\smargv^{\Seek^*}$ and is an optimal solution to the continuous knapsack problem \eqref{knapsack_2}. The profits of each object are given by: 
\begin{align*}
\OneAndAHalfSpacedXI
\forall \, i\in  \J\cup\{\pi^0(\ell_0+1)\}, \ 1-p_i\smargv_i^{\Seek^*} &=  1 - p_i\\
\forall \, i\in \K\setminus\{\pi^0(\ell_0+1)\}, \ 1-p_i\smargv_i^{\Seek^*} &<  1 -  p_{\pi^0(\ell_0+1)}.
\end{align*} 
Under the non-edge case assumptions, $1 - p_{\pi^0(1)} \geq \cdots \geq 1 - p_{\pi^0(\ell_0)} >  1 - p_{\pi^0(\ell_0+1)} > 0$. Since $ \sum_{j=1}^{\ell_0} c_{\pi^0(j)} < \rh < \sum_{j=1}^{\ell_0+1} c_{\pi^0(j)}$, then any best response to $\smargv^{\Seek^*}$ selects all copies of the objects in $\J$ and fills the remaining of the knapsack with objects in $\pi^0(\ell_0+1)$. Therefore $\hmargv^{\Hide^\prime}$  satisfies \eqref{Regime_2_H}.

\item[\emph{Regime Pattern 3:}  $i^* \geq 1$ and $\tau_{i^*-1} < \rh \leq \nu_{i^*}$.] We additionally consider the following non-edge case assumptions: $p_{i^*-1} c_{i^*-1} < p_{i^*} c_{i^*} < p_{i^*+1} c_{i^*+1}$, $\rh < \sum_{j=1}^{i^*} c_{j} + \sum_{j=1}^{\ell_{i^*}+1} c_{\pi^{i^*}(j)} +  p_{i^*}c_{i^*}S^{i^*}_{\ell_{i^*}+2}$, $\rs < k_{i^*-1}+1 + p_{\pi^{i^*-1}(k_{i^*-1}+1)}S^{i^*-1}_{k_{i^*-1}+2}$, $p_{\pi^{i^*}(\ell_{i^*})}<p_{\pi^{i^*}(\ell_{i^*}+1)}< 1$, and $p_{\pi^{i^*}(\ell_{i^*}+1)}< p_{\pi^{i^*}(\ell_{i^*}+2)}$ if  $\ell_{i^*} \leq n -i^*-2$.


Let $(\smargv^{\Seek^*},\hmargv^{\Hide^*}) \in \widetilde{\A}_\Seek \times \widetilde{\A}_\Hide$ satisfying \eqref{Regime_3_S}-\eqref{Regime_3_H} and let  $(\smargv^{\Seek^\prime},\hmargv^{\Hide^\prime}) \in \widetilde{\A}_\Seek \times \widetilde{\A}_\Hide$ be any NE of the game $\widetilde{\Gamma}$. 
Since $\smargv^{\Seek^\prime}$ is a best response to $\hmargv^{\Hide^*}$, it is an optimal solution to the continuous knapsack problem \eqref{knapsack_1}. The profits of each object are given by \eqref{Profits_Regime_3_H}. Under the non-edge case assumptions, $p_1 c_1 \leq \cdots \leq p_{i^*-1} c_{i^*-1} < p_{i^*} c_{i^*} < p_{i^*+1} c_{i^*+1} \leq \cdots \leq p_n c_n$. From \eqref{Profits_Ineq_Regime_3_H} and the fact that $|\J|+1<\rs \leq |\J| + |\K| +1$, we deduce that any best response to $\hmargv^{\Hide^*}$ selects all the objects in $\J \cup\{\pi^{i^*}(\ell_{i^*}+1)\}$, does not select any object in $\I \setminus \{i^*\}$, and fills the knapsack \eqref{knapsack_1} entirely. Therefore, $\smargv_i^{\Seek^\prime} = 0$ for every $i \in \I \setminus \{i^*\}$, $\smargv_i^{\Seek^\prime} = 1$ for every $i \in \J \cup\{\pi^{i^*}(\ell_{i^*}+1)\}$. and $\sum_{i=1}^n \smargv_i^{\Seek^\prime} = \rs$.

Since $\hmargv^{\Hide^*}$ is a best response to $\smargv^{\Seek^\prime}$, then it is an optimal solution to the continuous knapsack problem \eqref{knapsack_2}. Under the non-edge case assumptions, $0 < \hmargv_{\pi^{i^*}(\ell_{i^*}+1)}^{\Hide^*} < c_{\pi^{i^*}(\ell_{i^*}+1)}$. Thus, at optimality of the dual \eqref{knapsack_2_dual}, $\alpha^* = 1 - p_{\pi^{i^*}(\ell_{i^*}+1)} \smargv^{\Seek^\prime}_{\pi^{i^*}(\ell_{i^*}+1)} = 1 - p_{\pi^{i^*}(\ell_{i^*}+1)}$. Furthermore, since $0 < \hmargv_i^{\Hide^*} < c_i$  for every $i \in \K \setminus\{\pi^{i^*}(\ell_{i^*}+1)\}$ under the non-edge case assumptions, then $\smargv^{\Seek^\prime}_i = (1-\alpha^*)/p_i = p_{\pi^{i^*}(\ell_{i^*}+1)}/p_i$. Since $\smargv^{\Hide^\prime}$ must fill the knapsack \eqref{knapsack_1} entirely, then:
\begin{align*}
\smargv^{\Seek^\prime}_{i^*} = \rs - \ell_{i^*} - 1 - p_{\pi^{i^*}(\ell_{i^*}+1)} S_{\ell_{i^*}+2}^{i^*}  = \rs - \ell_{i^*} - p_{\pi^{i^*}(\ell_{i^*}+1)} S_{\ell_{i^*}+1}^{i^*}.
\end{align*}
Therefore, $\smargv^{\Seek^\prime}$ satisfies \eqref{Regime_3_S}.

Similarly, $(\smargv^{\Seek^*},\hmargv^{\Hide^\prime})$ is a NE of $\widetilde{\Gamma}$. Then, $\hmargv^{\Hide^\prime}$ is a best response to $\smargv^{\Seek^*}$ and is an optimal solution to the continuous knapsack problem \eqref{knapsack_2}. The profits of each object are given by \eqref{Profits_Regime_3_S}. Under the non-edge case assumptions, $1 - p_{\pi^{i^*}(1)} \geq \cdots \geq 1 - p_{\pi^{i^*}(\ell_{i^*})} > 1 - p_{\pi^{i^*}(\ell_{i^*}+1)} > 0$.
We next show that:
\begin{align}
 \rs - \ell_{i^*} - p_{\pi^{i^*}(\ell_{i^*}+1)} S^{i^*}_{\ell_{i^*} + 1} < \min\left\{ \frac{ p_{\pi^{i^*}(\ell_{i^*}+1)} }{  p_{i^*} },\,1 \right\}. \label{ineq2: technical inequality regime k_i > ell_i}
 \end{align}
 
 Let us assume that \eqref{ineq2: technical inequality regime k_i > ell_i} does not hold and let $j^* \in \lb 1, n-i^*+1\rb$ satisfying $\pi^{i^*-1}(j^*) = i^*$. If $\ell_{i^*}+1 \leq j^*-1$, then \eqref{ineq: technical inequality regime k_i > ell_i} implies that $\rs = \ell_{i^*} + 1 + p_{\pi^{i^*-1}(\ell_{i^*}+1)} S^{i^*-1}_{\ell_{i^*} + 2}$, which contradicts the non-edge case assumption.
 If on the other hand $\ell_{i^*}+1 \geq j^*$, then $j^* < \ell_{i^*}+2 \leq n - i^* +1$ and $\rs = \ell_{i^*} + 2 + p_{\pi^{i^*-1}(\ell_{i^*}+2)} S^{i^*-1}_{\ell_{i^*} + 3}$, which also contradicts the non-edge case assumptions.
Therefore, \eqref{ineq2: technical inequality regime k_i > ell_i} holds.

Since $\sum_{i \in \I\cup \J} c_i < \rh$, then any best response to $\smargv^{\Seek^*}$ must select all copies of the objects in $\I$ and $\J$, and must fill the knapsack \eqref{knapsack_2} entirely. Therefore, $\hmargv_i^{\Hide^\prime} = c_i$ for every $i \in \I \cup \J$ and $\sum_{i=1}^n \hmargv_i^{\Hide^\prime} = \rh$.
Since $\smargv^{\Seek^*}$ is a best response to $\hmargv^{\Hide^\prime}$, then it is an optimal solution to the continuous knapsack problem \eqref{knapsack_1}. Under the non-edge case assumptions, $0 < \smargv^{\Seek^*}_i < 1$ for every $i \in \{i^*\} \cup\K\setminus\{\pi^{i^*}(\ell_{i^*}+1)\}$. Therefore, at optimality of the dual \eqref{knapsack_1_dual}, $\eta^* = p_{i^*} \hmargv^{\Hide^\prime}_{i^*} = p_{i^*} c_{i^*}$, and $\hmargv^{\Hide^\prime}_{i^*} = \eta^*/p_{i} = p_{i^*}c_{i^*}/p_i$ for every $i \in \K\setminus\{\pi^{i^*}(\ell_{i^*}+1)\}$. Finally, since $\hmargv^{\Hide^\prime}$ fills the knapsack \eqref{knapsack_2} entirely, then:
\begin{align*}
\hmargv^{\Hide^\prime}_{\pi^{i^*}(\ell_{i^*}+1)} = \rh - \sum_{j=1}^{i^*} c_j - \sum_{j=1}^{\ell_{i^*}} c_{\pi^{i^*}(j)}- p_{i^*}c_{i^*} S_{\ell_{i^*}+2}^{i^*}.
\end{align*}
In conclusion, $\hmargv^{\Hide^\prime}$ satisfies \eqref{Regime_3_H}.\hfill\Halmos

\end{enumerate}

\end{proof}


\begin{proof}{\emph{Proof of Proposition \ref{prop:iff}.}}
In this proof, we allow the vector of capacities $c$ and the players’ resources $\rs$ and $\rh$ to be continuous in the game $\widetilde{\Gamma}$. Let $\Psi$ be the set of parameters given by \eqref{set_Psi} for which $\widetilde{\Gamma}$ is nontrivial.
First, we note that 
$$\Psi^{\prime}\coloneqq \left\{(n,p,c,\rs,\rh) \in \Psi \, : \, p_i < 1 \ \forall \, i \in \lb1,n\rb, \ p_i \neq p_j \text{ and } p_i c_i \neq p_j c_j \ \forall \, i \neq j \in \lb1,n\rb\right\}$$ 
is a dense subset of $\Psi$. Next, we consider an instantiation of the game parameters $(n,p,c,\rs,\rh) \in \Psi^{\prime}$. We order the indices such that $p_i c_i < p_{i+1} c_{i+1}$ for every $i \in \lb 1,n-1\rb$. Let $i^* \in \lb 0,n-1\rb$ such that $\tau_{i^{*}-1} < \rh \leq \tau_{i^{*}}$.

We first consider the case of Regime Pattern 1, i.e., $\nu_{i^*} < \rh \leq \tau_{i^*}$. We then consider new player resources $\hrs = \rs - \varepsilon$ and $\hrh = \rh - \varepsilon$ for $\varepsilon > 0$ arbitrarily small. To avoid confusion, we denote the corresponding parameters that depend on $\hrs$ and $\hrh$ as $\hat{k}_i$, $\hat{\ell}_i$, $\hat{\tau}_i$, $\hat{\nu}_i$, and $\hat{i}^*$.

By definition of $k_{i^*}$, and for arbitrarily small $\varepsilon$, we obtain:
\begin{align*}
k_{i^*} + p_{\pi^{i^*}(k_{i^*})} S_{k_{i^*}+1}^{i^*}  <  \hrs < \rs &\leq k_{i^*} + 1 + p_{\pi^{i^*}(k_{i^*}+1)} S_{k_{i^*}+2}^{i^*}.
\end{align*}
Thus, $\hat{k}_{i^*} = k_{i^*}$. This implies that $\hat{\tau}_{i^*} = \tau_{i^*}$ and $\hat{\nu}_{i^*} = \nu_{i^*}$. 
We then deduce the following inequalities for arbitrarily small $\varepsilon$: $\hat{\nu}_{i^*}=\nu_{i^*} < \hrh < \rh  \leq \tau_{i^*} =\hat{\tau}_{i^*}.$ Thus, $\hat{i}^* = i^*$.

Since $\nu_{i^*}< \hrh <{\tau}_{i^*}$ and $\hrs < k_{i^*}+1 + p_{\pi^{i^*}(k_{i^*}+1)}S_{k_{i^*}+2}^{i^*}$, then Proposition \ref{prop:iff_conditions} implies that all pure NE of the game $\widetilde{\Gamma}$ with the parameters $(n,p,c,\hrs,\hrh)$ for arbitrarily small $\varepsilon > 0$ satisfy the corresponding equilibrium conditions \eqref{Regime_1_S}-\eqref{Regime_1_H}. Furthermore $(n,p,c,\hrs,\hrh)$ is arbitrarily close to $(n,p,c,\rs,\rh)$.

We next consider the case of Regime Pattern 2, i.e., $i^* = 0$ and $\rh \leq \nu_0$. Proposition \ref{prop:iff_conditions} implies that all pure NE of the game $\widetilde{\Gamma}$ with the parameters $(n,p,c,\hrs,\hrh) \in \Psi^{\prime}$ satisfy the corresponding equilibrium conditions \eqref{Regime_2_S}-\eqref{Regime_2_H}.

Finally, we consider the case of Regime Pattern 3, i.e., $i^* \geq 1$ and $\tau_{i^*-1} < \rh \leq \nu_{i^*}$. We then consider new player resources $\hrs = \rs - \varepsilon$ and $\hrh = \rh - \varepsilon$ for $\varepsilon > 0$ arbitrarily small. Similarly, we denote the corresponding auxiliary parameters as $\hat{k}_i$, $\hat{\ell}_i$, $\hat{\tau}_i$, $\hat{\nu}_i$, and $\hat{i}^*$.

Using a similar derivation as above, we deduce that for arbitrarily small $\varepsilon$, $\hat{k}_{i^*} = k_{i^*}$ and $\hat{\nu}_{i^*} = \nu_{i^*}$. Then, by definition of $k_{i^*-1}$, we obtain:
\begin{align*}
k_{i^*-1} + p_{\pi^{i^*-1}(k_{i^*-1})} S_{k_{i^*-1}+1}^{i^*-1}  <  \hrs < \rs &\leq k_{i^*-1} + 1 + p_{\pi^{i^*-1}(k_{i^*-1}+1)} S_{k_{i^*-1}+2}^{i^*-1}.
\end{align*}

Thus, $\hat{k}_{i^*-1} = k_{i^*-1}$ and $\hat{\tau}_{i^*-1} = \tau_{i^*-1}$. Then, we obtain that $\hat{i}^* = i^*$ since $\hat{\tau}_{i^*-1}=\tau_{i^*-1} < \hrh < \rh  \leq \nu_{i^*} =\hat{\nu}_{i^*}$.
Finally, by definition of $\ell_{i^*}$, we obtain:
\begin{align*}
\sum_{j=1}^{i^*} c_{j} + \sum_{j=1}^{\ell_{i^*}} c_{\pi^{i^*}(j)} +  p_{i^*}c_{i^*}S^{i^*}_{\ell_{i^*}+1} <  \hrh < \rh &\leq \sum_{j=1}^{i^*} c_{j} + \sum_{j=1}^{\ell_{i^*}+1} c_{\pi^{i^*}(j)} +  p_{i^*}c_{i^*}S^{i^*}_{\ell_{i^*}+2}.
\end{align*}

Thus, $\hat{\ell}_{i^*} = \ell_{i^*}$. Since $\rh < \sum_{j=1}^{i^*} c_{j} + \sum_{j=1}^{\ell_{i^*}+1} c_{\pi^{i^*}(j)} +  p_{i^*}c_{i^*}S^{i^*}_{\ell_{i^*}+2}$ and $\rs < k_{i^*-1}+1 + p_{\pi^{i^*-1}(k_{i^*-1}+1)}S^{i^*-1}_{k_{i^*-1}+2}$, then Proposition \ref{prop:iff_conditions} implies that all pure NE of the game $\widetilde{\Gamma}$ with the parameters $(n,p,c,\hrs,\hrh)$ for arbitrarily small $\varepsilon > 0$ satisfy the corresponding equilibrium conditions \eqref{Regime_3_S}-\eqref{Regime_3_H}. Furthermore $(n,p,c,\hrs,\hrh)$ is arbitrarily close to $(n,p,c,\rs,\rh)$.\hfill\Halmos
%
%
%
\end{proof}

\section{Proofs of Section \ref{sec: equilibrium computation}}

Before proving Theorem \ref{thm:algorithm}, we show that Algorithm \ref{alg: construction of prob distribution from marginals} is well defined and terminates. We denote as $\kappa^* \in \Z_{\geq 0}\cup\{+\infty\}$ the number of iterations of the while loop (\ref{alg:while_start}-\ref{alg:while_end}) in Algorithm~\ref{alg: construction of prob distribution from marginals}.

\begin{proposition}
\label{prop:alg_well_defined}
Each iteration of Algorithm \ref{alg: construction of prob distribution from marginals} is well defined. In particular,
\begin{align}
\forall \, k \in \lb 1,\kappa^*+1\rb,& \ \bar{\gmargv}^k \in [0,1]^n \  \text{ and } \ \sum_{i=1}^n \bar{\gmargv}_i^k \leq \bar{\gres}, \label{alg_eq:feasibility1}\\
\forall \, k \in \lb 1,\kappa^*\rb,& \ q^k \in \lb 1,n\rb \  \text{ and } \ \delta^k \in [0,1). \label{alg_eq:feasibility2}
\end{align}
\end{proposition}

\begin{proof}{\emph{Proof of Proposition \ref{prop:alg_well_defined}.}} We show \eqref{alg_eq:feasibility1} and \eqref{alg_eq:feasibility2} by induction. We first consider $k=1$. By construction, $\bar{\gmargv}^1 = \gmargv - \floor*{\gmargv} \in [0,1]^n$. Furthermore, by definition of $\widetilde{\A}(\gcap,\gres)$, we obtain:
\begin{align}
\bar{\gres} = \gres - \sum_{i=1}^n \floor*{\gmargv_i} = \gres - \sum_{i=1}^n \gmargv_i +  \sum_{i=1}^n\bar{\gmargv}^1_i \geq  \sum_{i=1}^n\bar{\gmargv}^1_i. \label{maintain_feasibility}
\end{align}

Next, $q^1$ is constructed when the algorithm initiates the while loop (\ref{alg:while_start}-\ref{alg:while_end}), that is, when $\bar{\gmargv}^1 \notin \{0,1\}^n$. Since $\bar{\gmargv}^1 \geq \boldsymbol{0}_n$ and $\bar{r} \in \Z$, then $1 \leq \left| \left\{ i \in \lb 1,n\rb:\, \bar{\gmargv}^1_i > 0   \right\} \right| \leq n$ and $1 \leq \lceil \sum_{i=1}^n\bar{\gmargv}^1_i  \rceil \overset{\eqref{maintain_feasibility}}{\leq} \bar{\gres}$. Therefore, $1 \leq q^1 \leq n$.
Finally, $\delta^1$ is well defined since $q^1 \in \lb 1,n\rb$, and $\delta^1 \in [0,1]$ as a consequence of $\bar{\gmargv}^1 \in [0,1]^n$. We next show by contradiction that $\delta^1 < 1$. Indeed, if $\delta^1 = 1$, then we first deduce that for every $j \in \lb 1,q^1\rb$, $1 \geq \bar{\gmargv}^{1}_{\theta^1(j)} \geq \bar{\gmargv}^{1}_{\theta^1(q^1)} \geq \delta^1 = 1$. If $q^1 = n$, then this contradicts $\bar{\gmargv}^1 \notin \{0,1\}^n$. If $q^1 < n$, then, we derive the following inequalities:
%
%
%
%
\begin{align}
q^1 = \sum_{j=1}^{q^1} \bar{\gmargv}^1_{\theta^1(j)} \leq \sum_{i=1}^n \bar{\gmargv}^{1}_i \overset{\eqref{maintain_feasibility}}{\leq} \bar{\gres},\label{disjunction_1}\\
q^1 \leq \left| \left\{ i \in \lb 1,n\rb:\, \bar{\gmargv}^1_i =1   \right\} \right| \leq \left| \left\{ i \in \lb 1,n\rb:\, \bar{\gmargv}^1_i > 0   \right\} \right|.\label{disjunction_2}
\end{align}
If $q^1 = \bar{\gres}$, (resp. $q^1 = \left| \left\{ i \in \lb 1,n\rb:\, \bar{\gmargv}^1_i > 0   \right\} \right|$) then \eqref{disjunction_1} (resp. \eqref{disjunction_2}) implies that $ \bar{\gmargv}^{1}_{\theta^1(j)} = 0$ for every $j \in \lb q^1+1,n\rb$. This also contradicts $\bar{\gmargv}^1 \notin \{0,1\}^n$. Thus, $\delta^1 < 1$.


Next, we assume that  \eqref{alg_eq:feasibility1} and  \eqref{alg_eq:feasibility2} hold for $k \in \lb 1,\kappa^*\rb$. Since $\delta^k < 1$, then we obtain:
\begin{align*}
&\forall \, j \in \lb 1,q^k\rb,  \ 0 \leq \frac{\bar{\gmargv}^{k}_{\theta^k(j)} - \bar{\gmargv}^{k}_{\theta^k(q^k)}}{1-\delta^k} \leq \frac{\bar{\gmargv}^{k}_{\theta^k(j)} - \delta^k}{1-\delta^k} = \bar{\gmargv}^{k+1}_{\theta^k(j)}  \leq \frac{1- \delta^k}{1-\delta^k} = 1,\\
\text{and if $q^k < n$, then } &\forall j \in \lb q^k+1,n\rb,  \ 0 \leq  \frac{\bar{\gmargv}^{k}_{\theta^k(j)}}{1-\delta^k} = \bar{\gmargv}^{k+1}_{\theta^k(j)}  \leq \frac{\bar{\gmargv}^{k}_{\theta^k(j)}}{\bar{\gmargv}^{k}_{\theta^k(q^k+1)}} \leq 1.
\end{align*}


Therefore, for every $i \in \lb 1,n\rb$, $\bar{\gmargv}^{k+1}_i \in [0,1]$. Next, we show that $\sum_{i=1}^n \bar{\gmargv}_i^{k+1} \leq \bar{\gres}$: 
%
%
\begin{align*}
\sum_{i=1}^n \bar{\gmargv}^{k+1}_i &=  \sum_{j=1}^{q^k} \frac{\bar{\gmargv}^{k}_{\theta^k(j)} - \delta^k}{1-\delta^k} + \sum_{j=q^k+1}^n  \frac{\bar{\gmargv}^{k}_{\theta^k(j)}}{1-\delta^k} = \frac{1}{1-\delta^k}\left(\sum_{i=1}^n \bar{\gmargv}^{k}_i - q^k\delta^k\right).
\end{align*}
If $q^k = \bar{\gres}$, then:
\begin{align*}
\sum_{i=1}^n \bar{\gmargv}^{k+1}_i & \leq \frac{1}{1-\delta^k}\left(\bar{\gres} - \bar{\gres}\delta^k\right) = \bar{\gres}.
\end{align*}
If on the other hand $q^k = \left| \left\{ i \in \lb 1,n\rb:\, \bar{\gmargv}^k_i > 0   \right\} \right|$, then $\bar{\gmargv}^k \in [0,1]^n$ implies that:
\begin{align*}
\sum_{i=1}^n \bar{\gmargv}^{k+1}_i & = \frac{1}{1-\delta^k}\left(\sum_{i=1}^n \bar{\gmargv}^{k}_i - \left| \left\{ i \in \lb 1,n\rb:\, \bar{\gmargv}^k_i > 0   \right\} \right|\delta^k\right) \leq \frac{1}{1-\delta^k}\left(\sum_{i=1}^n \bar{\gmargv}^{k}_i - \delta^k\sum_{\{i \in\lb 1,n\rb \, : \, \bar{\gmargv}^{k}_i > 0\}}\bar{\gmargv}^{k}_i\right) \\
&= \sum_{i=1}^n \bar{\gmargv}^{k}_i \leq \bar{\gres}.
\end{align*}
Therefore, $\sum_{i=1}^n \bar{\gmargv}^{k+1}_i \leq \bar{\gres}$. 
Since $\bar{\gmargv}^{k+1} \in  [0,1]^n$, then the same argument as the one derived for $k=1$ can be applied to conclude that if $k < \kappa^*$ and $\bar{\gmargv}^{k+1} \notin \{0,1\}^n$, then $q^{k+1} \in \lb 1,n\rb$ and $\delta^{k+1}  \in [0,1)$. In conclusion, \eqref{alg_eq:feasibility1} and \eqref{alg_eq:feasibility2} hold by induction.
\hfill\Halmos
\end{proof}

\begin{proposition}
\label{prop:alg_terminates}
Algorithm \ref{alg: construction of prob distribution from marginals} terminates after $\kappa^* \leq n$ iterations of the while loop (\ref{alg:while_start}-\ref{alg:while_end}). In particular, for every $k \in \lb1,\kappa^*\rb$, $\delta^k > 0$, and
\begin{align*}
\left| \left\{ i \in \lb 1,n\rb:\, \bar{\gmargv}^{k+1}_i \in \{0,1\}   \right\} \right| > \left| \left\{ i \in \lb 1,n\rb:\, \bar{\gmargv}^{k}_i \in \{0,1\}   \right\} \right|.
\end{align*}
\end{proposition}

\begin{proof}{\emph{Proof of Proposition \ref{prop:alg_terminates}.}} Let $k \in \lb 1,\kappa^*\rb$. First, we show that $\delta^k > 0$. Since $q^k \leq  \left| \left\{ i \in \lb 1,n\rb:\, \bar{\gmargv}^k_i > 0   \right\} \right|$, then $\bar{\gmargv}^k_{\theta^k(q^k)} > 0$. Next, we show by contradiction that if $q^k< n$, then $\bar{\gmargv}^k_{\theta^k(q^k+1)}<1$: If instead $q^k< n$ and $\bar{\gmargv}^k_{\theta^k(q^k+1)}=1$, then we first deduce that $q^k < \left| \left\{ i \in \lb 1,n\rb:\, \bar{\gmargv}^k_i > 0   \right\} \right|$. Furthermore,
\begin{align*}
\bar{\gres} \overset{\eqref{alg_eq:feasibility1}}{\geq} \sum_{i=1}^n \bar{\gmargv}_i^k \geq \sum_{j=1}^{q^k+1}\bar{\gmargv}_{\theta^k(j)}^k = q^k + 1 > q^k.
\end{align*}
This contradicts the definition of $q^k$. Therefore if $q^k < n$, then $\bar{\gmargv}^k_{\theta^k(q^k+1)}<1$, which in turn implies that $\delta^k > 0$.

We now show that $\left| \left\{ i \in \lb 1,n\rb:\, \bar{\gmargv}^{k+1}_i \in \{0,1\}   \right\} \right| > \left| \left\{ i \in \lb 1,n\rb:\, \bar{\gmargv}^{k}_i \in \{0,1\}   \right\} \right|$. Let $j^\prime \in \lb 1,n\rb$ be such that $\bar{\gmargv}^k_{\theta^k(j^\prime)} = 0$. Necessarily, $j^\prime > \left| \left\{ i \in \lb 1,n\rb:\, \bar{\gmargv}^{k}_i >0   \right\} \right| \geq q^k$, which implies that $\bar{\gmargv}^{k+1}_{\theta^k(j^\prime)} = \bar{\gmargv}^k_{\theta^k(j^\prime)}/(1-\delta^k) = 0.$

Next, we consider $j^\prime \in \lb 1,n\rb$ such that $\bar{\gmargv}^k_{\theta^k(j^\prime)} = 1$. Then, $j^\prime \leq q^k$, as implied by the following inequalities:
\begin{align*}
&j^\prime = \sum_{j=1}^{j^\prime} \bar{\gmargv}^k_{\theta^k(j)} \leq \sum_{i=1}^n \bar{\gmargv}^{k}_i \overset{\eqref{alg_eq:feasibility1}}{\leq} \bar{\gres},\\
&j^\prime \leq \left| \left\{ i \in \lb 1,n\rb:\, \bar{\gmargv}^k_i =1   \right\} \right| \leq \left| \left\{ i \in \lb 1,n\rb:\, \bar{\gmargv}^k_i > 0   \right\} \right|.
\end{align*}
Thus, $j^\prime \leq q^k$ and $\bar{\gmargv}^{k+1}_{\theta^k(j^\prime)} = (\bar{\gmargv}^k_{\theta^k(j^\prime)} - \delta^k)/(1-\delta^k) = 1.$ This shows that $\left| \left\{ i \in \lb 1,n\rb:\, \bar{\gmargv}^{k+1}_i \in \{0,1\}   \right\} \right| \geq \left| \left\{ i \in \lb 1,n\rb:\, \bar{\gmargv}^{k}_i \in \{0,1\}   \right\} \right|$. To ensure a strict inequality, we must show that one fractional component of $\bar{\gmargv}^k$ becomes 0 or 1 in $\bar{\gmargv}^{k+1}$.

We know that $0 < \delta^k < 1$. If $\delta^k =\bar{\gmargv}^k_{\theta^k(q^k)}$, then $\bar{\gmargv}^{k+1}_{\theta^k(q^k)} = 0$. If $q^k < n$ and $\delta^k = 1 - \bar{\gmargv}^k_{\theta^k(q^k+1)}$, then $\bar{\gmargv}^{k+1}_{\theta^k(q^k+1)} = 1$. In both cases, a fractional component of $\bar{\gmargv}^k$ becomes 0 or 1 in $\bar{\gmargv}^{k+1}$. In conclusion $\left| \left\{ i \in \lb 1,n\rb:\, \bar{\gmargv}^{k+1}_i \in \{0,1\}   \right\} \right| > \left| \left\{ i \in \lb 1,n\rb:\, \bar{\gmargv}^{k}_i \in \{0,1\}   \right\} \right|$.

Since for every $k \in \lb 1,\kappa^*+1\rb$,  $\left| \left\{ i \in \lb 1,n\rb:\, \bar{\gmargv}^{k}_i \in \{0,1\}   \right\} \right| \leq n$, then the algorithm must terminate after at most $n$ iterations of the while loop (\ref{alg:while_start}-\ref{alg:while_end}). Therefore $\kappa^* \leq n$.\hfill \Halmos


\end{proof}

Now that we proved that Algorithm \ref{alg: construction of prob distribution from marginals} is well defined and terminates, we can show Theorem \ref{thm:algorithm}.


\begin{proof}{\emph{Proof of Theorem \ref{thm:algorithm}.}}
Consider a vector of capacities  $\gcap \in \Z^n_{>0}$, a budget of resources $\gres \in \Z_{>0}$, and a vector $\gmargv \in \widetilde{\A}(\gcap,\gres)$. For convenience, we denote $e^{\kappa^*+1} \coloneqq \bar{\gmargv}^{\kappa^*+1} \in \{0,1\}^n$. We first show that for every $k \in \lb 1,\kappa^*+1\rb$, $\floor*{\gmargv}+e^k \in \A(\gcap,\gres)$.


In Proposition \ref{prop:alg_terminates}, we showed that for every $k \in \lb 1,\kappa^*\rb$, if a component $i \in \lb 1,n\rb$ satisfies $\bar{\gmargv}_i^{k} = 0$, then $\bar{\gmargv}_i^{k+1} = 0$. Thus, for every $k \in \lb1,\kappa^*+1\rb$, $\bar{\gmargv}_i^{k} > 0$ only if $\bar{\gmargv}_i^{1} > 0$. By definition of $\widetilde{\A}(\gcap,\gres)$ and since $\gcap_i \in \Z$, we deduce that if $\bar{\gmargv}_i^1 >0$, then $\gcap_i \geq \lceil \gmargv_i\rceil = \floor*{\gmargv_i} + \ceil*{\bar{\gmargv}_i^1} = \floor*{\gmargv_i}  + 1$. 
Furthermore,
\begin{align*}
\forall \, k \in \lb 1,\kappa^*\rb, \ \sum_{i=1}^n (\floor*{\gmargv_i} + e_i^k) = q^k + \sum_{i=1}^n \floor*{\gmargv_i}  \leq  \bar{\gres} + \sum_{i=1}^n \floor*{\gmargv_i} &= \gres,\\
\text{and } \sum_{i=1}^n (\floor*{\gmargv_i} + e_i^{\kappa^*+1}) = \sum_{i=1}^n\bar{\gmargv}_i^{\kappa^*+1} + \sum_{i=1}^n \floor*{\gmargv_i}  \overset{\eqref{alg_eq:feasibility1}}{\leq}  \bar{\gres} + \sum_{i=1}^n \floor*{\gmargv_i} &= \gres.
\end{align*}
Therefore, for every $k \in \lb 1,\kappa^*+1\rb$, $\floor*{\gmargv}+e^k \in \A(\gcap,\gres)$.

Next, we show that $\sigma$ returned by the algorithm is a probability distribution. We first note that for every $k \in \lb 1,\kappa^*+1\rb$, $\gamma^k \geq 0$. Furthermore, 
\begin{align*}
\sum_{\gpure \in \A(\gcap,\gres)}\sigma_{\gpure} = \gamma^{\kappa^*+1} + \sum_{k=1}^{\kappa^*} \gamma^k \delta^k = \gamma^{\kappa^*+1} + \sum_{k=1}^{\kappa^*}(\gamma^k - \gamma^{k+1})  = \gamma^{\kappa^*+1} + \gamma^1 - \gamma^{\kappa^*+1}  = 1.
\end{align*}
Therefore, $\sigma \in \Delta(\gcap,\gres)$.
We now show that $\sigma$ returned by the algorithm is consistent with the vector $\gmargv$. To this end, we note the following equality:
\begin{align*}
\forall \, k \in \lb 1,\kappa^*\rb, \ \forall \, i \in \lb 1,n\rb, \ \gamma^k \bar{\gmargv}_i^k - \gamma^{k+1} \bar{\gmargv}_i^{k+1} = \gamma^k(\bar{\gmargv}_i^k - (1 - \delta^k)\bar{\gmargv}_i^{k+1}) = \gamma^k\delta^k e_i^k.
\end{align*}
Then, we obtain:
\begin{align*}
\forall \, i \in \lb 1,n\rb, \ \mathbb{E}_{\gpure \sim \sigma}[\gpure_i] &= \gamma^{\kappa^*+1}(\floor*{\gmargv_i} + \bar{\gmargv}^{\kappa^*+1}_i) + \sum_{k=1}^{\kappa^*} \gamma^k \delta^k (\floor*{\gmargv_i}+ e_i^k) \\
&= \floor*{\gmargv_i}\sum_{\gpure \in \A(\gcap,\gres)}\sigma_\gpure + \gamma^{\kappa^*+1}\bar{\gmargv}_i^{\kappa^*+1}+ \gamma^1\bar{\gmargv}_i^1 - \gamma^{\kappa^*+1}\bar{\gmargv}_i^{\kappa^*+1}\\
& = \floor*{\gmargv_i} + \bar{\gmargv}_i^1 = \gmargv_i.
\end{align*}
Thus, $\sigma$ returned by Algorithm \ref{alg: construction of prob distribution from marginals} is consistent with the vector $\gmargv$.

Since $\kappa^* \leq n$, then the support of $\sigma$ is of size at most $n+1$. Finally, we argue that Algorithm \ref{alg: construction of prob distribution from marginals} runs in time $O(n^2)$. Indeed, the first iteration of the while loop (\ref{alg:while_start}-\ref{alg:while_end}) can be implemented in time $O(n \log n)$ by using an efficient sorting algorithm (e.g. merge sort) to sort $\bar{\gmargv}^1$ and create the permutation $\theta^1$. Fortunately, the subsequent iterations can be implemented in time $O(n)$. Indeed, we note that for every $k \in \lb 1,\kappa^*-1\rb$, $\bar{\gmargv}^{k+1}_{\theta^k(1)} \geq \cdots \geq \bar{\gmargv}^{k+1}_{\theta^k(q^k)}$ and $\bar{\gmargv}^{k+1}_{\theta^k(q^k+1)} \geq \cdots \geq  \bar{\gmargv}^{k+1}_{\theta^k(n)}$. Therefore, we can sort $\bar{\gmargv}^{k+1}$ and create the permutation $\theta^{k+1}$ by merging and sorting the lists $(\bar{\gmargv}^{k+1}_{\theta^k(1)},\dots,\bar{\gmargv}^{k+1}_{\theta^k(q^k)})$ and $(\bar{\gmargv}^{k+1}_{\theta^k(q^k+1)},\dots,\bar{\gmargv}^{k+1}_{\theta^k(n)})$ that are already sorted. This operation can be carried out in time $O(n)$. Since the number of iterations of the while loop (\ref{alg:while_start}-\ref{alg:while_end})  is upper bounded by $n$, then the overall running time of Algorithm \ref{alg: construction of prob distribution from marginals} is $O(n^2)$. 
%
%
%
%
\hfill \Halmos
\end{proof}






%
%
%



\end{document}